\documentclass[a4paper]{article}

\usepackage[utf8]{inputenc}
\usepackage[T1]{fontenc}
\usepackage{textcomp}
\usepackage[british]{babel}
\usepackage[left=1in,right=1in,top=1in,bottom=0.5in,includefoot=true]{geometry}
\usepackage{amsmath, amsfonts, amsthm, amssymb}
\usepackage[shortlabels]{enumitem}
\usepackage{mathrsfs}
\usepackage{bbm}
\usepackage{mathabx}
\usepackage{accents}
\usepackage[backend=biber,doi=false,isbn=false,maxbibnames=4]{biblatex}
\usepackage{csquotes}
\usepackage{microtype}
\usepackage{xcolor}
\usepackage{booktabs}
\usepackage{graphicx}
\usepackage[labelformat=simple]{subcaption}
\usepackage{zref-clever}
\usepackage{hyperref}
\usepackage{todonotes}
\usepackage[normalem]{ulem}

\newcommand{\cref}[1]{\zcref{#1}}
\newcommand{\Cref}[1]{\zcref[S]{#1}}
\zcsetup{cap,nameinlink=false}
\zcRefTypeSetup{equation}{abbrev}
\AddToHook{env/definition/begin}{%
	\zcsetup{countertype={definition=definition}}}
\AddToHook{env/assumption/begin}{%
	\zcsetup{countertype={definition=assumption}}}
\AddToHook{env/example/begin}{%
	\zcsetup{countertype={definition=example}}}
\AddToHook{env/theorem/begin}{%
	\zcsetup{countertype={definition=theorem}}}
\AddToHook{env/corollary/begin}{%
	\zcsetup{countertype={definition=corollary}}}
\AddToHook{env/proposition/begin}{%
	\zcsetup{countertype={definition=proposition}}}
\AddToHook{env/lemma/begin}{%
	\zcsetup{countertype={definition=lemma}}}
\AddToHook{env/remark/begin}{%
	\zcsetup{countertype={definition=remark}}}

\DeclareFieldFormat{eprint:ssrn}{%
  SSRN\addcolon\space
  \ifhyperref
    {\href{https://ssrn.com/abstract=#1}{\nolinkurl{#1}}}
    {\nolinkurl{#1}}%
}
\DeclareFieldAlias{eprint:SSRN}{eprint:ssrn}

\newcommand{\N}{\mathbb{N}}
\newcommand{\R}{\mathbb{R}}
\newcommand{\E}{\mathbb{E}}
\newcommand{\Q}{\mathbb{Q}}

\makeatletter \renewcommand\d[1]{\ensuremath{%
  \;\mathrm{d}#1\@ifnextchar\d{\!}{}}}
\makeatother

\newcommand{\condbar}{\;\middle|\;}
\newcommand\munderbar[1]{\underaccent{\bar}{#1}}
\newcommand{\one}{\textbf{1}}
\DeclareMathOperator{\diag}{diag}

\theoremstyle{definition}
\newtheorem{definition}{Definition}[section]
\newtheorem{assumption}[definition]{Assumption}

\theoremstyle{plain}
\newtheorem{theorem}[definition]{Theorem}
\newtheorem{corollary}[definition]{Corollary}
\newtheorem{proposition}[definition]{Proposition}
\newtheorem{lemma}[definition]{Lemma}
\theoremstyle{remark}
\newtheorem{remark}[definition]{Remark}

\zcRefTypeSetup{assumption}{
Name-sg = Assumption ,
name-sg = assumption ,
Name-pl = Assumptions ,
name-pl = assumptions ,
}

\numberwithin{table}{section}
\numberwithin{figure}{section}

\numberwithin{equation}{section}

\definecolor{linkcolor}{cmyk}{.2,.98,1.,.18}
\definecolor{citecolor}{cmyk}{.49,.0,.85,.38}
\hypersetup{colorlinks=true,linkcolor=linkcolor,citecolor=citecolor}

\allowdisplaybreaks[1]

\usepackage{authblk}

\title{Optimal Investment and Consumption in a Stochastic Factor Model}

\author[1]{Florian Gutekunst}
\author[1,2]{Martin Herdegen}
\author[1]{David Hobson}
\affil[1]{Department of Statistics, University of Warwick}
\affil[2]{Institut für Stochastik und Anwendungen, University of Stuttgart}

\date{2nd October 2026}

\makeatletter
\def\blfootnote{\xdef\@thefnmark{}\@footnotetext}
\makeatother

\begin{document}

		\maketitle
		\blfootnote{Florian Gutekunst gratefully acknowledges funding from the Statistics Centre for Doctoral Training at the University of Warwick.}

		\begin{abstract}
		        In this article, we study optimal investment and consumption in an incomplete stochastic factor model for a power utility investor on the infinite horizon.
		        When the state space of the stochastic factor is finite, we give a complete characterisation of the well-posedness of the problem, and provide an efficient numerical algorithm for computing the value function. 
		        When the state space is a (possibly infinite) open interval and the stochastic factor is represented by an Itô diffusion, we develop a general theory of sub- and supersolutions for second-order ordinary differential equations on open domains without boundary values to prove existence of the solution to the Hamilton-Jacobi-Bellman (HJB) equation along with explicit bounds for the solution.
		        By characterising the asymptotic behaviour of the solution, we are also able to provide rigorous verification arguments for various models, including -- for the first time -- the Heston model. 
		        Finally, we link the discrete and continuous setting and show that the value function in the diffusion setting can be approximated very efficiently through a fast discretisation scheme.
		\end{abstract}

		\section{Introduction}
		\label{section:introduction}

		In Merton's investment and consumption problem \cite{Merton1969,Merton1971}, an investor seeks to maximise their expected lifetime utility from consumption while investing in a stock and risk-free bond.
		The case of a Black-Scholes market, i.e., a model with constant coefficients, has been studied extensively, and it is straight-forward to derive a candidate value function for a power utility investor.
		  In this setting, multiple verification arguments have been given, see e.g.,  \textcite{Karatzas1986,Davis1990,Herdegen2021}.

		In a stochastic factor model, i.e., a model in which the market coefficients depend on some other stochastic process, the situation becomes much more complicated. 
        Most importantly, the market is now usually incomplete, and the candidate value function and optimal investment and consumption rates derived from the first-order condition depend on the solution of a semi-linear equation.
		Outside of the special cases of logarithmic utility and complete markets, no closed form solutions for the investment--consumption problem on the infinite horizon are known.

		The case where the stochastic factor is represented by a continuous-time Markov chain with finite state space has been studied by \textcite{Sotomayor2009}. 
        Under the key assumption that the problem is well-posed in each state individually, i.e., if the stochastic factor was frozen in that state, they prove a verification theorem as well as the existence of a solution to the Hamilton-Jacobi-Bellman (HJB) equation. However, this frozen-state well-posedness assumption excludes important examples, e.g., discretisations of the Vasicek or the Heston model.
        
		In the case in which the stochastic factor is represented by an Itô diffusion, many variants of the problem have been studied. 
        Over the finite horizon, \Textcite{Karatzas1991} and \textcite{Zariphopoulou2001} study the problem without consumption in an incomplete market, and \textcite{Wachter2002} studies the problem with consumption in a complete market.
		Over the infinite horizon, \textcite{Fleming2003} provide an existence and verification result in a very specific incomplete stochastic volatility model.
		\Textcite{Hata2012,Hata2012a} study existence and verification in a stochastic factor model, but under relatively strict assumptions on the coefficients.
        Most importantly, they assume uniform ellipticity of the differential operator, as well as global Lipschitz-continuity of the coefficients, which among other models rules out the Heston model.
        Moreover, both \textcite{Fleming2003} and \textcite{Hata2012,Hata2012a} work exclusively in a strongly myopically well-posed setting, i.e., a setting in which the myopic optimal consumption rates obtained by freezing the stochastic factor in each state are positive and uniformly bounded away from zero.

		One of the most notable contributions to the investment-consumption problem with a stochastic factor over the infinite horizon was made by \textcite{Guasoni2020}.
		Working under some broad assumptions, they use sub- and supersolutions to the HJB equation to prove existence of a positive solution to the HJB equation, together with upper and lower bounds on the solution.
        More precisely, their sub- and supersolutions are linked to optimal consumption rates in fictitious complete markets that are obtained from the original incomplete model, but with distorted dynamics.
        As in \cite{Sotomayor2009}, they require that for each state the problem is well-posed when the stochastic factor is frozen in that state. 
		While they also prove a verification theorem, the latter depends on a growth condition on the derivative of the solution that is difficult to check for concrete models.
        In particular, the results of \textcite{Guasoni2020} do not yield verification in the Heston model.
        
		In a follow-up paper, \textcite{Guasoni2019} focus exclusively on the Vasicek model. 
		Using a bespoke verification theorem, they are able to prove that the candidate value function is indeed optimal. 
        Whilst this is a significant achievement, the bespoke nature of the verification theorem means that it cannot be used for different stochastic factor models.
        
        Very recently, \textcite{Guasoni2025} approached the problem through a variational formulation. They show that the HJB equation can be interpreted as the Euler-Lagrange equation corresponding to an auxiliary non-linear convex functional.
        Using methods from calculus of variations and under the assumption that the model is strongly myopically well-posed and the stochastic factor is ergodic with a scale function diverging at the boundary of the state space, they are able to prove that the auxiliary functional has a unique minimiser, which is a classical solution to the HJB equation. Moreover, they provide a verification argument both in complete and incomplete markets (with constant correlation) and derive asymptotics of the minimiser and its logarithmic derivative in a complete market model.
        Furthermore, the variational formulation naturally leads to a numerical scheme for computing the value function by discretising the functional.
        \Textcite{Bayraktar2026} extend the variational approach to the investment-consumption problem under Epstein-Zin stochastic differential utility.
        
        \medskip{}

		The main contributions of our paper are as follows: First, in the case where the stochastic factor has a finite state space, we fully characterise the well-posedness of the problem.
        We show that the assumption that the problem is well-posed for every state individually is not necessary. 
        Indeed, even if there are states in which the problem would be ill-posed if the state was frozen, as long as the process spends little enough time in these “bad” states, the overall problem may still be well-posed.
		Such models arise naturally, for example as discretisations of a Heston model with risk aversion $R \in (0, 1)$, or of a Vasicek model. 
        Second, turning to the case where the stochastic factor is a diffusion, we develop a theory of sub- and supersolutions for second-order ordinary differential equations on the whole real line (and more generally any open interval).
        Here, the key difficulty is that there are no natural boundary values.
		Third, we use this theory to construct a candidate solution to the HJB equation for the investment-consumption problem with a stochastic factor given by a diffusion process and also to derive asymptotic estimates for the optimal consumption rate.
        Fourth, we use those asymptotic estimates to give a verification argument.
		Finally, we combine our discrete and continuous results to propose an efficient numerical scheme for computing the value function and optimal consumption rate in the diffusion setting.
		
		Describing our contributions in more detail, in the finite regime setting, we use the theory of $Z$- and $M$-matrices to show that the HJB equation has a unique positive solution if and only if a certain matrix, involving the $Q$-matrix of the Markov process describing the stochastic factor, the risk aversion parameter $R$ and the optimal consumption rates in the frozen state models, is a non-singular $M$-matrix. 
        Moreover, we show that the solution to the HJB equation can be efficiently computed using a fixed-point iteration.

		In the diffusion setting, like many papers in the extant literature \cite{Hata2012,Hata2012a,Guasoni2020,Guasoni2019}, we use sub- and supersolutions to prove the existence of a candidate solution. 
        However, unlike the approach in the extant literature, we do not construct a bespoke solution given sup- and supersolution for a specific equation at hand, but rather we develop a general theory of sub- and supersolutions for second-order problems on a (potentially unbounded) open domain without boundary values. 
        Our theory is an extension of the general theory of second-order boundary value problems on a bounded domain that goes back to \textcite{Nagumo1937}.
        We refer to \textcite{DeCoster2001} for a more recent survey paper.
        Our approach is motivated on the one hand by the observation that the theory of sub- and supersolutions for second order differential equations is also a very helpful tool for non-linear HJB equations arising from different diffusion-based infinite horizon control problems other than the consumption problem, and on the other hand, by a desire to find new solution methods for singular problems.
        Consider for example the homogeneous Dirichlet problem with a non-linear zero-order term that has a singularity at $0$.
		Here, the closed domain theory presented in \textcite{DeCoster2001} fails due to the lack of boundedness at the singularity.
		Compared to the boundary value theory on closed bounded domains, the open domain theory requires less regularity of the equation at values on the boundary of the admissible domain of the solution.
		If the boundary values are strictly enforced by the sub- and supersolution, the open domain theory can prove the existence of solutions to singular boundary value problems, even when the closed domain theory is not applicable.
		An example where this technique is applied can be found in \cref{section:bounded_domain} where we construct a solution to a class of HJB equations.

        It is insightful to compare our approach to that of \textcite{Guasoni2020} which in many ways is closest to our approach.
        As in \textcite{Guasoni2020}, our arguments to prove the existence of a global positive solution to the HJB equation when the stochastic factor is an Itô diffusion are based on sub- and supersolutions, but we construct them differently: 
        In \textcite{Guasoni2020}, sub- and supersolutions are constructed by either solving the HJB equation for a fictitious complete market in closed form, or abstractly through Lemma 3.1, where verifying the continuity (and even just the finiteness) of (3.4) and (3.5) is non-trivial in a general model. 
        By contrast, we use sub- and supersolutions that are proportional to the frozen consumption rate.
		While these are not as tight (and hence less usable as approximate consumption policies), verifying that they are in fact sub- and supersolutions is substantially easier than in either of the approaches in \textcite{Guasoni2020}.
		The second advantage of sub- and supersolutions that are proportional to the frozen consumption rate is that it is very easy to read off further properties like growth behaviour at infinity.
		The bounds resulting from the proportional sub- and supersolutions then allow us to characterise the asymptotic behaviour of solutions to the HJB equation as well as their logarithmic derivatives.
		This, in turn, enables us to verify that the candidate solution we constructed is indeed the optimal consumption rate.
		In particular, we are able to prove a full verification theorem for the infinite-horizon incomplete Heston model with consumption for the first time in the literature.
		Furthermore, using our techniques we are able to treat the Vasicek model in a semi-unified manner that does not require a bespoke verification theorem as in \cite{Guasoni2019}.

        It should be noted that the assumption of each state being well-posed if the stochastic factor was frozen in that state (which is assumed in \cite{Sotomayor2009,Hata2012,Guasoni2020}) is a major limitation of the current literature.
        There is no economic reason for making this assumption, and there are several models (including the Heston model with $R < 1$) where it is not satisfied.
        In the discrete setting, our characterisation of the well-posedness of the problem completely removes this limitation.
        In the diffusion setting, we take a first step toward getting rid of the assumption.
        Our asymptotic results still apply in models in which it does not hold, which makes dealing with such models easier.
        This can be seen for example in our existence and verification arguments in the Vasicek model (which are much simpler than the ones in \cite{Guasoni2019}), and in some parameter configurations of the Heston model.
        Still, there is no systematic way of dealing with models that fail to be strongly myopically well-posed; they require manual handling.
        
		As our final contribution, we propose a method for numerically computing the value function in models in which the stochastic factor is an Itô diffusion.\footnote{The code for our implementation of the scheme is available at \url{https://github.com/flogut/optimal-consumption}.}
		We discretise the stochastic factor into a Markov chain with a finite number of states, and then use the fixed point iteration for the finite regime setting to compute the value function of the discretisation.
		This procedure is very efficient and can handle a large number of states; computing the value function of a discretisation with 1\,000\,000 states takes one second on a regular laptop.
		We show that as the grid size of the discretisation becomes small (while keeping the minimal and maximal states constant), the discrete solutions approach the solution to the HJB equation under Neumann boundary conditions on a bounded domain, which corresponds to the HJB equation when the stochastic factor is a reflected diffusion.
		It should be noted that while a discrete model can never have correlation between the stochastic factor and the Brownian motion driving the risky asset, the equations for the optimal consumption rate in any model with constant correlation can be transformed into the equations for the optimal consumption rate for a modified model with zero correlation, and vice-versa.
		This way, our numerical scheme can handle arbitrary constant correlations.
        Note that the restriction to models with constant correlation only concerns the numerical scheme; the existence and asymptotics results for the HJB equation also hold when the correlation is non-constant.

		The rest of the paper is organised as follows: 
		\cref{section:setting} describes the problem setting. \cref{section:finite_state_space} considers factor processes with a finite state space. We give a necessary and sufficient condition for the existence of a solution to the HJB equation, and fully characterise the well-posedness of the optimal investment and consumption problem.
		We show that in simple models this yields an easily checkable condition for the well-posedness of the problem.
		\cref{section:diffusion} considers the optimal investment and consumption problem in a model in which the stochastic factor is given by an Itô diffusion. 
		We provide an easy way to generate explicit sub- and supersolutions for a large class of models. 
		We characterise the asymptotic behaviour of solutions to the HJB equation and their logarithmic derivatives.
		This then allows us to prove a verification theorem.
		In \cref{section:numerics}, we propose a numerical scheme for computing the value function and optimal consumption rate.
		We show that a fixed-point iteration can be used to efficiently solve the HJB equation in a model with a stochastic factor with finite state space. 
		Through discretisation, this yields a numerical procedure for stochastic factors that are Itô diffusions. 
        We show that the discretisations converge as the grid size becomes small.
        \cref{section:global_existence} extends the theory of sub- and supersolutions for one-dimensional second order boundary value problems on closed bounded domains to problems without boundary values on a (potentially unbounded) open domain.
		This theory drives our existence results. 
		Finally, in \cref{section:examples}, we apply our results to some examples and discuss their economic interpretation.
		In particular, we obtain (to the best of our knowledge) the first verification in a Heston model.
        
		In \cref{section:bounded_domain}, we show that the well-posedness criterion from the finite regime setting carries over to the continuous case and characterises the existence of solutions to the HJB equation on a bounded domain under homogeneous Neumann and Dirichlet boundary conditions.
        \Cref{section:auxiliary} collects some results of a more technical nature.

        \section{Problem Setting}
		\label{section:setting}

        We work on a filtered probability space $(\Omega, \mathcal{F}, \mathbb{F} = (\mathcal{F}_t)_{t\ge 0}, \mathbb{P})$, where the filtration $\mathbb{F} = (\mathcal{F}_t)_{t\ge 0}$ is assumed to be right-continuous and to satisfy the extension property that whenever $(\Q_t)_{t \geq 0}$ is a coherent family of probability measures on $(\Omega, \mathcal{F})$, i.e., $\Q_t$ and $\Q_s$ coincide on $\mathcal{F}_{t \wedge s}$ for all $t, s \geq 0$, then there exists a probability measure $\Q$ on $(\Omega, \mathcal{F})$ such that $\Q$ coincides with $\Q_t$ on $\mathcal{F}_t$ for all $t \geq 0$. Note that this implies that $\mathbb{F}$ cannot satisfy the completeness assumption made in the usual conditions. For this reason, we assume instead that $\mathbb{F}$ satisfies the so-called natural augmentation; see \textcite[Section 2]{Najnudel2011}. Then
         virtually all results needed for stochastic calculus carry over verbatim but the extension property is preserved; see \textcite[Section 3 and Corollary 4.10]{Najnudel2011}. Also note that the extension property is always satisfied if we work on the (naturally augmented) path spaces of continuous functions in $\R^d$ for some $d > 0$; see \textcite[Proposition 4.4]{Najnudel2011}. We assume that $(\Omega, \mathcal{F}, \mathbb{F} = (\mathcal{F}_t)_{t \ge 0}, \mathbb{P}) $  supports a Brownian motion $W = (W_t)_{t \geq 0}$ and any other process that is needed.

		The agent can invest into a risk-free bond or bank account and a risky asset, the dynamics of which are influenced by the stochastic factor $Y$.
		Denote the state space of $Y$ by $E$.
		The bank account process $B$ follows the dynamics \[
				dB_t = r(Y_t) B_t dt, \quad B_0 = 1
		,\] where $r: E \to \R $ is a (state-dependent) interest rate.
		The risky asset $S$ has dynamics \[
				dS_t = S_t ( (r(Y_t) + \lambda(Y_t) \sigma(Y_t)) dt + \sigma(Y_t) dW_t), \quad S_0 = s_0 > 0
		,\] where $\lambda: E \to \R$ is the market price of risk and $\sigma: E \to (0, \infty)$ is the volatility of the asset.
        For ease of notation, we assume that there is only one risky asset.
        Our results extend to the case of multiple risky assets, see \cref{finite_state_space:multi-asset,hjb:multi-asset}.

		For initial wealth $x > 0$, the agent allocates a fraction of wealth $\Pi = (\Pi_t)_{t \ge 0}$ into the risky asset and consumes at a rate given by the fraction of wealth $\Xi = (\Xi_t)_{t \ge 0}$. Their wealth process $X^{\Pi,\Xi}$ has dynamics \[
				\d X^{\Pi,\Xi}_t = X^{\Pi,\Xi}_t ((r(Y_t) + \Pi_t \lambda(Y_t) \sigma(Y_t) - \Xi_t) dt + \Pi_t \sigma(Y_t) dW_t), \quad X^{\Pi,\Xi}_0 =x.
		\] 
		The optimal investment and consumption problem is to find
		\begin{align}
				\label{setting:control_problem}
				V(x, y) = \sup_{(\Pi, \Xi) \in \mathcal{A}} \E\left[\int_{0}^{\infty} e^{-\int_{0}^{t} \delta(Y_s) \d s } \frac{(\Xi_t X^{\Pi, \Xi}_t)^{1-R}}{1-R} \d t \condbar X^{\Pi, \Xi}_0=x, Y_0=y \right] 
		,\end{align}
		where $\delta : E \to \R$ denotes the (state-dependent) impatience rate and $R \in (0, \infty) \setminus \{1\} $ is the agent's risk aversion.
		The set of admissible strategies is denoted by $\mathcal{A}$.
		We call a strategy $(\Pi, \Xi)$ admissible if $\Pi$ and $\Xi$ are $\mathbb{F}$-progressively measurable, $\Xi_t \ge 0$ a.s.\ for all $t \ge 0$, and 
        \begin{equation}
            \label{setting:admissiblity}
            \int_0^t (|r(Y_s)| + |\Pi_s \lambda(Y_s) \sigma(Y_s)| + \Xi_s
        + \Pi_s^2 \sigma(Y_s)^2) \d s < \infty  \; \; a.s.
        \end{equation}  
        for all $t \ge 0$.\footnote{Note that nonnegativity of $X$ follows automatically.}

        \begin{definition}
                We call problem \eqref{setting:control_problem} \emph{well-posed} if $|V(x, y)| < \infty$ for all $x > 0, y \in E$.
        		Otherwise, problem \eqref{setting:control_problem} is called \emph{ill-posed}.
                We call the function
                \[
                        \eta(y) = \frac{1}{R} \left(\delta(y) - (1-R) \left(r(y) + \frac{\lambda(y)^2}{2R}\right)\right)
                \]
                the \emph{myopic consumption rate}.
                We say that the model is \emph{myopically well-posed at $y \in E$} if $\eta(y) > 0$, and \emph{myopically ill-posed at $y \in E$} if $\eta(y) \le 0$.
                The model is \emph{myopically well-posed everywhere} if $\eta(y) > 0$ for all $y \in E$.
                If there exists a constant $C > 0$ with $\eta \ge C$, we call the problem \emph{strongly myopically well-posed}.
        \end{definition}

        This terminology is motivated by the fact that if the stochastic factor $Y$ is frozen at the value $y$ (i.e., we are in a Black--Scholes market), problem \eqref{setting:control_problem} is well-posed if and only if $\eta(y) > 0$ and in this case $\eta(y)$ is the optimal consumption rate; see e.g.\ \cite{Herdegen2021}. 
        Note that we allow $\eta$ to be negative in some states (which corresponds to the frozen problem being ill-posed).
        Under the myopic policy, no intertemporal hedging takes place.
        
        For ease of notation, we will henceforth often omit the dependence of the coefficients on $Y$.

		\section{Stochastic Factor with Finite State Space}
		\label{section:finite_state_space}
        
        In this section, we study the well-posedness of the optimal investment and consumption problem \eqref{setting:control_problem} in a regime-dependent market with a finite number $N \in \N$ of regimes, so that $E = \{1, \ldots, N\} $.
        Such models are of interest in their own right since this setting allows for a complete characterisation of the effects of myopic ill-posedness on the well-posedness of problem \eqref{setting:control_problem}.
        Moreover, they arise naturally from a discretisation of a diffusion model, which will allow us to develop a numerical scheme for the diffusion setting in \cref{section:numerics}.
        
		We assume that $Y$ is a continuous-time Markov chain with Q-matrix $Q$ that is independent of the Brownian motion $W$.
		
		In the following, we identify functions from $E$ to $\R$ with vectors in $\R^{N}$.
		Inequalities between vectors and matrices are to be understood component-wise.
		For $p \in \R$, $x \in (0, \infty)^{N}$, we write $x^{p} := (x_1^{p}, \ldots, x^n_p)^\top \in (0, \infty)^{N}$.
		We write $\one = (1, \ldots, 1)^\top \in \R^{N}$ for the vector of ones.

        The Hamilton-Jacobi-Bellman (HJB) equation in this model is given by \begin{multline*}
                \sup_{(\pi, \xi) \in \R \times (0, \infty)} \bigg\{ \frac{(\xi x)^{1-R}}{1-R} + (r + \pi \lambda \sigma - \xi) x \frac{\partial V}{\partial x}(x,y) +\\+ \frac{1}{2} \pi^2 \sigma^2 x^2 \frac{\partial^2 V(x,y)}{\partial x^2} + (Q V(x, \cdot))(y) - \delta V(x, y) \bigg\} = 0
        .\end{multline*}
		Making the Ansatz $V(x, y) = \frac{x^{1-R}}{1-R} f(y)$, one can show that the HJB equation reduces to the Matrix HJB equation 
		\begin{align}
				\label{finite_state_space:hjb}
				\left(\diag(\eta) - \frac{1}{R} Q\right) f = f^{1-\frac{1}{R}}
		,\end{align}
		with candidate optimal controls $\hat{\Pi}_t = \hat{\pi}(Y_t)$ and $\hat{\Xi}_t = \hat{\xi}(Y_t)$, where \[
				\hat{\pi}(y) = \frac{\lambda(y)}{R \sigma(y)}, \quad \hat{\xi}(y) = f(y)^{-\frac{1}{R}}
		.\]
		\Textcite[Lemmas~4.1~\&~4.2]{Sotomayor2009} show that \cref{finite_state_space:hjb} has a solution if $\eta_i > 0$ for $i=1,\ldots,N$.
		We extend their results to give a necessary and sufficient condition for when \cref{finite_state_space:hjb} has a solution; this will include cases where $\eta_i \le 0 $ for some (but not all) $i$.
        For example, this may occur for many parameter combinations within discretisations of the Vasicek class of models or the Heston model with risk aversion $R \in (0, 1)$.
		
        \begin{remark}
            \label{finite_state_space:multi-asset}
            Consider the multi-asset model \[
                    dS_t = \diag(S_t) ( (r(Y_t) \one + \Sigma(Y_t) \lambda_m(Y_t) ) dt + \Sigma(Y_t) dW_t)
            ,\] where $S$ is the vector of $d$ risky assets, $W$ is an $n$-dimensional Brownian motion, $\Sigma: E \to \R^{d \times n}$ is the volatility matrix, and $\lambda_m: E \to \R^n$ is the market price of risk vector.
            In the multi-asset case, $Y$ remains a continuous-time Markov chain with state space $E$ that is independent of $W$.
            We assume that $C = \Sigma \Sigma^T$ is invertible.
            The HJB equation in this model is obtained by replacing $\lambda^2 \leadsto \| P \lambda_m \|_2^2$ in \eqref{finite_state_space:hjb}, where $P = \Sigma^T C^{-1} \Sigma$ is the orthogonal projection onto the image of $\Sigma^T$.
            Note that $P$ is the identity if $n = d$.
            Since the equation in the multi-asset case is structurally identical to \eqref{finite_state_space:hjb}, the results of \cref{section:finite_state_space,section:numerics} go through unchanged.
        \end{remark}

        First, we give a sufficient criterion for the existence of a solution to the Matrix HJB equation \eqref{finite_state_space:hjb}.
        Recall that the inequality $A^{-1} \ge 0$ is to be understood componentwise.

		\begin{lemma}
				\label{finite_state_space:sufficient}
				Let $A\in \R^{N\times N}$ be invertible, and let $p < 1$.
				If $A^{-1} \ge 0$ and $A^{-1}_{ii}>0$ for all $i=1,\ldots,N$, then the equation $Ax = x^{p}$ has a solution $x \in \R^{N}$ with $x > 0$.
		\end{lemma}
		\begin{proof}
				First, consider the case $p < 0$.
				Denote
				\[
						f^{<}_i(x) = (A^{-1} x^{p})_i - x_i = \sum_{j \neq i} A^{-1}_{ij} x_j^{p} + A^{-1}_{ii} x_i^{p} - x_i, \quad i = 1,\ldots,N
				.\]
				Let $a_i$ be the positive solution of
				\[
						A^{-1}_{ii} a_i^{p} - a_i = 1
				,\]
				and $b_i$ the positive solution to \[
						\sum_{j \neq i} A^{-1}_{ij} a_j^{p} + A^{-1}_{ii} b_i^{p} - b_i = -1
				.\]
				Solutions to these equations exist by the intermediate value theorem.
				In particular, notice that $a_i < b_i$.
				Now, set $K = \bigtimes_{i=1}^{N} [a_i, b_i]$.
				For any $x\in K$ with $x_i = a_i$ for some $i=1,\ldots,N$, we have \[
						f^{<}_i(x) \ge A^{-1}_{ii} x_i^{p} - x_i = A^{-1}_{ii} a_i^{p} - a_i = 1 > 0
				,\] and for $x\in K$ with $x_i = b_i$, we have \[
						f^{<}_i(x) \le \sum_{j\neq i} A^{-1}_{ij} a_j^{p} + A^{-1}_{ii} x_i^{p} - x_i = \sum_{j \neq i} A^{-1}_{ij} a_j^{p} + A^{-1}_{ii} b_i^{p} - b_i = -1 < 0
				.\]
				Hence, the Poincaré-Miranda theorem implies that there exists an $x\in K$ with $f^{<}(x) = 0$, i.e. $Ax = x^{p}$.

				Next, consider the case $0 < p < 1$.
				Denote \[
						f^{>}_i(x) = \frac{\sum_{i \neq j} A^{-1}_{ij} x_j^{p} + A^{-1}_{ii} x_i^{p}}{x_i}
				.\]
				Set $a_i = \left( \frac{1}{2} A^{-1}_{ii} \right)^{\frac{1}{1-p}}$, and $b_i = \left(2 \sum_{k,j=1}^{N} A^{-1}_{kj}\right)^{\frac{1}{1-p}}$.
				Note that $a_i < b_i$, and that all $b_i$ coincide.
				Again, denote $K = \bigtimes_{i=1}^{N} [a_i, b_i]$.
				For $x\in K$ with $x_i = a_i$ for some $i = 1,\ldots,N$, we have \[
						f^{>}_i(x) \ge \frac{A^{-1}_{ii} x_i^{p}}{x_i} = \frac{A^{-1}_{ii} a_i^{p}}{a_i} = 2 > 1
				.\]
				For $x\in K$ with $x_i = b_i$, we have \[
						f^{>}_i(x) \le \frac{\sum_{j\neq i} A^{-1}_{ij} b_j^{p} + A^{-1}_{ii} x_i^{p}}{x_i} = \frac{\sum_{j=1}^{N} A^{-1}_{ij} b_i^{p}}{b_i} \le \frac{\sum_{k,j=1}^{N} A^{-1}_{kj} b_i^{p}}{b_i} = \frac{1}{2} < 1
				.\]
				By the Poincaré-Miranda theorem, there exists an $x\in K$ with $f^{>}_i(x) = 1$ for all $i=1,\ldots,N$, i.e.\ $Ax = x^{p}$.

				Finally, if $p = 0$, the solution is given by $x = A^{-1} \one$.
		\end{proof}

		We are mostly interested in applying this result to the case $A = A_{\eta, Q}$ where $A_{\eta,Q}:= \diag(\eta) - \frac{1}{R}Q$.
		Since $Q$ is a Q-matrix, all off-diagonal entries of this matrix are non-positive.
		This additional structure is helpful for obtaining a necessary and sufficient condition for the existence of a solution.

        For the next theorem, we recall the definition of $Z$-matrices and $M$-matrices and collect some equivalent characterisations of $M$-matrices that will be used throughout the rest of the paper.

   		\begin{definition}
				\label{finite_state_space:$M$-matrix:definition}
				A matrix $A \in \R^{N \times N}$ is called a \emph{$Z$-matrix} if all off-diagonal elements are non-positive, i.e.\ $A_{ij} \le 0$ for all $i, j = 1,\ldots,N$ with $i \neq j$.
				
				A $Z$-matrix $A$ is called an \emph{$M$-matrix} if it can be written as $A = s \operatorname{Id} - B$ for some matrix $B \in \R^{N \times N}$ with $B \ge 0$ and $s \in \R$ with $s \ge \rho(B)$, where $\rho(B)$ is the spectral radius of $B$.
		\end{definition}
		
		\begin{lemma}
				\label{finite_state_space:$M$-matrix:properties}
				Let $A \in \R^{N \times N}$ be a $Z$-matrix.
				Then the following statements are equivalent:
				\begin{enumerate}[(\roman*),ref=\roman*]
						\item \label{finite_state_space:$M$-matrix:properties:m_matrix} $A$ is a non-singular $M$-matrix.
						\item \label{finite_state_space:$M$-matrix:properties:inverse_nonnegative} $A$ is invertible, and $A^{-1} \ge 0$.
						\item \label{finite_state_space:$M$-matrix:properties:monotone} For all $x\in \R^{N}$ with $Ax \ge 0$, we have $x \ge 0$.
						\item \label{finite_state_space:$M$-matrix:properties:positive_image} There exists some $x\in \R^{N}$ with $x > 0$ and $Ax > 0$.
						\item \label{finite_state_space:$M$-matrix:properties:principal_minors} All leading principal minors of $A$ are positive.
						\item \label{finite_state_space:$M$-matrix:properties:eigenvalues} The real parts of all eigenvalues of $A$ are positive.
						\item \label{finite_state_space:$M$-matrix:properties:diagonal_nonsingular} For any non-negative diagonal matrix $D$, $A + D$ is non-singular.
				\end{enumerate}
				If $A = s \operatorname{Id} - B$ is a non-singular $M$-matrix, then $A^{-1} = \frac{1}{s} \sum_{n=0}^{\infty} \frac{1}{s^{n}} B^{n}$, and \[
						\|A^{-1}\|_\infty \le \frac{\|x\|_\infty}{\min_{j=1,\ldots,N} (Ax)_j}
				\] for any $x\in \R^{N}$ with $Ax > 0$.
				Furthermore, all diagonal elements of $A$ are positive.
		\end{lemma}
		\begin{proof}
				The equivalence of \eqref{finite_state_space:$M$-matrix:properties:m_matrix}--\eqref{finite_state_space:$M$-matrix:properties:diagonal_nonsingular} is a subset of \cite[Thm.~1]{Plemmons1977}, specifically conditions $F_{15}$, $F_{16}$, $K_{33}$, $D_{11}$, $J_{29}$, $A_3$.

				Now, let $A = s \operatorname{Id} - B$ be a non-singular $M$-matrix. 
				By \eqref{finite_state_space:$M$-matrix:properties:eigenvalues}, we have $s > \rho(B) $, so that $\rho(\frac{1}{s} B) < 1$.
				Hence, the Neumann series for $\frac{1}{s}B$ converges, so \[
						A^{-1} = \frac{1}{s} \left(\operatorname{Id} - \frac{1}{s} B\right)^{-1} = \frac{1}{s} \sum_{n=0}^{\infty} \frac{1}{s^{n}} B^{n}
				.\] 

				The bound on $\|A^{-1}\|_\infty$ follows from \cite[Thm.~2.1]{Axelsson1990}.
				Lastly, the positivity of the diagonal elements follows from condition $A_1$ of \cite[Thm.~1]{Plemmons1977} since the diagonal elements are principal minors.
		\end{proof}

		Using the concepts of $Z$- and $M$-matrices we can now prove a necessary and sufficient criterion for the equation $Ax = x^{p}$ to have a solution.

		\begin{theorem}
				\label{finite_state_space:z_matrix:criterion}
				Let $A \in \R^{N \times N}$ be a $Z$-matrix and $p < 1$.
				The equation $Ax = x^{p}$ has a solution $x\in \R^{N}$ with $x > 0$ if and only if $A$ is a non-singular $M$-matrix.
		\end{theorem}
		\begin{proof}
				First, let $x > 0$ be a solution to $Ax = x^{p}$.
				In particular, this means that $x > 0$ and $Ax > 0$. 
				Hence, $A$ is a non-singular $M$-matrix by condition \eqref{finite_state_space:$M$-matrix:properties:positive_image} in \cref{finite_state_space:$M$-matrix:properties}.

				Conversely, assume that $A$ is a non-singular $M$-matrix.
				By condition \eqref{finite_state_space:$M$-matrix:properties:inverse_nonnegative} in \cref{finite_state_space:$M$-matrix:properties}, we have $A^{-1} \ge 0$.
				Furthermore, we have the representation \[
						A^{-1} = \frac{1}{s} \sum_{n=0}^{\infty} \frac{1}{s^{n}} B^{n}
				\] for some matrix $B\ge 0$ and scalar $s>0$.
				In particular, this implies that $A^{-1}_{ii} \ge \frac{1}{s} > 0$ for all $i=1,\ldots,N$.
				The existence of a solution $x>0$ to $Ax = x^{p}$ now follows from \cref{finite_state_space:sufficient}.
		\end{proof}

        We call a vector $x \in \R^N$ with $x > 0$ a \emph{supersolution} to the equation $Ax = x^p$ if $Ax \ge x^p$.
        Similarly, we call $x$ a \emph{subsolution }if $Ax \le x^p$.
        We proceed to show that there exists an ordering between sub- and supersolutions to the equation $Ax = x^p$, which implies in particular that solutions are unique.

		\begin{proposition}
				\label{finite_state_space:z_matrix:uniqueness}	
				Let $A \in \R^{N \times N}$ be a non-singular $M$-matrix and $p < 1$.
				If $x \in \R^N$ with $x > 0$ is a subsolution to the equation $Ax = x^p$ and $y \in \R^N$ with $y > 0$ is a supersolution, then $x \le y$.
				In particular, there exists a unique solution $x \in \R^N$ with $x > 0$  to the equation $Ax = x^p$.
		\end{proposition}
		\begin{proof}
				Set $\beta = \max_{i=1,\ldots,N} \frac{x_i}{y_i} $.
				Let $i$ be an index at which the maximum is attained.
				Since $A_{ij} \le 0$ for $j \neq i$ and $x_i = \beta y_i$, we have \[
						\beta^{p} y_i^{p} = x_i^{p} \ge (Ax)_i = A_{ii} x_i + \sum_{j \neq i} A_{ij} x_j \ge A_{ii} \beta y_i + \sum_{j \neq i} A_{ij} \beta y_j = \beta (Ay)_i \ge \beta y_i^{p}
				.\]
				Hence, $\beta \le 1$, i.e. $x \le y$.
				This immediately implies uniqueness of a solution through symmetry.
		\end{proof}

		We can now completely characterise the well-posedness of the optimal investment and consumption problem \eqref{setting:control_problem}.
        Note that for $R > \frac{1}{2}$ the value function can be efficiently computed numerically via a fixed point iteration as presented in \cref{numerics:fixed_point_iteration}.

        The proof of the verification theorem for $R > 1$ by \textcite[Thm.~3.2]{Sotomayor2009} appears to implicitly rely on the additional assumption that only those strategies $(\Pi, \Xi)$ for which $\liminf_{t \to \infty} \E[e^{-\delta t} \hat V(X^{\Pi, \Xi}_t, Y_t)] \ge 0$ are admissible.
        This restricts the class of admissible strategies quite severely; in a Black-Scholes market even some constant-proportional strategies fail to satisfy this transversality condition, see \cite[Rem.~4.7]{Herdegen2021}.
        Since we do not want to restrict the class of admissible strategies, we give a self-contained proof that covers both $R \in (0, 1)$ and $R > 1$.
        
		\begin{theorem}
				\label{finite_state_space:criterion}
				The optimal investment and consumption problem \eqref{setting:control_problem} is well-posed if and only if $A_{\eta, Q} := \diag(\eta) - \frac{1}{R} Q$ is a non-singular $M$-matrix.
				In this case, the value process and optimal controls are given by $V_t = \hat{V}(X_t, Y_t)$, $\hat{\Pi}_t = \hat{\pi}(Y_t)$, and $\hat{\Xi}_t = \hat{\xi}(Y_t)$, where \[
						\hat{V}(x, y) = \frac{x^{1-R}}{1-R} f(y), \quad \hat{\pi}(y) = \frac{\lambda(y)}{R \sigma(y)}, \quad \hat{\xi}(y) = f(y)^{-\frac{1}{R}}
				,\] and where $f\in \R^{N}$ with $f > 0$ is the unique positive solution to the equation $A_{\eta,Q} f = f^{1-\frac{1}{R}}$.
		\end{theorem}
		\begin{proof}
				Note that $A_{\eta,Q} = \diag(\eta) - \frac{1}{R} Q$ is a $Z$-matrix since $Q$ is a Q-matrix.

				First, assume that $A_{\eta,Q}$ is a non-singular $M$-matrix.
				In this case, a unique solution $f > 0$ to the Matrix HJB equation $A_{\eta, Q} f = f^{1-\frac{1}{R}}$ exists by \cref{finite_state_space:z_matrix:criterion}.

                We aim to verify that $\hat{V}$ is indeed the value function. 
                The perturbation argument for this from the Black-Scholes case by \textcite[Thm.~5.1, Cor.~5.4]{Herdegen2021} carries over to the stochastic factor setting.
                The only facts used in the argument that require special consideration in our setting are that 
                \begin{enumerate}[(\roman*),ref=\roman*]
                    \item \label{finite_state_space:verification:attained} $\E[\int_{0}^{\infty} \exp(-\int_{0}^{t} \delta(Y_s) \d s) \frac{(\hat{\Xi}_t X^{\hat{\Pi}, \hat{\Xi}}_t)^{1-R}}{1-R}\d t | X_0 = x, Y_0 = y] = \hat{V}(x, y)$,
                    \item \label{finite_state_space:verification:transversality} $\liminf_{t\to \infty} \E[\exp(-\int_{0}^{t} \delta(Y_s) \d s) \hat{V}(X^{\hat{\Pi}, \hat{\Xi}}_t, Y_t) ] = 0$,
                    \item \label{finite_state_space:verification:martingale} $\Lambda_t := \int_{0}^{t} \exp(-\int_{0}^{s} \delta(Y_u) \d u) \frac{(\hat{\Xi}_s X^{\hat{\Pi}, \hat{\Xi}}_s)^{1-R}}{1-R} \d s + \exp(-\int_{0}^{t} \delta(Y_s) \d s) \hat{V}(X^{\hat{\Pi}, \hat{\Xi}}_t, Y_t)$ is a martingale.
                \end{enumerate}
                
                For brevity, denote the candidate wealth process by $\hat{X} = X^{\hat{\Pi}, \hat{\Xi}}$, which is given by \[
						\hat{X}_t = X_0 \exp\left(\int_{0}^{t} \left(r(Y_s) + \frac{\lambda(Y_s)^2}{R} - f(Y_s)^{-\frac{1}{R}}\right) \d s\right) \mathcal{E}\left(\frac{\lambda(Y_\cdot)}{R} \cdot W\right)_t
				.\] 
                From this and Fubini-Tonelli, we get 
                \begin{multline*}
                    \E\left[\int_{0}^{\infty} \exp\left(-\int_{0}^{t} \delta(Y_s) \d s\right) \frac{(\hat{\Xi}_t \hat{X}_t)^{1-R}}{1-R}\d t \condbar Y_0 = y \right] = \\ =
                    \frac{X_0^{1-R}}{1-R} \int_{0}^{\infty} \E\left[\exp\left( - \int_{0}^{t} \left(R \eta(Y_s) + (1-R) f(Y_s)^{-\frac{1}{R}}\right) \d s \right) f(Y_t)^{1-\frac{1}{R}} \condbar Y_0 = y \right] \d t
                .\end{multline*}
                Set \[
						(h(t))_j = \E\left[\exp\left( - \int_{0}^{t} \left( R \eta(Y_s) + (1-R) f(Y_s)^{-\frac{1}{R}} \right) \d s \right) f(Y_t)^{1-\frac{1}{R}} \condbar Y_0 = j\right] 
				.\] 
				By Feynman-Kac, $h$ satisfies \[
						h' = Qh - \left( R \diag(\eta) + (1-R) \diag\left(f^{-\frac{1}{R}}\right) \right) h = - \left(RA + (1-R) \diag\left( f^{-\frac{1}{R}} \right) \right) h
				\] and $h(0) = f^{1-\frac{1}{R}}$, so $h(t) = \exp(-B t) f^{1-\frac{1}{R}}$, where $B = RA_{\eta,Q} + (1-R) \diag(f^{-\frac{1}{R}})$.
				Since
				\[
						Bf = Rf^{1-\frac{1}{R}} + (1-R) f^{1-\frac{1}{R}} = f^{1-\frac{1}{R}} > 0
				,\]
				$B$ is a non-singular $M$-matrix by condition \eqref{finite_state_space:$M$-matrix:properties:positive_image} in \cref{finite_state_space:$M$-matrix:properties}.
				By condition \eqref{finite_state_space:$M$-matrix:properties:eigenvalues}, this implies that the real parts of all eigenvalues of $B$ are positive. 
				This yields that \[
                        \frac{X_0^{1-R}}{1-R} \int_{0}^{\infty} h(t) \d t = \frac{X_0^{1-R}}{1-R} B^{-1} f^{1-\frac{1}{R}} = \frac{X_0^{1-R}}{1-R} f
                ,\] which is exactly \eqref{finite_state_space:verification:attained}.

                For \eqref{finite_state_space:verification:transversality}, we have 
                \begin{multline*}
                    \E\left[\exp\left(-\int_{0}^{t} \delta(Y_s) \d s\right) \hat{V}(\hat{X}_t, Y_t) \right] = \E\left[\exp\left(-\int_{0}^{t} \delta(Y_s) \d s\right) \frac{\hat{X}_t^{1-R}}{1-R} f(Y_t) \right] =\\= \frac{X_0^{1-R}}{1-R} \exp(-Bt) f \to 0 \text{ as } t \to \infty
                .\end{multline*}

                Finally, for \eqref{finite_state_space:verification:martingale}, notice that 
                \begin{multline*}
                        \hat{X}_t^{2-2R} =\\= X_0^{2-2R} \exp\left((2-2R) \int_{0}^{t} \left(r(Y_s) + \frac{\lambda(Y_s)^2}{R} - f(Y_s)^{-\frac{1}{R}} + \frac{1}{2} \frac{1-2R}{R^2} \lambda(Y_s)^2 \right) \d s \right) \\ \mathcal{E}\left(\frac{2-2R}{R} \lambda(Y_\cdot) \cdot W\right)_t
                .\end{multline*}
                Since $E$ is finite, all functions from $E$ to $\R$ are bounded, and so there exists some constant $C > 0$ s.t.\ $\E[\hat{X}_t^{2-2R}] \le X_0^{2-2R} e^{Ct}$.
                Hence, \[
                        \E\left[\int_{0}^{t} \hat{X}_s^{2-2R} \d s\right] < \infty
                \] for all $t > 0$.
                By Itô's formula for Markov-modulated diffusion processes (see \cref{auxiliary:ito_markov_modulated}) and the HJB equation we have \[
                        \Lambda_t = \Lambda_0 + \int_{0}^{t} \exp\left(- \int_{0}^{s} \delta(Y_u) d u\right) \frac{\lambda(Y_s)}{R} f(Y_s) \hat{X}_s^{1-R} dW_s + M_t
                ,\] $M$ is a true martingale by the square-integrability of $\hat{X}^{1-R}$.
                By considering the quadratic variation and using the boundedness of the $\delta$, $\lambda$, and $f$, as well as the square-integrability of $\hat{X}^{1-R}$, the stochastic integral is also a true martingale.

                With these ingredients, the perturbation argument from \cite[Thm.~5.1, Cor.~5.4]{Herdegen2021} then yields that $\hat{V}$ is indeed the value function, and that $\hat{\Pi}$ and $\hat{\Xi}$ are the optimal controls.
                				
				Now, assume that $A_{\eta, Q}$ is not a non-singular $M$-matrix.
				Denote the minimal real part of the eigenvalues of $A_{\eta, Q}$ by $\sigma_{min}(A_{\eta, Q}) = \min_{\lambda \in \sigma(A_{\eta, Q})} \operatorname{Re}(\lambda)$, where $\sigma(A_{\eta, Q})$ is the spectrum of $A_{\eta, Q}$.
				By \eqref{finite_state_space:$M$-matrix:properties:eigenvalues} in \cref{finite_state_space:$M$-matrix:properties}, $A_{\eta, Q}$ not being a non-singular $M$-matrix means that $\sigma_{min}(A_{\eta, Q}) \le 0$.

				Let $\Delta^{*} = -R\sigma_{min}(A_{\eta, Q})$.
				For $\Delta \in \R$, define $\eta(\Delta): E \mapsto \R$ and $A(\Delta) \in \R^{N \times N}$ by
				\begin{align*}
						\eta(\Delta) &= \frac{1}{R} \left(\delta + \Delta \one - (1-R) \left(r + \frac{\lambda^2}{2 R}\right)\right), \\
						A(\Delta) &= \diag(\eta(\Delta)) - \frac{1}{R} Q = A_{\eta,Q} + \frac{\Delta}{R} \operatorname{Id}
				.\end{align*}
				We have $\sigma_{min}(A(\Delta^{*})) = 0$.
				Set $\Delta_n = \Delta^{*} + \frac{R}{n}$.
				Since the control problem is monotone decreasing in $\delta$ when $R \in (0, 1)$ and monotone increasing in $\delta$ when $R>1$, it suffices to prove that $(f_n)_i \to \infty$ for some $i=1,\ldots,N$, where $f_n > 0$ is the solution to $A(\Delta_n) f_n = f_n^{1-\frac{1}{R}}$.
				This then yields
                \begin{align*}
                        V(x, y_i) &\ge \sup_n V_n(x, y_i) \ge \lim_{n \to \infty} V_n(x, y_i) = \frac{x^{1-R}}{1-R} \lim_{n \to \infty} (f_n)_i = \infty \quad \text{ for } R \in (0,1), \\
                        V(x, y_i) &\le \inf_n V_n(x, y_i) \le \lim_{n \to \infty} V_n(x, y_i) = \frac{x^{1-R}}{1-R} \lim_{n \to \infty} (f_n)_i = -\infty \quad \text{ for } R > 1
                ,\end{align*}
                where $V_n$ is the value function of the problem with discount factor $\delta + \Delta_n \one$.
                Note that $\sigma_{min}(A(\Delta_n)) = \frac{1}{n} > 0$, so $A(\Delta_n)$ is a non-singular $M$-matrix by \eqref{finite_state_space:$M$-matrix:properties:eigenvalues} in \cref{finite_state_space:$M$-matrix:properties}, and $f_n$ is well-defined by \cref{finite_state_space:z_matrix:criterion}.

				We have \[
						A(\Delta_n) f_{n+1} = \left(A(\Delta_{n+1}) + \frac{1}{R} (\Delta_n - \Delta_{n+1}) \operatorname{Id}\right) f_{n+1} \ge A(\Delta_{n+1}) f_{n+1} = f_{n+1}^{1-\frac{1}{R}}
				,\] so $f_n \le f_{n+1}$ by \cref{finite_state_space:z_matrix:uniqueness}.
				Hence, $f := \lim_{n\to \infty} f_n$ exists in $[f_1, \infty]$.
				Assume for sake of contradiction that $f_i < \infty$ for all $i=1,\dots,N$.
				Since $A(\Delta_n) \to A(\Delta^{*})$ and $f_n \to f$, we have $A(\Delta^{*}) f = f^{1-\frac{1}{R}}$.
				Since $\sigma_{min}(A(\Delta^{*})) = 0$, $A(\Delta^{*})$ is not a non-singular $M$-matrix by condition \eqref{finite_state_space:$M$-matrix:properties:eigenvalues} in \cref{finite_state_space:$M$-matrix:properties}, which is a contradiction to \cref{finite_state_space:z_matrix:criterion}.
				Hence, $f_i = \infty$ for some $i = 1,\ldots,N$, and so the problem \eqref{setting:control_problem} is ill-posed.
		\end{proof}

        \begin{remark}
            \label{finite_state_space:well-posedness:easy_criteria}
            
            \Cref{finite_state_space:criterion} has some immediate consequences: 
				\begin{enumerate}
						\item The problem is well-posed if $\eta_i > 0$ for all $i=1,\ldots,N$: In this case, $A_{\eta, Q}$ is a non-singular $M$-matrix by condition \eqref{finite_state_space:$M$-matrix:properties:positive_image} of \cref{finite_state_space:$M$-matrix:properties} since $A_{\eta, Q} \one = \eta > 0$.
						\item The problem is ill-posed if $\eta_i \le 0$ for all $i=1,\ldots,N$: In this case, $A_{\eta, Q} - \diag(\eta) = - \frac{1}{R} Q $ is singular, so by condition \eqref{finite_state_space:$M$-matrix:properties:diagonal_nonsingular} of \cref{finite_state_space:$M$-matrix:properties}, $A_{\eta, Q}$ is not a non-singular $M$-matrix.
						\item The problem is ill-posed if $\eta_i \le -\frac{1}{R} \sum_{j\neq i} q_{ij}$ for some $i=1,\ldots,N$: In this case, the $i$\nobreakdash-th diagonal element of $A_{\eta, Q}$ is non-positive, which implies that $A_{\eta,Q}$ is not a non-singular $M$-matrix by \cref{finite_state_space:$M$-matrix:properties}.
								The financial interpretation of this is as follows: If the bound is violated, then (when starting in state $i$) the agent accumulates infinite expected utility before leaving state $i$ for the first time. 
								To see this, notice that $\sum_{j\neq i} q_{ij}$ is the rate at which the regime process jumps out of state $i$, and that $\eta_i + \frac{1}{R} \sum_{j\neq i} q_{ij}$ is the well-posedness constant for the consumption problem in a Black-Scholes market with coefficients from state $i$ and $\operatorname{Exp}\left(\sum_{j\neq i} q_{ij}\right) $-distributed random time horizon (see \cite[Thm.~VI]{Merton1971}).
						\item Heuristically, the problem is ill-posed if the average of the myopic consumption rates is negative:
								In order for the problem to be well-posed, it is necessary that $\det(A_{\eta, Q}) > 0$.
								Assume that $Q$ is an irreducible Q-matrix with invariant measure $\pi$.
								Then $\operatorname{rk}(Q) = N-1$, and we have $Q^{T} \pi = 0$ and $Q \one = 0$.
								By the formula for the adjugate of an $N\times N$-matrix with rank $N-1$, we have $\operatorname{adj}(-Q) = C \one \pi^{T}$ for some scalar $C \neq 0$.
								Since the submatrix obtained by deleting the $i $-th row and column of $-Q$ is diagonally dominant with non-negative diagonal, it has a non-negative determinant, and so we have $C > 0$ since $\one \pi^{T} \ge 0$.
								By Jacobi's formula, we obtain \[
										\frac{\partial}{\partial \eta_i} \det \left( \diag(\eta) - \frac{1}{R} Q \right)  \Bigr|_{\eta = 0} = \operatorname{adj}\left(-\frac{1}{R} Q\right)_{ii} = \frac{C}{R^{N-1}} \pi_i
								.\] 
								Hence, a first order Taylor expansion yields that for small $\eta$ \[
										\det(A_{\eta, Q}) = \det\left( \diag(\eta) - \frac{1}{R} Q \right) = \frac{C}{R^{N-1}} \E_{Y\sim \pi}[\eta(Y)] + o(\|\eta\|)
								.\]
				\end{enumerate}
		\end{remark}

		\subsection{Examples}

		\subsubsection{Cyclic Model}

		Consider the case where $Y$ cycles through the different states one after the other, jumping from state $i$ to state $(i+1) \mod N$ at rate $q_i$.
		Economically, this would correspond to e.g.\ a business cycle model.
		The Q-matrix is given by \[
				Q = \begin{pmatrix}
						-q_1 & q_1 & 0 & 0 & \cdots & 0 & 0\\
						0 & -q_2 & q_2 & 0 & \cdots & 0 & 0 \\
						\vdots & \vdots & \vdots & \vdots & \ddots & \vdots & \vdots \\
						q_N & 0 & 0 & 0 & \cdots & 0 & -q_N
				\end{pmatrix}
		.\]
		Assume that $\eta_i > -\frac{q_i}{R}$ for all $i=1,\ldots,N$ to ensure the diagonal elements of $A_{\eta, Q}$ are positive, otherwise the problem is ill-posed by \cref{finite_state_space:well-posedness:easy_criteria}.
        By condition \eqref{finite_state_space:$M$-matrix:properties:principal_minors} in \cref{finite_state_space:$M$-matrix:properties}, $A_{\eta, Q}$ is a non-singular $M$-matrix if and only if the leading principal minors of $A_{\eta, Q}$ are positive.
		The leading principal minors $(m_n)_{1 \leq n \leq N}$ are given by
		\begin{align*}
				m_n &= \prod_{i=1}^{n} \left(\eta_i + \frac{q_i}{R}\right), \quad n < N \\
				m_N &= \det(A_{\eta, Q}) = \prod_{i=1}^{N} \left(\eta_i + \frac{q_i}{R}\right) - \prod_{i=1}^{N} \left(\frac{q_i}{R}\right)
		.\end{align*}
		By assumption, $m_n > 0$ for $n=1,\ldots,N-1$.
		By \cref{finite_state_space:criterion}, the problem is thus well-posed if and only if $m_N > 0$, i.e. \[
				\prod_{i=1}^{N} \left( 1 + \frac{R}{q_i} \eta_i \right) > 1
		.\]
		Notice that \[
				\prod_{i=1}^{N} \left( 1 + \frac{R}{q_i} \eta_i \right) = \exp\left( \sum_{i=1}^{N} \log\left( 1 + \frac{R}{q_i} \eta_i \right)  \right) \le \exp\left( R\sum_{i=1}^{N} \frac{\eta_i}{q_i} \right) = \exp\left( C \E_{Y\sim \pi}[\eta(Y)]  \right)
		,\] where $C = R \frac{\sum_{i=1}^{N} \prod_{j\neq i} q_j}{\prod_{i=1}^{N} q_i} > 0$ is independent of $\eta$ and $\pi = (\pi_i)_{i=1,\dots,N}$ with $\pi_i = \frac{\prod_{j\neq i} q_j}{\sum_{k=1}^{N} \prod_{j\neq k} q_j}$ is the stationary distribution of the Markov chain.
		Hence, the problem is ill-posed if the average of the myopic consumption rates is non-positive under the stationary measure.
        This means that the Taylor expansion from \cref{finite_state_space:well-posedness:easy_criteria} gives a sufficient criterion for ill-posedness in this model.
		The converse inequality, i.e.\ the problem being well-posed if it is well-posed in the average of the individual states, does not hold due to the effects of higher order terms in $\eta$.

		\subsubsection{Nearest-Neighbour Model}

		Consider the case when the factor can jump only to its direct neighbours, i.e. \[
				Q = \begin{pmatrix}
						-q_{1,+} & q_{1,+} & 0 & \cdots & 0 & 0 & 0 \\
						q_{2,-} & -(q_{2,-}+q_{2,+}) & q_{2,+} & \cdots & 0 & 0 & 0 \\
						\vdots & \vdots & \vdots & \ddots & \vdots & \vdots & \vdots \\
						0 & 0 & 0 & \cdots & q_{N-1,-} & -(q_{N-1,-} + q_{N-1,+}) & q_{N-1,+} \\
						0 & 0 & 0 & \cdots & 0 & q_{N,-} & -q_{N,-}
				\end{pmatrix}
		.\]
		Markov chains of this type arise e.g.\ as discretisations of diffusion processes.
		For convenience of notation, set $q_{1,-} = q_{N,+} = 0$.
		As in the cyclic model, we assume that $\eta_i > - \frac{q_{i,-} + q_{i,+}}{R}$ for all $i=1,\ldots,N$ since the problem is ill-posed otherwise by \cref{finite_state_space:well-posedness:easy_criteria}.
		The leading principal minors of $A_{\eta, Q} = \diag(\eta) - \frac{1}{R} Q$ are given by the recursion \begin{align*}
				m_{-1} &= 0, \quad m_0 = 1, \\ 
				m_n &= \left( \eta_n + \frac{q_{n,-} + q_{n,+}}{R} \right) m_{n-1} - \frac{q_{n-1,+} \times q_{n,-}}{R^2} m_{n-2} \quad \text{for} \quad n=1,\ldots,N
		.\end{align*}
		  By condition \eqref{finite_state_space:$M$-matrix:properties:principal_minors} in \cref{finite_state_space:$M$-matrix:properties}, $A_{\eta, Q}$ is a non-singular $M$-matrix (and thus the problem well-posed by \cref{finite_state_space:criterion}) if and only if $m_i > 0$ for all $i=1,\ldots,N$.
		Note that the principal minors can become very large if $N$ is large.
		In this case, it is more numerically stable to instead check if the ratio $r_i = \frac{m_i}{m_{i-1}}$ satisfies $r_i > 0$ for $i=1,\ldots,N$.
		The sequence $(r_i)_{i=1,\ldots,N}$ is given by the recursion \[
				r_1 = \eta_1 + \frac{q_{1,+}}{R}, \quad r_{n} = \eta_n + \frac{q_{n,-} + q_{n,+}}{R} - \frac{q_{n-1,+} \times q_{n,-}}{R^2 r_{n-1}} \quad \text{ for } \quad n=2,\ldots,N
		.\] 

		\section{Stochastic Factor with Diffusion Dynamics}
		\label{section:diffusion}

		\subsection{Setting and HJB Equation}
        
		We now consider a setting with a continuous state space, where the stochastic factor $Y$ is given by a one-dimensional Itô diffusion with dynamics\[
				dY_t = a(Y_t) dt + b(Y_t) d\tilde{W}_t
		\] and state space $E \subseteq \R$, where $E$ is an open interval and $W$ and $\tilde{W}$ are Brownian motions with correlation $\rho(Y_t)$.\footnote{The extension to a multi-dimensional factor process is non-trivial and left for future research.}
        We assume that $r, \lambda, \sigma, a, b, \rho$ as well as the discount rate $\delta$ are locally Lipschitz-continuous functions, and that $b(y) \neq 0$,  $\sigma(y) > 0 $ as well as $ \rho(y) \in [-1, 1] $ for all $y \in E$.
        The dependence of the coefficients on $Y$ will be omitted from here on for ease of notation.
        
        The Hamilton-Jacobi-Bellman (HJB) equation in this setting is given by \begin{multline*}
                \sup_{(\pi, \xi) \in \R \times (0, \infty)} \bigg\{ \frac{(\xi x)^{1-R}}{1-R} + (r + \pi \lambda \sigma - \xi) x \frac{\partial V}{\partial x} + \frac{1}{2} \pi^2 \sigma^2 x^2 \frac{\partial^2 V}{\partial x^2} +\\+ a \frac{\partial V}{\partial y} + \frac{1}{2} b^2 \frac{\partial^2 V}{\partial y^2} + \rho b \pi \sigma x  \frac{\partial^2 V}{\partial x \partial y} - \delta V \bigg\} = 0
        .\end{multline*}
		After the transformation $V(x, y) = \frac{x^{1-R}}{1-R} f(y)$, the HJB equation becomes 
        \begin{multline*}
                \sup_{(\pi, \xi) \in \R \times (0, \infty)} \bigg\{ \xi^{1-R} + (1-R) (r + \pi \lambda \sigma - \xi) f - (1-R) \frac{R}{2} \pi^2 \sigma^2 f +\\+ (1-R) \rho \pi \sigma b f' + af' + \frac{1}{2} b^2 f'' - \delta f \bigg\} = 0
        .\end{multline*}
        Solving the first-order conditions, we get
		\begin{align}
				\label{hjb:hjb:pre_transformation}
				0 = \frac{1}{2}b^2 f'' + &\left(a + \frac{1-R}{R} \rho \lambda b\right) f' - R \eta f  + R f^{1-\frac{1}{R}} + \frac{1}{2} \frac{1-R}{R} \rho^2 b^2 \frac{(f')^2}{f} 
		,\end{align}
		with candidate optimal control functions \[
				\hat{\xi} = f^{-\frac{1}{R}}, \quad \hat{\pi} = \frac{1}{\sigma R} \left(\lambda + \rho b \frac{f'}{f}\right)
		.\] 
		After the transformation $f = u^{-R}$, the equation becomes 
		\begin{align}
                \label{hjb:hjb}
				\qquad 0 &= \frac{1}{2} b^2 u'' + \left(a + \frac{1-R}{R} \rho \lambda b\right) u' + \eta u  - u^2 - \frac{1}{2} b^2 \left(\left(1 - \rho^2\right) R + \rho^2 + 1\right) \frac{(u')^2}{u} \notag \\
                &= \frac{1}{2} b^2 u'' + \tilde a u' + \eta u  - u^2 - d \frac{(u')^2}{u},
		\end{align}
        where $\tilde{a} := a + \frac{1-R}{R} \rho \lambda b$ and $d := \frac{1}{2} b^2 ((1-\rho^2) R + \rho^2 + 1)$. 
        The candidate optimal controls are now given by
		\[
				\hat{\xi} = u, \quad \hat{\pi} = \frac{1}{\sigma R} \left(\lambda - R \rho b \frac{u'}{u}\right)
		,\] i.e., we parametrise the problem in terms of the consumption rate.

		In the case of constant correlation $\rho(y) \equiv \rho$, using a distortion transform $f = v^{\varphi}$ in \eqref{hjb:hjb:pre_transformation} with $\varphi = \frac{1}{1 - \frac{R-1}{R} \rho^2}$ yields 
		\begin{align}
				\label{hjb:hjb:constant_correlation}
				0 = \frac{1}{2} b^2 v'' + \tilde a v' - \frac{R}{\varphi} \eta v + \frac{R}{\varphi} v^{1 - \frac{\varphi }{R}}
		,\end{align}
		with candidate optimal controls \[
				\hat{\xi} = v^{-\frac{\varphi}{R}}, \quad \hat{\pi} = \frac{1}{\sigma R} \left(\lambda + \rho \varphi b \frac{v'}{v}\right) 
		.\] 
		This parametrization is useful in some places as it removes the non-linear first-order term.
		It can be interpreted as the HJB equation \eqref{hjb:hjb:pre_transformation} of a model with an independent factor, $\tilde{\rho} \equiv 0$, with adjusted drift $\tilde{a}$, and risk aversion $\tilde{R} = \frac{R}{\varphi} = (1 - \rho^2) R + \rho^2$, while keeping the myopic consumption rate $\tilde{\eta} = \eta$ the same.

        The results in the following sections rely on the general theory of sub- and supersolutions to second-order problems without boundary values that is presented in \cref{section:global_existence}.
        We state here the main notions that are needed for the following sections, the reader is directed to \cref{section:global_existence} for the theory in its full generality.

        We call a function $\alpha: E \to [0, \infty)$ that is $C^2$ apart from a finite set of kinks at which $\alpha$ is left- and right-differentiable a \emph{subsolution} to the HJB equation \eqref{hjb:hjb} if \[
                \frac{1}{2} b^2 \alpha'' \ge -\tilde{a} \alpha' - \eta \alpha + \alpha^2 + d \frac{(\alpha')^2}{\alpha} 
        \] 
        and if, at the kinks,  $D^{-} \alpha < D^{+} \alpha$, where $D^{\pm}$ denotes the right and left derivative.
        Similarly, we call a function $\beta: E \to [0, \infty)$ that is $C^2$ apart from a finite set of kinks at which $\beta$ is left- and right-differentiable a \emph{supersolution} to the HJB equation \eqref{hjb:hjb} if \[
                \frac{1}{2} b^2 \beta'' \le -\tilde{a} \beta' - \eta \beta + \beta^2 + d \frac{(\beta')^2}{\beta} 
        \] and $D^{-} \beta > D^{+} \beta$ at the kinks.
        Sub- and supersolutions to \cref{hjb:hjb:pre_transformation,hjb:hjb:constant_correlation} are defined analogously.

		Notice that some of the transformations flip the notion of sub- and supersolutions:
		Supersolutions to \cref{hjb:hjb:pre_transformation,hjb:hjb:constant_correlation} correspond to subsolutions to \cref{hjb:hjb} (and vice versa for subsolutions).

        The main result of \cref{section:global_existence} is \cref{global_existence:existence}: 
        If $0 < \alpha \le \beta$ are sub- and supersolutions to \cref{hjb:hjb}, respectively, then there exists a global solution $u: E \to (0, \infty)$ to \cref{hjb:hjb} that satisfies $\alpha \le u \le \beta$. 
        We also have a corollary which extends the result on existence of a solution to the case where we only have $0 \leq \alpha_n \leq \beta$ for some sequence of subsolutions with $\sup_{n \in \N} \alpha_n > 0$.
        Note that even when the sub- and supersolutions have kinks, this solution is a classical $C^2$-solution with no kinks.

        \begin{remark}
            \label{hjb:multi-asset}
            Consider the multi-asset model \[
                    dS_t = \diag(S_t) ( (r(Y_t) \one + \Sigma(Y_t) \lambda_m(Y_t) ) dt + \Sigma(Y_t) dW_t)
            ,\] where $S$ is the vector of $d$ risky assets, $W$ is an $n$-dimensional Brownian motion, $\Sigma: E \to \R^{d \times n}$ is the volatility matrix, and $\lambda_m: E \to \R^n$ is the market price of risk vector.
            Here, $Y$ is a one-dimensional Itô diffusion driven by a scalar Brownian motion $\tilde{W}$.
            We assume that $W$ and $\tilde{W}$ have correlation $\rho_m(Y_t)$, where $\rho_m: E \to \R^n$ with $\| \rho_m(y) \|_2 \le 1$ for all $y \in E$.
            We assume that $C = \Sigma \Sigma^T$ is invertible.
            The HJB equation in this model is obtained by replacing \[
                    \lambda^2 \leadsto \| P \lambda_m \|_2^2, \quad \rho \lambda \leadsto \lambda_m^T  P \rho_m, \quad \rho^2 \leadsto \| P \rho_m \|_2^2
            \] in \eqref{hjb:hjb}, where $P = \Sigma^T C^{-1} \Sigma$ is the orthogonal projection onto the image of $\Sigma^T$.
            Note that $P$ is the identity if $n = d$.
            If $\| P \rho_m \|_2^2$ is constant, the distortion transform with $\varphi = \frac{1}{1 - \frac{R-1}{R} \| P \rho_m \|^2}$ yields an equation of the form \eqref{hjb:hjb:constant_correlation} after the same substitutions as above.
            
            Since the equations in the multi-asset case are structurally identical to \eqref{hjb:hjb} and \eqref{hjb:hjb:constant_correlation}, the results of \cref{section:diffusion,section:numerics} go through unchanged.
        \end{remark}

		\subsection{Construction of a Candidate Solution for $R > 1$}
        \label{section:hjb:existence:R>1}
        
        In this and the following section, we present results that guarantee a global positive solution to the HJB equations \eqref{hjb:hjb} and \eqref{hjb:hjb:constant_correlation}.
        In this section, we consider under minimal assumptions the case that the risk aversion $R$ is greater than unity (there are no corresponding results for the case that $R < 1$), and in the following section the general case under slightly stronger assumptions.
       
		If $R > 1$, we can always construct a subsolution (and thus a solution) to \cref{hjb:hjb:constant_correlation} if we have a supersolution. 
		In terms of the consumption rate, this means that we can always bound the optimal consumption rate away from $\infty$, provided that we can bound the optimal consumption rate away from $0$.

		\begin{lemma}
				\label{hjb:existence:abstract}
				Assume that $R > 1$, and that $\rho$ is constant.
                Let $E_m := [e^{-}_m, e^{+}_m] \uparrow E$ be an approximating sequence for the state space $E$.
                Assume that the principal eigenvalue of the operator $Lv := -\frac{1}{2}b^2 v'' - \tilde{a} v' + \frac{R}{\varphi} \eta v$ under Dirichlet boundary conditions is positive over $E_m$ for all $m \in \N$.
				Let $\beta \in C^1(E, (0, \infty))$ be a supersolution to \cref{hjb:hjb:constant_correlation}.
				Then \cref{hjb:hjb:constant_correlation} has a global solution $v$ with $0 < v \le \beta$.
		\end{lemma}
		\begin{proof}
				Notice that since $R > 1$, $\frac{R}{\varphi} = (1 - \rho^2) R + \rho^2 \ge 1$ with equality if and only if $\rho^2 = 1$.
				By \cref{bounded_domain:criterion:dirichlet}, there exists a positive solution $v_m$ to \cref{hjb:hjb:constant_correlation} over the domain $E_m$ under Dirichlet boundary conditions.
				Setting $K_m = \sup_{E_m} \frac{v_m}{\beta} \vee 1$, we have \[
						L(K_m^{-1} v_m) = K_m^{-1} L v_m = \frac{R}{\varphi} K_m^{-1} v_m^{1-\frac{\varphi}{R}} \le \frac{R}{\varphi} (K_m^{-1} v_m)^{1-\frac{\varphi}{R}}
				,\] i.e.\ $K_m^{-1} v_m$ is a subsolution to \cref{hjb:hjb:constant_correlation} over $E_m$.
				Note that $0$ is a global solution to \cref{hjb:hjb:constant_correlation} when $\rho^2 < 1$ and a global subsolution when $\rho^2 = 1$.
				Define \[
						\alpha_m(y) = \begin{cases}
								K_m^{-1} v_m(y), & y\in E_m, \\
								0, & \text{otherwise}.
						\end{cases}
				\]
				Since $v_m > 0$ on $E_m^\circ$, we have
				\[
						\left(L + \frac{R}{\varphi} \eta^{-}\right) v_m = \frac{R}{\varphi} v_m^{1-\frac{\varphi}{R}} + \frac{R}{\varphi} \eta^{-} v_m \ge 0 \quad \text{ on } E_m^\circ
				,\]
				and $v_m = 0$ on $\partial E_m$ by the Dirichlet boundary conditions.
				Hence, Hopf's Lemma (applied to $L + \frac{R}{\varphi} \eta^{-}$) yields $v_m'(e^{-}_m) > 0, v_m'(e^{+}_m) < 0$.
				Thus, $\alpha_m$ satisfies the subsolution property at the kinks, and is hence a global subsolution.
				Note that $0 \le \alpha_m \le \beta $ by construction.

                Since $\alpha_m > 0$ on $E_m^{\circ}$, $\sup_{m \in \N} \alpha_m > 0$. 
                Thus, \cref{global_existence:existence:family} yields the existence of a global solution $v$ to \cref{hjb:hjb:constant_correlation} with $0 < v \le \beta$.
		\end{proof}

		Similarly, when the correlation is non-constant, it is enough to find a supersolution to the model with an independent factor and a modified drift. 

		\begin{corollary}
				\label{hjb:existence:abstract:general_correlation}
				Assume $R > 1$.
                Let $E_m := [e^{-}_m, e^{+}_m] \uparrow E$ be an approximating sequence for the state space $E$.
                Denote \[
                        L_0 v := -\frac{1}{2} b^2 v'' - \left(a + \frac{1-R}{R} \rho \lambda b\right) v'
                .\]
                Assume that the principal eigenvalues of the operators $L_0 + R \eta$ and $L_0 + \eta$ under Dirichlet boundary conditions are positive over $E_m$ for all $m \in \N$.
				Let $\beta \in C^1(E, (0, \infty))$ be a supersolution to the equation $L_0 v + R \eta v = R v^{1-\frac{1}{R}}$.              Then \cref{hjb:hjb:pre_transformation} has a global solution $0 < f \le \beta$.
		\end{corollary}
		\begin{proof}
				First, consider a modified model with drift $a + \frac{1-R}{R} \rho \lambda b$ and correlation $0$.
				In this model, \cref{hjb:hjb:pre_transformation,hjb:hjb:constant_correlation} coincide and read \begin{equation}
						\label{hjb:existence:abstract:general_correlation:model_0}
						L_0 f + R \eta f = R f^{1-\frac{1}{R}}
						\tag{i}
				.\end{equation} 
				By \cref{hjb:existence:abstract}, there exists a global solution $f_0$ to \eqref{hjb:existence:abstract:general_correlation:model_0} with $0 < f_0 \le \beta $.

				Next, consider a second modified model with drift $a + \frac{1-R}{R} \rho \lambda b - \frac{1-R}{R} \lambda b$ and correlation $1$.
				In this model, \cref{hjb:hjb:pre_transformation} reads \begin{equation}
 						\label{hjb:existence:abstract:general_correlation:model_1}
						L_0 f + R \eta f = R f^{1-\frac{1}{R}} + \frac{1}{2} \frac{1-R}{R} b^2 \frac{(f')^2}{f}
						\tag{ii}
				\end{equation} 
				and \cref{hjb:hjb:constant_correlation} reads 
				\begin{equation}
						\label{hjb:existence:abstract:general_correlation:model_1:distortion}
						L_0 v + \eta v = 1
						\tag{iii}
				.\end{equation}
				Since $f_0$ is a solution to \eqref{hjb:existence:abstract:general_correlation:model_0} and $R > 1$, we have \[
						L_0 f_0 + R \eta f_0 = R f_0^{1-\frac{1}{R}} \ge R f_0 ^{1-\frac{1}{R}} + \frac{1}{2} \frac{1-R}{R} b^2 \frac{(f_0')^2}{f_0}
				,\] i.e.\ $f_0$ is a supersolution to \eqref{hjb:existence:abstract:general_correlation:model_1}.
				Since the transformation from \cref{hjb:hjb:pre_transformation} to \cref{hjb:hjb:constant_correlation} preserves supersolutions, $f_0^{\frac{1}{R}}$ is a supersolution to \eqref{hjb:existence:abstract:general_correlation:model_1:distortion}.
				By \cref{hjb:existence:abstract}, there exists a global solution $v_1$ to \eqref{hjb:existence:abstract:general_correlation:model_1:distortion} with $0 < v_1 \le f_0^{\frac{1}{R}}$.
				Transforming back, $f_1 := v_1^{R}$ is a global solution to \eqref{hjb:existence:abstract:general_correlation:model_1} with $0 < f_1 \le f_0$.

				Finally, \cref{hjb:hjb:pre_transformation} in the original model reads 
				\begin{equation}
						\label{hjb:existence:abstract:general_correlation:original}
						L_0 f + R \eta f = R f^{1-\frac{1}{R}} + \frac{1}{2} \frac{1-R}{R} \rho^2 b^2 \frac{(f')^2}{f}
						\tag{iv}
				.\end{equation}
				Since $R > 1$ and $\rho^2 \in [0, 1]$, $f_0$ is a supersolution and $f_1$ a subsolution to \eqref{hjb:existence:abstract:general_correlation:original}.
				By \cref{global_existence:existence}, there exists a global solution $f$ to \eqref{hjb:existence:abstract:general_correlation:original} that satisfies $0 < f_1 \le f \le f_0 \le \beta$.
		\end{proof}

		In strongly myopically well-posed models, there exists a trivial supersolution to \cref{hjb:hjb:constant_correlation}. 
		Hence, there exists a global solution to the HJB equation.

		\begin{theorem}
				\label{hjb:existence:abstract:strongly_wellposed}
				Assume that $R > 1$, and that the model is strongly myopically well-posed.
				Then there exists a global solution $f > 0$ to \cref{hjb:hjb:pre_transformation}.
		\end{theorem}
		\begin{proof}
				Since the model is strongly myopically well-posed, there exists some constant $C > 0$ with $\eta \ge C$.
				One easily sees that $\beta \equiv C^{-R}$ is a supersolution to the equation $L_0 v + R \eta v = Rv^{1-\frac{1}{R}}$, where $L_0$ is the operator from \cref{hjb:existence:abstract:general_correlation}.
                Moreover, the principal eigenvalues of the operators $L_0 + R\eta$ and $L_0 + \eta$ under Dirichlet boundary conditions are positive by \cref{bounded_domain:strongly_wellposed}.
				Now, the existence follows from \cref{hjb:existence:abstract:general_correlation}.
		\end{proof}

        \subsection{Construction of a Candidate Solution in the General Case}

        If $R < 1$, the results of Section \ref{section:hjb:existence:R>1} are not applicable. In this case, we need a different approach. The idea is to use more explicit sub- and supersolutions. In return, we get stronger properties of the solution that are also very useful in the case $R >1$.
   
		To construct the explicit sub- and supersolutions, we define the operator \[
				\Psi g = 1 + \frac{\frac{1}{2} b^2 g'' + (a + \frac{1-R}{R} \rho \lambda b) g'}{g^2} - \frac{1}{2} b^2 ((1 - \rho^2) R + \rho^2 + 1) \frac{(g')^2}{g^3} 
		.\] 
        Note that if $u$ is a solution to \cref{hjb:hjb}, then $\Psi u = 2 - \frac{\eta}{u}$.

		\begin{theorem}
				\label{hjb:existence}
                Let $g_1, g_2 \in C^2(E)$ with $0 < g_1 \le \eta \le g_2$.
                Suppose that $0 < C_1 \le C_2$ satisfy $C_1 \le \inf_E \Psi g_1 $ and $C_2 \ge \sup_E \Psi g_2$.
				Then $C_1 g_1$ and $C_2 g_2$ are sub- and supersolutions to \cref{hjb:hjb}, respectively, and \cref{hjb:hjb} has a global solution $u$ that satisfies \[
						C_1 g_1 \le u \le C_2 g_2
				.\]
		\end{theorem}
        
		\begin{remark}
                \begin{enumerate}[label=(\alph*)]
                        \item Note that $g_1$ and $g_2$ can be chosen independently, in particular they may have completely different growth behaviour.
                        \item We have $C_i = 1$ if $g_i$ is constant, $i=1,2 $.
                \end{enumerate}
        \end{remark}
		\begin{proof}
				Set $\tilde{a} = a + \frac{1-R}{R} \rho \lambda b$, $d = \frac{1}{2} b^2 ((1 - \rho^2) R + \rho^2 + 1)$ and $\alpha = C_1 g_1$.
				By assumption, \[
						C_1 \le 1 + \frac{\frac{1}{2} b^2 g_1'' + \tilde{a} g_1'}{g_1^2} - d \frac{(g_1')^2}{g_1^3}
				.\]
				This implies that \begin{align*}
						\frac{1}{2}b^2 \alpha'' &=
						\frac{1}{2}b^2 C_1 g_1'' \ge
						C_1 \left( (C_1 - 1) g_1^2 + d\frac{(g_1')^2}{g_1} - \tilde{a}g_1' \right)  =\\&=
						-\tilde{a} \alpha' - g_1\alpha + \alpha^2 + d \frac{(\alpha')^2}{\alpha} \ge
						-\tilde{a} \alpha' - \eta \alpha + \alpha^2 + d \frac{(\alpha')^2}{\alpha}
				,\end{align*}
				so $\alpha$ is a subsolution to \cref{hjb:hjb}.

				Analogously, one sees that $\beta = C_2 g_2$ is a supersolution.
				The existence of a solution $u$ with $\alpha \le u \le \beta $ then follows from \cref{global_existence:existence}.
		\end{proof}

		In a strongly myopically well-posed model with $\sup_E \Psi \eta < \infty$, \cref{hjb:existence} guarantees the existence of a global positive solution that is lower-bounded by the same constant as the myopic consumption rate $\eta$.
		We now prove a stronger result which will be useful later as it will provide bounds on the solution that ensure that the solution has the same order of growth as $\eta$.
        Note that an analogous result holds for $y \to -\infty$ if $\eta$ is eventually decreasing.
        
		\begin{corollary}
				\label{hjb:existence:strongly_wellposed}
				Assume that $E = (E_{-}, \infty)$ for some $E_{-} \in \{-\infty\} \cup \R$.
				Suppose that the model is strongly myopically well-posed, $\eta\in C^2(E)$, $C_2 := \sup_{E} \Psi \eta < \infty$, and $\Psi \eta \to 1$ as $y\to \infty$.
				Furthermore, assume that $\eta' > 0	$ on $[y_0, \infty)$ for some $y_0\in E$.
				Then \cref{hjb:hjb} has a global solution $u > 0$ with $\tilde{C}_1 \eta(y) \le u(y) \le C_2 \eta(y)$ for all $y \ge y_1$ for some $y_1 \ge y_0, \tilde{C}_1 > 0$.
		\end{corollary}
		\begin{proof}
				Since the model is strongly myopically well-posed, there exists a constant $C_1 > 0$ such that $\eta > C_1$.
				By \cref{hjb:existence}, $C_1$ and $C_2 \eta$ are a pair of global sub- and supersolutions for \cref{hjb:hjb}.
				As $\eta$ is increasing on $[y_0, \infty)$ and $\Psi \eta \to 1$ as $y\to \infty$, there exists some $y_1 \ge y_0$ such that $\Psi \eta(y) \ge \frac{C_1}{\eta(y_1)} =: \tilde{C}_1 < 1$ for all $y \ge y_1$.
				Hence, $\tilde{C}_1\eta$ is a subsolution on $(y_1, \infty)$ by \cref{hjb:existence}.
				Since $\eta'(y_1) > 0$, \[
						\alpha(y) := \begin{cases}
								C_1, &y \le y_1, \\
								\tilde{C}_1 \eta, & y \ge y_1
						\end{cases}
				\] is a global subsolution.

				Now, a global solution $u$ with $\alpha \le u \le C_2 \eta$ exists by \cref{global_existence:existence}.
		\end{proof}

		\subsection{Uniqueness, Asymptotics and Verification}

		In this section, we give a condition for uniqueness of the solution, and characterise the asymptotic behaviour of the solution and its log-derivative.
		These asymptotics then allow us to verify the assumptions of the verification theorem of \textcite{Guasoni2020}.  
        
		Note that while the results in this section have many assumptions, almost all assumptions only concern the model coefficients and not the solution itself.
		For any given model, the coefficients are known, so the assumptions are easy to check.
        The models to keep in mind as running examples (and which will be studied in detail in \cref{section:examples}) are the following:
        \begin{itemize}
            \item Heston Model: $\lambda(y) = \lambda \sqrt{y}$, $\sigma(y) = \sqrt{y}$, $a(y) = -\kappa (y - \theta)$, $b(y) = \nu \sqrt{y}$
            \item Kim-Omberg Model: $\lambda(y) = y$, $a(y) = - \kappa (y - \theta)$, $b(y) = \nu$
            \item Vasicek Model: $r(y) = y$, $a(y) = - \kappa (y - \theta)$, $b(y) = \nu$
        \end{itemize}
        All coefficients not mentioned above are constants. For details on the parameters ranges, see \cref{section:examples}.

        \subsubsection{Uniqueness}

		We begin our analysis by showing that the asymptotic growth behaviour uniquely determines the solution.
        This theorem is particularly useful if $E = \R$ since the growth behaviour of solutions as $|y| \to \infty$ will be characterised in \cref{hjb:asymptotics}.

		\begin{theorem}
				\label{hjb:uniqueness}
				Let $u, \tilde{u}$ be two global positive solutions to \cref{hjb:hjb} with $\frac{\tilde{u}}{u} \to 1$ as $y \to \partial E$.
				Then $u = \tilde{u}$.
		\end{theorem}
		\begin{proof} Set $w = \frac{\tilde{u}}{u}$. Then
				$w$ is well-defined, positive, and converges to $1$ as $y \to \partial E$.
				By substituting $\tilde{u} = uw$ into \eqref{hjb:hjb} and using that $u$ also satisfies \eqref{hjb:hjb}, one sees that $w$ satisfies the equation \[
						\frac{1}{2}b^2 u w'' + ((b^2 - 2d) u' + \tilde{a} u) w' - d u \frac{(w')^2}{w} = u^2 w (w-1)
				.\] 

				Now, assume that $w$ has an extremum with $w > 1$.
				At that extremum, $w' = 0$ and \[
						\frac{1}{2} b^2 u w'' = u^2 w (w - 1) > 0
				,\] so it is a minimum.
				Similarly, all extrema of $w$ with $w < 1$ are maxima.

				Since $w \to 1$ as $y \to \partial E$, $w$ attains its global maximum if there is a point at which $w > 1$.
				As there is no local maximum with $w > 1$, this means $w \le 1$ everywhere.
				Analogously, it follows that $w \ge 1$, so that in total $w \equiv 1$, i.e.\ $u \equiv \tilde{u}$.
		\end{proof}

        \subsubsection{General Asymptotics}
        
        The following two subsections are dedicated to the study of asymptotic properties of the solution. 
        In this subsection, we consider general factor processes, and will specialise to the mean-reverting case in the next subsection.
        For ease of notation, we only consider the case $E = (E_{-}, \infty)$ for some $E_{-} \in \{-\infty\} \cup \R$ and $y \to \infty$. Note, however, that all results of this section also apply to the case $E = (-\infty, E_{+})$ for some $E_{+} \in \R \cup \{\infty\}$ and $y \to -\infty$ up to a sign change in all first order terms (which arises from the change of variables $y \leadsto -y$).
        
		We now state an assumption on the model coefficients under which the solution $u$ to \cref{hjb:hjb} is asymptotically equivalent to the myopic consumption $\eta$ if it has the same order of growth as $\eta$ .
		
		\begin{assumption}
				\label{hjb:asymptotics:assumption} There exists some $y_0 \in E$ such that
				\begin{enumerate}[label={(A\arabic*)},ref={A\arabic*}]
						\item \label{hjb:asymptotics:assumption:positivity}
								$\eta > 0$ over $[y_0, \infty)$ and $\eta \in C^2([y_0, \infty))$.
						\item \label{hjb:asymptotics:assumption:boundedness}$
								\frac{b^2}{\eta}$, $\frac{\tilde{a}}{\eta}$, $\frac{\eta'}{\eta}$ are bounded over $[y_0, \infty)$.
						\item \label{hjb:asymptotics:assumption:convergence}
								$\Psi \eta = 1 + \frac{\frac{1}{2}b^2 \eta'' + \tilde{a} \eta'}{\eta^2} - d \frac{(\eta')^2}{\eta^3} \to 1$ as $y \to \infty$.
				\end{enumerate}
		\end{assumption}

		\begin{remark}
				Note that \eqref{hjb:asymptotics:assumption:boundedness} implies \eqref{hjb:asymptotics:assumption:convergence} if $\frac{\eta'}{\eta} \to 0$ and $\frac{\eta''}{\eta} \to 0$ as $y\to \infty$.
		\end{remark}

        \Cref{hjb:asymptotics:assumption} is satisfied for the Heston, Kim-Omberg, and Vasicek model.
        As with the other results in this section, the next result also applies to the case $y \to -\infty$.
        If $E = \R$ and \cref{hjb:asymptotics:assumption} holds for both $y \to \infty$ and $y \to -\infty$, \cref{hjb:uniqueness,hjb:asymptotics} together show that there is at most one solution $u$ to \eqref{hjb:hjb} for which $\liminf_{|y| \to \infty} \left|\frac{u}{\eta}\right| > 0$ and $\limsup_{|y| \to \infty} \left|\frac{u}{\eta}\right| < \infty$.

		\begin{theorem}
				\label{hjb:asymptotics}
				Let $u$ be a solution to \cref{hjb:hjb} on $[y_0, \infty)$ with $C_1 \eta \le u \le C_2 \eta$ for some $0 < C_1 < C_2$.
				Moreover, assume that the model satisfies \cref{hjb:asymptotics:assumption}.
				Then $\frac{u}{\eta} \to 1$ as $y \to \infty$.
		\end{theorem}
        \begin{proof}
				Set $w = \frac{u}{\eta}$.
				We have $w \in [C_1, C_2]$, and by substituting $u = \eta w$ in \eqref{hjb:hjb} one sees that $w$ satisfies the equation \begin{align}
						\label{hjb:asymptotics:ode}
						\frac{\frac{1}{2}b^2}{\eta} w'' + \left(\frac{\tilde{a}}{\eta} + (b^2 - 2d) \frac{\eta'}{\eta^2}\right) w' + (\Psi \eta - 1) w - \frac{d}{\eta} \frac{(w')^2}{w} = w(w - 1)
				.\end{align}
				Note that $\eta > 0$ on $[y_0, \infty)$ by \eqref{hjb:asymptotics:assumption:positivity} in \cref{hjb:asymptotics:assumption}.

				Let $\varepsilon > 0$, and consider an extremum with $w \ge 1 + \varepsilon$ and $ \left| \Psi \eta - 1 \right| < \frac{\varepsilon}{2}$.
				Note that the latter is the case for all $y$ large enough by \eqref{hjb:asymptotics:assumption:convergence} in \cref{hjb:asymptotics:assumption}.
				At that extremum, $w' = 0$ and \[
						\frac{\frac{1}{2} b^2}{\eta} w'' = w(w-1) - (\Psi \eta - 1) w \ge \left(w - 1 - \frac{\varepsilon}{2}\right) w \ge \frac{\varepsilon}{2} w > 0
				,\] so the extremum is a minimum.
				Similarly, all extrema with $w < 1 - \varepsilon$ and $y$ large enough are maxima.

				Suppose that $w$ is not eventually monotone, i.e.\ there exists no $y$ s.t.\ w is either increasing or decreasing on $[y, \infty)$.
                By \cref{auxiliary:non_eventually_monotone}, there exist sequences of local minima $(y^1_n)$ and maxima $y^2_n$ with $y^i_n \to \infty$, $i = 1, 2$.
                Moreover, we can choose the sequences such that $w(y^1_n) \to \liminf_{y \to \infty} w$ and $w(y^2_n) \to \limsup_{y \to \infty} w$ as $n \to \infty$.
                Since all minima with $y$ large enough satisfy $w \ge 1 - \varepsilon$, we have $\liminf_{y \to \infty} w \ge 1 - \varepsilon$.
                Analogously, we get $\limsup_{y \to \infty} w \le 1 + \varepsilon$.
                Since $\varepsilon$ was arbitrary, we have $\lim_{y \to \infty} w = 1$.

				It remains to consider the case where $w$ is eventually monotone.
				Since $w$ is bounded, monotonicity implies that $w$ converges to some limit $L\in [C_1, C_2]$.
				By \cref{auxiliary:monotone_convergence}, there exists a sequence $y_n$ with $y_n \to \infty$, $w'(y_n) \to 0$, and $w''(y_n) \to 0$.
				Using \eqref{hjb:asymptotics:assumption:boundedness} and \eqref{hjb:asymptotics:assumption:convergence} in \cref{hjb:asymptotics:assumption}, taking limits in \cref{hjb:asymptotics:ode} along this sequence yields \[
						L (L - 1) = 0
				,\] i.e.\ $L=1$ as $L \ge C_1 > 0$.
		\end{proof}

		With the help of \cref{hjb:asymptotics} and under some additional growth conditions on the model coefficients, we are able to give a bound on the growth rate of the log-derivative of $u$, which is necessary to prove the verification theorem.

		\begin{proposition}
				\label{hjb:asymptotics:derivative}
				Let $u$ be a solution to \cref{hjb:hjb} on $[y_0, \infty)$ with $C_1 \eta \le u \le C_2 \eta$ for some $0 < C_1 < C_2$.
				Furthermore, assume that the model satisfies \cref{hjb:asymptotics:assumption}.

				Assume that $\frac{\tilde{a}}{\eta} + (b^2 - 2d) \frac{\eta'}{\eta^2} \sim K_1y^{-k}$ and $\frac{d}{\eta} \sim K_2 y^{-2l}$ for some $l \ge k \ge 0$ and some $K_1, K_2 \neq 0 $.
				Then $\frac{u'}{u} \in \mathcal{O}(y^{2l - k})$.
		\end{proposition}
        The assumptions of \cref{hjb:asymptotics:derivative} are satisfied in the Heston model with $k = 0$ and $l = 0$ (so $\frac{u'}{u}$ is bounded), the Kim-Omberg model with $k = 1$ and $l = 1$ (so $\frac{u'}{u}$ grows at most linearly), and the Vasicek model with $k = 0$ and $l = \frac{1}{2}$ (so $\frac{u'}{u}$ grows at most linearly).
		\begin{proof}
				Let $w = \frac{u}{\eta}$ as in \cref{hjb:asymptotics}.
				We have $\frac{u'}{u} = \frac{\eta'}{\eta} + \frac{w'}{w}$. 
				As $\frac{\eta'}{\eta}$ is bounded by \cref{hjb:asymptotics:assumption} and $w$ is bounded and bounded away from $0$, it is enough to consider $w'$.

                Set $h(y) = \frac{w'(y)}{y^{2l-k}}$.
                Suppose for sake of contradiction that $h$ is unbounded.

                Since $w \to 1$, we have $\liminf_{y \to \infty} |w'(y)| = 0$, so that $\liminf_{y \to \infty} |h(y)| = 0$ as well.
                Since $h$ is continuous and unbounded, this means that $h$ is not eventually monotone. 
                Since $\limsup_{y \to \infty} |h(y)| = \max \{ \limsup_{y \to \infty} h(y), -\liminf_{y \to \infty} h(y) \} $, there exists a sequence $(y_n)_{n \in \N}$ with $y_n \to \infty$ and $|h(y_n)| \to \limsup_{y \to \infty} |h(y)| = \infty$ s.t.\ each $y_n$ is a local extremum of $h$ by \cref{auxiliary:non_eventually_monotone}.

                For any $n$, we then have $w'(y_n) = y_n^{2l-k} h(y_n)$ and $h'(y_n) = \frac{w''(y_n) - w'(y_n) (2l-k) y_n^{-1}}{y_n^{2l-k}} = 0$, i.e.\ \[ 
                        w''(y_n) = (2l-k) \frac{w'(y_n)}{y_n} = (2l-k) y_n^{2l-k-1} h(y_n) 
                .\]
                Substituting this into \eqref{hjb:asymptotics:ode}, we obtain \[
                        \frac{b^2}{2\eta} (2l-k) y_n^{2l-k-1} h(y_n) + \left( \frac{\tilde{a}}{\eta} + (b^2 - 2d) \frac{\eta'}{\eta^2} \right) y_n^{2l-k} h(y_n) + (\Psi \eta - 1) w - \frac{d}{\eta} \frac{(y_n^{2l-k} h(y_n))^2}{w} = w(w-1)
                .\]
                Dividing by $\frac{d}{\eta} (y_n^{2l-k} h(y_n))^2$ yields \[
                        \frac{b^2}{2d} (2l-k) \frac{1}{y_n^{2l-k+1} h(y_n)} + \left( \frac{\tilde{a}}{\eta} + (b^2 - 2d) \frac{\eta'}{\eta^2} \right) \frac{\eta}{d} \frac{1}{y_n^{2l-k} h(y_n)} + (\Psi \eta - 1) \frac{\eta w}{d (y_n^{2l-k} h(y_n))^2} - \frac{1}{w} = \frac{\eta w(w-1)}{d (y_n^{2l-k} h(y_n))^2}
                .\]
                Now, using that $\frac{\tilde{a}}{\eta} + (b^2 - 2d) \frac{\eta'}{\eta^2} \sim K_1y^{-k}	$, $\frac{d}{\eta} \sim K_2 y^{-2l}$, $\frac{b^2}{d}$ is bounded, $\Psi \eta \to 1$, $w \to 1$, and $|h(y_n)| \to \infty$, taking the limit $n \to \infty$ yields \[
                        0 + \frac{K_1}{K_2} \cdot 0 + 0 - 1 = 0
                ,\]
                a contradiction. Hence, $h$ is bounded, and so $\frac{u'}{u} \in \mathcal{O}(y^{2l-k})$.
		\end{proof}

		If the log-derivative is bounded, we can show that it must in fact even converge to $0$.

		\begin{corollary}
				\label{hjb:asymptotics:derivative:vanishes}
				Suppose that the assumptions of \cref{hjb:asymptotics:derivative} hold for $k = l = 0$, and that $\frac{\eta'}{\eta}\to 0$.
				Then $\frac{u'}{u} \to 0$ as $y \to \infty$.
		\end{corollary}
        The assumptions of \cref{hjb:asymptotics:derivative:vanishes} are satisfied in the Heston model.
		\begin{proof}
                First, notice that $\frac{d}{\eta} \sim K_2$ implies that $\frac{b^2}{\eta} \in [\tilde{K}^{-}_2, \tilde{K}^{+}_2]$ eventually for some constants $\tilde{K}^{\pm}_2 > 0$ since $d = \frac{1}{2} b^2 ((1-\rho^2) R + \rho^2 + 1)$.
				By \cref{hjb:asymptotics:derivative},  $w'$ is bounded.
				Hence, $w''$ is also bounded by \cref{hjb:asymptotics:ode} and $\frac{b^2}{\eta}$ being bounded away from $0$.
				Since $w \to 1$, Barbalat's lemma now yields $w' \to 0$.
				This means we have $\frac{u'}{u} = \frac{w'}{w} + \frac{\eta'}{\eta} \to 0$ as $y\to \infty$ since $\frac{\eta'}{\eta} \to 0$.
		\end{proof}	

        \subsubsection{Asymptotics for Mean-Reverting Factors}

		Now, we focus on models in which the stochastic factor is mean-reverting. 
		For these, we obtain stronger asymptotics for the log-derivative, and we show that under regularity assumptions the optimal consumption rate must eventually lie below the myopic consumption rate $\eta$.
		Note that mean reversion here refers to the dynamics under the minimal distortion measure, i.e.\ to the drift term $\tilde{a}$.

        While the Vasicek model is automatically mean-reverting under the minimal distortion measure, we have to assume that $\kappa > \frac{1-R}{R} \rho \lambda \nu$ in the Heston model, and $\kappa > \frac{1-R}{R} \rho \nu $ in the Kim-Omberg model to ensure that the models are mean-reverting.
        
        As in the previous subsection, we only consider the case $E = (E_{-}, \infty)$ for some $E_{-} \in \{-\infty\} \cup \R$ and $y \to \infty$; all results also apply to the case $E = (-\infty, E_{+})$ for some $E_{+} \in \R \cup \{\infty\}$ and $y \to -\infty$ up to a sign change in all first order terms (which arises from the change of variables $y \leadsto -y$).
        
        \begin{theorem}
				\label{hjb:asymptotics:mean_reverting}
				Let $u$ be a positive solution to \cref{hjb:hjb} on $[y_0, \infty)$ with $\frac{u}{\eta} \to 1$ as $y\to \infty$.
				Assume that $\bar{a} := \frac{\tilde{a}}{\eta} + (b^2-2d) \frac{\eta'}{\eta^2} \le 0$, $\eta > 0$, and $\Psi \eta \le 1$ on $[y_0, \infty)$.
				Then $u$ satisfies either $u > \eta$ eventually or $u \le \eta$ eventually.

				If additionally $\frac{\tilde{a}}{b^2} \le -C$ on $[y_0, \infty)$ for some $C > 0$, $\eta \to \infty$ and $\frac{\eta(y)}{e^{2C y}} \to 0$ as $y\to \infty$, then $u \le \eta$ eventually.
		\end{theorem}
        The assumptions of \cref{hjb:asymptotics:mean_reverting} are satisfied in the mean-reverting Heston model, the mean-reverting Kim-Omberg model, and the Vasicek model.
        \begin{proof}
                Set $w = \frac{u}{\eta}$ as in \cref{hjb:asymptotics}.
                Assume that $w > 1$, $w' \ge 0$ at some point, then 
                \begin{equation}
                        \label{hjb:asymptotics:mean_reverting:eq}
                        \frac{\frac{1}{2}b^2}{\eta} w'' = w (w - \Psi \eta) - \bar{a} w' + \frac{d}{\eta} \frac{(w')^2}{w} > 0
                ,\end{equation}
                as $\bar{a} \le 0$ and $w > 1 \ge \Psi \eta$, i.e.\ $w'' > 0$ at that point.

                Applying \cref{auxiliary:crossing} with $h \equiv 1$, we have either $w > 1$ or $w \le 1$ eventually, and thus also $u > \eta$ or $u \le \eta$ eventually.

                Now, assume that the additional hypotheses of the theorem hold, and suppose that $u > \eta$ eventually.
				Since $\eta \to \infty$ and $\frac{u}{\eta} \to 1$, there exists $y_1 \ge y_0$ with $u'(y_1) > 0$ and $u > \eta$ over $[y_1, \infty)$.
                At any point in $[y_1, \infty)$ with $u' \ge 0$, we have \[
						\frac{1}{2} b^2 u'' = -\tilde{a} u' + u (u - \eta) + d \frac{(u')^2}{u} > -\tilde{a} u'
				,\] so that $u''(y_1) > 0$.
                By \cref{auxiliary:trap} with $h \equiv 0$, we have $u'(y) \ge 0$, and thus $u''(y) > 0$, for all $y \ge y_1$, so that $u'(y) \ge u'(y_1) > 0$ for all $y \ge y_1$.
                Since $\frac{\tilde{a}}{b^2} \le -C$, $u'$ is a supersolution to the first-order ODE \[
						f' = 2C f, \quad f(y_1) = u'(y_1)
				.\] 
                Hence, we have $u'(y) \ge u'(y_1) e^{2C(y-y_1)}$ and $u(y) \ge u(y_1) + \frac{u'(y_1)}{2C} (e^{2C(y-y_1)}-1)$.
                Since $\frac{u}{\eta} \to 1$, this contradicts $\frac{\eta(y)}{e^{2Cy}} \to 0$, so we must have $u \le \eta$ eventually.
        \end{proof}

        Under the assumption that $\Psi \eta$ is eventually concave, we obtain the stronger bound $u \le \eta \Psi \eta \le \eta$.
        Note that the eventual concavity is satisfied if $\Psi \eta$ is a rational function that converges to $1$ from below.

		\begin{proposition}
				\label{hjb:asymptotics:mean_reverting:less_than_Psi_eta}
				Let $u$ be a positive solution to \cref{hjb:hjb} on $[y_0, \infty)$ with $\frac{u}{\eta} \to 1$ and $u \le \eta$ eventually.
				Assume that $\bar{a} := \frac{\tilde{a}}{\eta} + (b^2 - 2d) \frac{\eta'}{\eta^2} \le 0$ and $\eta > 0$ on $[y_0, \infty)$, and that $\Psi \eta \in C^1([y_0, \infty))$ is increasing and concave on $[y_0, \infty)$.
				Then $u \le \eta \Psi \eta$ eventually.
		\end{proposition}
        The assumptions of \cref{hjb:asymptotics:mean_reverting:less_than_Psi_eta} are satisfied in the Heston model with $\kappa > \frac{1-R}{R} \rho \lambda \nu$, the Kim-Omberg model with $\kappa > \frac{1-R}{R} \rho \nu$, and the Vasicek model.
 		\begin{proof}
				Set $w = \frac{u}{\eta}$ as in \cref{hjb:asymptotics}.
				If $w > \Psi \eta$ and $w' \ge 0$, then \[
						\frac{\frac{1}{2}b^2}{\eta} w'' = (w - \Psi \eta) w - \bar{a} w' + \frac{d}{\eta} \frac{(w')^2}{w} > 0
				,\] i.e.\ $w'' > 0$. 
                Note that $w \le 1$ eventually by assumption.
                By \cref{auxiliary:crossing} with $h = \Psi \eta$, we thus have $w \le \Psi \eta$, i.e.\ $u \le \eta \Psi \eta$, eventually.
		\end{proof}

		If the mean reversion rate of the stochastic factor is large compared to the myopic consumption rate, we obtain that the log-derivative vanishes at infinity.

		\begin{theorem}
				\label{hjb:asymptotics:derivative:mean_reverting}
				Let $u$ be a positive solution to \cref{hjb:hjb} on $[y_0, \infty)$ with $\frac{u}{\eta} \to 1$ as $y\to \infty$.
				Assume that $\eta > 0$, $\Psi \eta \le 1$, $\tilde{a} \le 0$, and $\bar{a} := \frac{\tilde{a}}{\eta} + (b^2 - 2d) \frac{\eta'}{\eta^2} \le -C$ on $[y_0, \infty)$ for some $C > 0$.
				Furthermore, assume that $\Psi \eta$ is strictly increasing on $[y_0, \infty)$, and that $\eta \to \infty$, $\Psi \eta \to 1$, and $\frac{\eta'}{\eta} \to 0$.
				Then $\frac{u'}{u} \to 0$ as $y\to \infty$.
		\end{theorem}
        The assumptions of \cref{hjb:asymptotics:derivative:mean_reverting} are satisfied in the Heston model with $\kappa > \frac{1-R}{R} \rho \lambda \nu$ and the Vasicek model, but not in the Kim-Omberg model.
		\begin{proof}
				Let $w = \frac{u}{\eta}$.
				Since  $\frac{u'}{u} = \frac{w'}{w} + \frac{\eta'}{\eta}$, $w\to 1$ and $\frac{\eta'}{\eta} \to 0$, it is sufficient to show that $w' \to 0$.

				By \cref{hjb:asymptotics:mean_reverting}, we have either $u > \eta$ or $u \le \eta$ eventually.

				Assume that $u > \eta$ eventually. 
				If $u$ has an extremum with $u > \eta$, then \[
						\frac{1}{2}b^2 u'' = u(u-\eta) > 0
				,\] so the extremum is a minimum.
				Since $u$ can only have local minima, it must be eventually monotone, see \cref{auxiliary:non_eventually_monotone}.
				Since $\eta \to \infty$ and $\frac{u}{\eta} \to 1$, $u$ must then be eventually increasing.
				Now,\[
						\frac{1}{2}b^2 u'' = -\tilde{a} u' + u(u-\eta) + d \frac{(u')^2}{u} > 0
				\] we have eventually since $\tilde{a} \le 0$ and $u' \ge 0$, so $u$ is eventually convex. 
				By \cref{auxiliary:convex:derivative}, we then have $\frac{u'}{u} \to 0$.

				Now, consider the case $u \le \eta$, i.e.\ $w \le 1$, eventually.
				At any extremum, $w$ satisfies \[
						\frac{\frac{1}{2}b^2}{\eta} w'' = w (w - \Psi \eta) 
				,\] so the extremum is a maximum when $w < \Psi \eta$ and a minimum when $w > \Psi \eta$.
                Hence, a local maximum can't be followed by a local minimum: 
                After a local maximum, $w$ will be decreasing until the next critical point. 
                Since $\Psi \eta$ is strictly increasing, it is impossible that $w \ge \Psi \eta$ at the next critical point, which is required for a local minimum.
				Thus, $w$ must be eventually monotone, see \cref{auxiliary:non_eventually_monotone}.
				Since $w \to 1$, $w$ is eventually increasing.
				Thus, $w' \ge 0$ eventually and $\liminf_{y\to \infty} w' = 0$.

				Let $\varepsilon > 0$.
				For all $y$ large enough, we have $|w - \Psi \eta| < C \varepsilon$ since $w \to 1$ and $\Psi \eta \to 1$. 
				At any point with $w' \ge \varepsilon$, we then have \[
						\frac{\frac{1}{2} b^2}{\eta} w'' = w (w - \Psi \eta) - \bar{a} w' + \frac{d}{\eta} \frac{(w')^2}{w} > - C \varepsilon + C \varepsilon + 0 = 0 
				,\] i.e.\ $w'' > 0$.
                If $w'(y_1) \ge \varepsilon$ for some sufficiently large $y_1$, this means that $w' \ge \varepsilon$ over $[y_1, \infty)$ by \cref{auxiliary:trap} with $h(y) = \varepsilon (y-y_1)$.
				This contradicts $\liminf_{y \to \infty} w'(y) = 0$, so we must have $w' \le \varepsilon$ eventually.
				Since $\varepsilon$ was arbitrary, this means that $\limsup_{y\to \infty} w' \le 0$, and so $w' \to 0$.
		\end{proof}

        \subsubsection{Verification}

		For convenience, we state a version of the verification theorem from \textcite{Guasoni2020} in our setting.

        \begin{theorem}[{\textcite[][Thm.~3.3]{Guasoni2020}}]
				\label{verification}
				Let $u\in C^2(E)$ be a positive solution to \cref{hjb:hjb} and assume that 
				\begin{enumerate}[(\roman*),ref=\roman*]
						\item \label{verification:sde} 
								The SDE \[
                                        dZ_t = \left(a(Z_t) + \frac{1-R}{R} \rho(Z_t) \lambda(Z_t) b(Z_t) - b(Z_t)^2 ((1-\rho(Z_t)^2) R + \rho(Z_t)^2) \frac{u'(Z_t)}{u(Z_t)}\right) dt + b(Z_t) dW_t 
                                \] has a unique in law $E$-valued global weak solution.
					\item \label{verification:infinite_consumption}
							$\int_{0}^{\infty} u(Y_t) \d t = \infty $ $\hat{\mathbb{P}}$-a.s., where the probability measure $\hat{\mathbb{P}}$ exists by the extension property of $\mathbb{F} = (\mathcal{F}_t)_{t \geq 0}$ and is given by \[
                                        \frac{d \hat{\mathbb{P}}}{d \mathbb{P}} \Bigr|_{\mathcal{F}_t} = \mathcal{E}\left( (1-R) \left( \frac{\lambda(Y)}{R}
-\rho(Y) b(Y) \frac{u'(Y)}{u(Y)} \right) \cdot W - R b(Y) \frac{u'(Y)}{u(Y)} \cdot \tilde{W} \right) 
                                .\]
				\end{enumerate}
				Then the controls \[
						\hat{\pi} = \frac{1}{R \sigma} \left(\lambda - R \rho b \frac{u'}{u}\right), \quad \hat{\xi} = u
				\] are optimal for the control problem \eqref{setting:control_problem}, and its value function is $V(x,y) = \frac{x^{1-R}}{1-R} u(y)^{-R}$.
		\end{theorem}
        \begin{proof}
                Since the coefficients of the martingale problem (i) in \textcite[Thm.~3.3]{Guasoni2020} only depend on the second component, the well-posedness of the two-dimensional martingale problem is equivalent to the well-posedness of the one-dimensional martingale problem with coefficients $b^2$ and $a + \frac{1-R}{R} \rho \lambda b - b^2 ((1-\rho^2) R + \rho^2) \frac{u'}{u}$.
                In turn, the well-posedness of this martingale problem is equivalent to \eqref{verification:sde}, see e.g.\ \textcite[Cor.~4.8,4.9,Prop.~4.11]{Karatzas1998}.
                In particular, note that this implies that $\hat{\mathbb{P}}$ is well-defined.
                The optimality now follows by \cite[Thm.~3.3]{Guasoni2020}.
        \end{proof}
        
		In particular, the assumptions of \cref{verification} are satisfied if the model is strongly myopically well-posed and the coefficients in the SDE in condition \eqref{verification:sde} of \cref{verification} satisfy a linear growth condition.
        If the model is not strongly myopically well-posed, the results of this section may still be useful for establishing the conditions of \cref{verification}, see e.g.\ the Vasicek model in \cref{section:examples}.
        
        \begin{corollary}
                \label{verification:sufficient}
                \begin{enumerate}[(\roman*),ref=\roman*]
                        \item \label{verification:sufficient:sde}
                                Condition \eqref{verification:sde} in \cref{verification} is satisfied if $E = \R$ and $b$ and $a + \frac{1-R}{R} \rho \lambda b - b^2 ((1-\rho^2) R + \rho^2) \frac{u'}{u}$ both grow at most linearly.
                        \item \label{verification:sufficient:infinite_consumption}
                                Condition \eqref{verification:infinite_consumption} in \cref{verification} is satisfied if the model is strongly myopically well-posed.
                \end{enumerate}
        \end{corollary}
        \begin{proof}
                For \eqref{verification:sufficient:sde}, it is well-known that an SDE with locally Lipschitz-continuous coefficients that satisfy a linear growth condition has a pathwise unique strong solution, see e.g.\ \cite[Ex.~IX.2.10]{Revuz1999}.

                For \eqref{verification:sufficient:infinite_consumption}, notice that the solution $u$ is bounded away from $0$, so \eqref{verification:infinite_consumption} in \cref{verification} is trivially satisfied.
        \end{proof}
        
		\section{Numerical Approximation and Convergence}
		\label{section:numerics}

		In this section, we propose a numerical scheme for computing the optimal consumption rate (and thus the value function) for the optimal investment and consumption problem \eqref{setting:control_problem}.

		First, we show that in the finite-regime setting of \cref{section:finite_state_space} the solution to the matrix HJB equation \eqref{finite_state_space:hjb} (and thus also the value function and optimal policies) can be efficiently computed using a fixed point iteration if $R > \frac{1}{2}$.

		\begin{proposition}
				\label{numerics:fixed_point_iteration}
				Let $A \in \R^{N \times N}$ be a non-singular $M$-matrix, $p \in (-1,1)$, and let $x_* \in \R^{N}$ with $x_* > 0$ be the unique positive solution to the equation $Ax = x^{p}$.
				Set \[
						C_{min} = \min_{i=1,\ldots,N} (A^{-1} \one)_i, \quad C_{max} = \max_{i=1,\ldots,N} (A^{-1} \one)_i
				\] and 
				\begin{align*}
						m = \begin{cases}
								C_{min}^{\frac{1}{1-p}}, & \text{ if } p \in [0, 1), \\
								(C_{min} C_{max}^{p})^{\frac{1}{1-p^2}}, & \text{ if } p \in (-1, 0),
						\end{cases} \quad 
						M = \begin{cases}
								C_{max}^{\frac{1}{1-p}}, & \text{ if } p \in [0, 1), \\
								(C_{min}^{p} C_{max})^{\frac{1}{1-p^2}}, & \text{ if } p \in (-1, 0).
						\end{cases}
				\end{align*}
				Fix $x_1 \in \R^{N}$ with $m \one \le x_1 \le M \one$, and define the sequence $(x_n)_{n\in \N}$ by \[
						x_{n+1} = A^{-1} x_n^{p}, \quad n \ge 1
				.\] 
				Then $x_n \to x_*$ as $n\to \infty$.
				The convergence is geometric, with rate of convergence of at most $|p|$.
		\end{proposition}
		\begin{proof}
				Let $X = \{x\in \R^{N}: m \one \le x \le M \one\} $, and denote $Tx = A^{-1} x^{p}$.
				As $A^{-1} \ge 0$ by condition \eqref{finite_state_space:$M$-matrix:properties:inverse_nonnegative} of \cref{finite_state_space:$M$-matrix:properties}, we have for any $x\in X$
				\begin{align*}
						m \one = C_{min} m^{p} \one \le& Tx \le C_{max} M^{p} \one = M \one && \text{ if } p \in [0, 1), \\
						m \one = C_{min} M^{p} \one \le& Tx \le C_{max} m^{p} \one = M \one && \text{ if } p \in (-1, 0)
				,\end{align*}
				i.e., $T$ maps $X$ into itself.
				
				Consider the metric $d(x, y) = \|\log x - \log y\|_{\infty}$ on $X$.
				As \begin{align*}
						Tx &= A^{-1} x^{p} \le A^{-1} \left(e^{d(x,y)} y\right)^{p} = e^{p d(x, y)} Ty && \text{ if } p \in (0,1), \\
						Tx &= A^{-1} x^{p} \le A^{-1} \left( e^{-d(x,y)} y \right)^{p} = e^{-p d(x, y)} Ty && \text{ if } p \in (-1, 0)
			         \end{align*} for any $x, y \in X$, we have $d(Tx, Ty) \le |p| d(x, y)$.
                This means that $T$ is a contraction w.r.t.\ $d$, and by Banach's fixed point theorem and uniqueness of the solution, we have $x_* \in X$ and $d(x_n, x_*) \le |p|^{n} d(x_0, x_*)$.

				Since the map $x \mapsto \frac{|x-1|}{|\log x|}$ is bounded over $[\frac{m}{M}, \frac{M}{m}]$ by $C := \frac{\frac{M}{m} - 1}{|\log \frac{M}{m}|}$, we have \[
						\|x_n - x_*\|_{\infty} \le \|x_*\|_{\infty} \left\Vert\frac{x_n}{x_*} - \one\right\Vert_{\infty} \le C \|x_*\|_{\infty} d(x_n, x_*) \le C \|x_*\|_{\infty} d(x_0, x_*) |p|^{n}
				,\] so $x \to x_*$ geometrically with rate $|p|$ in norm as well.
		\end{proof}

		\begin{remark}
				\label{numerics:fixed_point_iteration:number_of_iterations}
				Further bounding the constants that depend on (the a priori unknown) $x_*$ yields that at most \[
						\frac{\log \varepsilon - \log\left( \frac{M^2}{m} - M \right) }{\log |p|}
				\] iterations are required to compute $x_*$ with an error of $\varepsilon$.
				In particular, notice that this does not depend on the number of states $N$.
		\end{remark}

        Now, we consider the diffusion setting of \cref{section:diffusion}.
		Let $E_m := [e^{-}_m, e^{+}_m] \uparrow E$ be an approximating sequence for the state space $E$.
		We approximate the problem \eqref{setting:control_problem} by replacing the stochastic factor $Y$ by the corresponding reflected diffusion $Y^{m}$ with state space $E_m$.
		Note that $Y$ and $Y^{m}$ coincide until $\tau_m := \inf \{t \ge 0: Y_t \not\in E_m\} $.
		We expect the approximation to be good for large $m$, especially when the dynamics of $Y$ are mean-reverting. 
        In this case, $\tau_m$ will be large, and the discounting will make the utility gained after $\tau_m$ negligible.
        In a complete market, the error introduced by the reflection can be quantified in terms of the Laplace transform of $\tau_m$, see \cref{numerics:reflection_error}.
		Now, we discretise the reflected diffusion $Y^{m}$ into a continuous-time Markov chain with finite state space. 
		The well-posedness of the optimal investment and consumption problem for such a stochastic factor is fully characterised by \cref{finite_state_space:criterion}, and its value function can be computed efficiently through the fixed point iteration of \cref{numerics:fixed_point_iteration}.

		In more detail, consider the diffusion process \[
				dY_t = a(Y_t) dt + b(Y_t) d \tilde{W}_t
		\] with state space $E \subseteq \R$, where $\tilde{W}$ is a Brownian motion independent of $W$.
		Note that the assumption of independence is not a large limitation as all models with constant correlation can be viewed as models with an independent Brownian motion and adjusted drift and risk aversion, see \cref{hjb:hjb:constant_correlation}.

		We approximate $Y$ by a reflected diffusion process $Y^{m}$ with state space $[e^{-}_m, e^{+}_m] \subset E$, and discretise this as a continuous-time Markov chain with state space $\{y_0, \ldots, y_N\} $: 
        For $N \in \N$, set $h = \frac{e^{+}_m - e^{-}_m}{N}$ and $y_{i} = e^{-}_m + ih$ for $i \in \{0, \ldots, N\}$. 
        Define the tridiagonal Q-matrix $Q^h =(Q^h_{i, j \in \{0, \ldots, N\}})$ by
		\begin{alignat*}{2}
				& Q^h_{0,0} := -\frac{1}{h^2} b(y_0)^2 - \frac{1}{h} a(y_0)^{+}, \quad && Q^h_{0,1} := \frac{1}{h^2} b(y_0)^2 + \frac{1}{h} a(y_0)^{+} \\
				& Q^h_{i,i-1} := \frac{1}{2h^2} b(y_{i})^2 + \frac{1}{h} a(y_{i})^{-}, \quad && Q^h_{ii} := - \frac{1}{h^2} b(y_{i})^2 - \frac{1}{h} |a(y_{i})|, \\ & && Q^h_{i, i+1} = \frac{1}{2h^2} b(y_{i})^2 + \frac{1}{h} a(y_{i})^{+}, \\
				& Q^h_{N,N-1} := \frac{1}{h^2} b(y_N)^2 + \frac{1}{h} a(y_N)^{-}, \quad && Q^h_{N,N} := -\frac{1}{h^2} b(y_N)^2 - \frac{1}{h} a(y_N)^{-}
		,\end{alignat*}
        and set 
        \begin{equation}
                \label{numerics:discretised_A}
                A_h = \diag(\eta) - \frac{1}{R} Q^h.
        \end{equation}
		Note that $Q^h$ corresponds to an upwind finite difference matrix for the generator of $Y^{m}$.
        For more detail on how to construct approximating Markov chains for diffusion processes, see e.g.\ \cite[Chapter~5]{Kushner2001}.
        On the relation between the stability of upwind finite difference discretisations and M-matrices in the case of linear ODEs, see \cite[Chapter~3]{Stynes2018}.
		
		As $h \to 0$, we will show that the solutions to the discretised HJB equation $A_h x_h = x_h^{p}$ converge to the solution of the Neumann problem \[
				Lf := \eta f - \frac{1}{2R} b^2 f'' - \frac{1}{R} a f' = f^{1-\frac{1}{R}}, \quad f'(e^{-}_m) = 0 = f'(e^{+}_m)
		.\] 
		This Neumann problem is precisely the HJB equation of the optimal investment and consumption problem \eqref{setting:control_problem} with stochastic factor $Y^{m}$.
		By \cref{bounded_domain:criterion:neumann}, a solution to the Neumann problem exists if and only if the principal eigenvalue of $L$ is positive.
		This mirrors the behaviour in the case of a finite state space.

		As $A_h$ is tridiagonal, the amount of time needed to compute each step in the fixed point iteration is linear in the number of states of the chain when using the tridiagonal matrix algorithm to compute $A_h^{-1} x^{p}$.
		By \cref{numerics:fixed_point_iteration:number_of_iterations}, the number of iterations needed in the fixed point iteration to achieve a fixed error $\varepsilon$ is bounded as $h\to 0$ since $A_h^{-1} \one$ converges to the solution of $L w = 1$.
		Together, this means that the time needed to solve the discrete HJB equation $A_h x_h = x_h^{p}$ is linear in the number of states, i.e. of order $\mathcal{O}(h^{-1})$.

        The following result presents the exact convergence result. 
        Note that for each $h > 0$ (and corresponding $N_h \in \N$), we tacitly identify a function $f: E_m \to \R$ with the vector $(f(y_0), \ldots, f(y_{N_h}))^\top \in \R^{N_h+1}.$
        Notice that unlike in \cref{numerics:fixed_point_iteration} we do not assume that $p > -1$ (i.e. $R > \frac{1}{2}$) here.
		In particular, this result shows that the Neumann problem has a unique solution.

		\begin{theorem}
				Let $m \in \N$, $E_m := (e^{-}_m, e^{+}_m)$, $p < 1$ and $f \in C^3(\bar E_m)$ be a positive solution to the Neumann problem $Lf = f^{p}$ over $E_m$. 
                Then for each $h$ small enough, the upwind finite difference matrix $A_h$ defined in \eqref{numerics:discretised_A} is a non-singular $M$-matrix, and the unique positive solution $x_h$ to $A_h x_h = x_h^{p}$ satisfies $\|f - x_h\|_\infty \to 0$ as $h\to 0$, with $\|f - x_h\|_\infty \in \mathcal{O}(h)$.
		\end{theorem}
		\begin{proof}
				Since $f\in C^3$ and $A_h$ is an upwind finite difference matrix, a standard Taylor expansion yields that $A_h f = f^{p} + \tau_h$ with truncation error $\tau_h \in \mathcal{O}(h)$:
                At the left boundary $i = 0$, we have
				\begin{align*}
						(Q^{h} f)_0 &= 
						\left(-\frac{1}{h^2} b(y_0)^2 - \frac{1}{h} a(y_0)^{+}\right) f_0 + \left( \frac{1}{h^2}b(y_0)^2 + \frac{1}{h}a(y_0)^{+} \right) f_1 = \\ &= 
						\left( \frac{1}{h^2} b(y_0)^2 + \frac{1}{h} a(y_0)^{+} \right) (f_1 - f_0) = 
						\left( \frac{1}{h^2} b(y_0)^2 + \frac{1}{h} a(y_0)^{+} \right) \left(\frac{1}{2} h^2 f''(y_0) + \mathcal{O}(h^3)\right) = \\ &=
						\frac{1}{2} b(y_0)^2 f''(y_0) + \mathcal{O}(h)
				.\end{align*} 
				The right boundary $i = N$ follows analogously.
				In the interior $i \in \{1,\ldots,N-1\} $, we have
				\begin{align*}
						(Q^{h} f)_i &= 
						\left( \frac{1}{2h^2} b(y_i)^2 + \frac{1}{h} a(y_i)^{-} \right) f_{i-1} + \left( -\frac{1}{h^2} b(y_i)^2 - \frac{1}{h} |a(y_i)| \right) f_i + \left( \frac{1}{2h^2} b(y_i)^2 + \frac{1}{h} a(y_i)^{+} \right) f_{i+1} = \\ &=
						\left( \frac{1}{2h^2} b(y_i)^2 + \frac{1}{h} a(y_i)^{-} \right) (f_{i-1} - f_i) + \left( \frac{1}{2h^2} b(y_i)^2 + \frac{1}{h} a(y_i)^{+} \right) (f_{i+1} - f_i) = \\ &=
						\left( \frac{1}{2h^2} b(y_i)^2 + \frac{1}{h} a(y_i)^{-} \right) \left(-h f'(y_i) + \frac{h^2}{2} f''(y_i) + \mathcal{O}(h^3)\right) \\ & \qquad + \left( \frac{1}{2h^2} b(y_i)^2 + \frac{1}{h} a(y_i)^{+} \right) \left(h f'(y_i) + \frac{h^2}{2} f''(y_i) + \mathcal{O}(h^3)\right) = \\ & =
						\frac{1}{2} b(y_i)^2 f''(y_i) + a(y_i) f'(y_i) + \mathcal{O}(h)
				.\end{align*} 
                
				Notice that $f$ is bounded and bounded away from $0$ over $E_m$.
				Since $A_h f = f^{p} + \tau_h \ge \frac{1}{2} f^{p} > 0$ for $\tau_h$ small enough, $A_h$ is a non-singular $M$-matrix for all $h$ small enough by condition \eqref{finite_state_space:$M$-matrix:properties:positive_image} of \cref{finite_state_space:$M$-matrix:properties}.
				By \cref{finite_state_space:z_matrix:criterion}, $x_h$ is thus well-defined.

				By the mean value theorem, the error $e_h = f - x_h$ satisfies \[
						A_h e_h = \tau_h + f^{p} - x_h^{p} = \tau_h + \diag(p \xi_h^{p-1}) e_h
				\] for some $\xi_h$ between $f$ and $x_h$, which is in particular positive as $f$ and $x_h$ are. 
                Setting $B_h := A_h - \diag(p \xi_h^{p-1})$ and rearranging, we obtain \[
						B_h e_h = \tau_h
				.\]
                We aim to show that for all $h > 0$ small enough, there exists a constant $c > 0$ independent of $h$ such that \[
                        \min_{j=0,\ldots,N_h} (B_h f)_j \geq c
                .\]
                Then $B_h$ is an invertible M-matrix by condition \eqref{finite_state_space:$M$-matrix:properties:positive_image} of \cref{finite_state_space:$M$-matrix:properties}. 
                Moreover, \[
						\|B_h^{-1}\|_{\infty} \le \frac{\|f\|_{\infty}}{\min_{j=0,\ldots,N_h} (B_h f)_j} \leq \frac{\|f\|_{\infty}}{c}
				\]
                by \cref{finite_state_space:$M$-matrix:properties}, and the result follows from \[
						\|f - x_h\|_\infty = \|e_h\|_\infty = \|B_h^{-1} \tau_h \|_\infty \le \frac{\|f\|_{\infty}}{c} \|\tau_h \|_\infty \to 0 \text{ as } h \to 0
				.\]
                
				Consider first the case $p \le 0$.
				Since $p \le 0$ and $f > 0$, for all $h$ small enough \[
                        B_h f \ge A_h f \geq \frac{1}{2} f^{p} \geq \frac{1}{2}\inf_{E_m}  f^p =: c > 0
                .\]
                
				Now, consider the case $p \in (0, 1)$.
				By \cref{finite_state_space:$M$-matrix:properties}, we have \[
						\|A_h^{-1}\|_\infty \le \frac{\|f\|_\infty}{\min_{j=0,\ldots,N_h} (A_h f)_j} \le \frac{\|f\|_\infty}{\frac{1}{2} \inf_{E_m} f^{p}}
				\] for $h$ small enough.
				This implies that $A_h^{-1} \one$ converges to the solution to $L w = 1$ as $h \to 0$.
				By the bounds in \cref{numerics:fixed_point_iteration}, this implies that there exist constants $C_x > c_x > 0$ such that $c_x \one \leq x_h \leq C_x \one$ for all $h$ small enough.

                Set $\tilde{B}_h := A_h - \diag(p x_h^{p-1})$. 
                Then $\tilde{B}_h x_h = (1-p) x_h^{p} \geq  (1- p) c_x^p > 0$ for all $h > 0$ small enough. 
                By the same argument as above, $\tilde{B}_h$ is an invertible $M$-matrix and satisfies \[
						\|\tilde{B}_h^{-1}\|_{\infty} \le \frac{\|x_h\|_\infty}{\min_{j=0,\ldots,N_h} (\tilde{B}_h x_h)_j}  \leq \frac{C_x}{ (1- p) c_x^p} =: \tilde C
				.\]
                Since $\diag(p (\xi_h^{p-1} - x_h^{p-1})) e_h \le 0$ by definition of $\xi_h$ and $e_h$, it follows that
                   \[
						\tilde B_h e_h = B_h e_h + \diag(p (\xi_h^{p-1} - x_h^{p-1})) e_h \leq \tau_h
				.\]
                Now using that $\tilde{B}_h^{-1} \ge 0$ by \cref{finite_state_space:$M$-matrix:properties} \eqref{finite_state_space:$M$-matrix:properties:inverse_nonnegative}, we obtain \[
						\|e_h^{+}\|_\infty \le \|\tilde{B}_h^{-1} \tau_h \|_\infty \le \tilde{C} \|\tau_h\|_\infty \to 0 \text{ as } h\to 0
				,\] where $e_h^{+}$ denotes the positive part of $e_h$. 
                Hence, using that $f$ is uniformly bounded from below, for all  $h$ small enough, we have $\xi_h \ge (\frac{1+p}{2p})^{\frac{1}{p-1}} f $.
		          This gives \[
						B_h f = f^{p} + \tau_h - p \diag(\xi_h^{p-1}) f \ge \frac{1-p}{2} f^{p} + \tau_h \ge \frac{1-p}{4} f^{p} \geq \frac{1-p}{4}\inf_{E_m}  f^p =: c > 0
                \]
				for all $h > 0$ small enough.
		\end{proof}

        \begin{remark}
                If $h < h_* := \frac{\inf_{E_m} b^2}{\|a\|_\infty}$, a central finite difference discretisation of the generator of $Y^{m}$ with grid size $h$ will also yield a Q-matrix $\tilde{Q}^h$.
                Denote $ \tilde{A}_h := \diag(\eta) - \frac{1}{R} \tilde{Q}^h$.
                As $h \to 0$, the solutions $\tilde{x}_h$ to $\tilde{A}_h \tilde{x}_h = \tilde{x}_h^p $ also satisfy $\| f - \tilde{x}_h \|_\infty \to 0$, but with faster convergence speed $\mathcal{O}(h^2)$.
                However, when $E_m$ is large, $h_*$ can be intractably small in models with high mean-reversion or models in which $b$ vanishes at $\partial E$.
        \end{remark}

        Next, we give an upper bound on the error introduced by reflecting the stochastic factor onto $E_m$ in a complete market.
        Note that $|\rho| = 1$ corresponds to $p = 0$ in \eqref{hjb:hjb:constant_correlation}.
        
        \begin{theorem}
                \label{numerics:reflection_error}
                Suppose that $p = 0$ and that the model is strongly myopically well-posed with constant $C$.
                Let $f$ be a bounded global solution to $Lf = f^p$ and $f_m$ a solution to the Neumann problem over $E_m$.
                Then $|f_m(y) - f(y)| \le \frac{2}{C} \E\left[ \exp(-C \tau_m) \condbar Y^m_0 = y \right]$ and $f_m \to f$ pointwise.
                
        \end{theorem}
        \begin{proof}
                By Feynman-Kac, we have \[
                        f(y) = \E\left[ \exp\left(-\int_0^T \eta(Y_t) \d t\right) f(Y_T) + \int_0^T \exp\left(-\int_0^t \eta(Y_s) \d s \right) \d t \condbar Y_0 = y \right]
                \] for all $T \ge 0$.
                Since $f$ is bounded and $\eta \ge C > 0$, letting $T \to \infty$ yields \[
                        f(y) = \E\left[ \int_0^\infty \exp\left(-\int_0^t \eta(Y_s) \d s \right) \d t \condbar Y_0 = y \right]
                \] by dominated convergence in the first term and monotone convergence in the second term.
                Analogously, one derives the representation \[
                        f_m(y) = \E\left[ \int_0^\infty \exp\left(-\int_0^t \eta(Y^m_s) \d s \right) \d t \condbar Y^m_0 = y \right]
                .\]
                Since $Y$ and $Y^m$ coincide until $\tau_m$, we have 
                \begin{align*}
                        |f_m(y) - f(y)| &\le 
                        \E\left[ \int_{\tau_m}^\infty \exp\left(-\int_0^t \eta(Y^m_s) \d s \right) \d t + \int_{\tau_m}^\infty \exp\left(-\int_0^t \eta(Y_s) \d s \right) \d t \condbar Y^m_0 = y \right] \\ & \le
                        2 \E\left[ \exp(-C \tau_m) \int_{\tau_m}^\infty \exp\left(-C (t - \tau_m) \right) \d t \condbar Y^m_0 = y \right] =
                        \frac{2}{C} \E\left[ \exp(-C \tau_m) \condbar Y^m_0 = y \right]
                .\end{align*}
                Now, dominated convergence yields that $f_m \to f$ pointwise since $\tau_m \to \infty$ a.s.
        \end{proof}

        \begin{remark}
                \label{numerics:reflection_error:ou}
                For instance, if $Y$ is an Ornstein-Uhlenbeck process with drift $\kappa$, long-term mean $\theta$, volatility $\nu$, and $E_m = [\theta-m, \theta+m]$, we have (see e.g. \cite[Eq. (3.0.1) in II.7.3]{Borodin2002}) \[
                        \E\left[ \exp(-C \tau_m) \condbar Y^m_0 = y \right] = \frac{{}_1F_1\left( \frac{C}{2 \kappa}; \frac{1}{2}; \frac{\kappa (y - \theta)^2}{\nu^2} \right)}{{}_1F_1\left( \frac{C}{2 \kappa}; \frac{1}{2}; \frac{\kappa m^2}{\nu^2} \right)} \sim K(y) \exp\left(-\frac{\kappa}{\nu^2} m^2 \right) \left( \frac{\kappa}{\nu^2} m^2 \right)^{\frac{1}{2} (1 - \frac{C}{\kappa})}
                ,\] where ${}_1F_1$ denotes Kummer's confluent hypergeometric function of the first kind and $K(y)$ is independent of $m$.
                Hence, the convergence is super-exponential.
        \end{remark}

        \begin{remark}
                \Cref{numerics:reflection_error} only concerns the complete-market setting.
                Even in a strongly myopically well-posed model, obtaining an analogous result in the incomplete setting appears very difficult.
                The reason for this is that any quantitative error bound of this type in particular implies uniqueness of the solution to the HJB equation. 
                Hence, an error bound with similarly general assumptions as \cref{numerics:reflection_error} would require a much stronger uniqueness statement than \cref{hjb:uniqueness}.
                The reason that an error estimate can be obtained in the complete market setting is that \cref{hjb:hjb:constant_correlation} is linear in this case, which makes uniqueness much easier than in the general non-linear setting.
        \end{remark}
        
		\section{Global Existence for Second Order Problems on Open Domains}
		\label{section:global_existence}

        In this section, we provide a general method for constructing global solutions to non-linear second-order ordinary differential equations.
		To this end, we extend the theory of sub- and supersolutions for second-order boundary value problems on bounded domains to problems without boundary values on general open domains.
        The main difficulty of this setting is the absence of boundary conditions.

		Consider the equation
		\begin{align}
				\label{global_existence:ode}
				u'' = f(y, u, u')
		.\end{align} 
        Let $E_1$ and $E_2$ be open intervals.
        As we sometimes want to work with solutions that take values in $\partial E_2$ (for example when $E_2 = (0, \infty)$, and we want solutions to be allowed to take the value $0$), we assume that solutions are $\tilde{E}_2$-valued, where $E_2 \subseteq \tilde{E}_2 \subseteq \bar{E}_2$.
        We assume that the right-handside $f: E_1 \times \tilde{E}_2 \times \R \to \R$ is continuous, and that $f$ is locally Lipschitz-continuous on $E_1 \times E_2 \times \R$.
        To simplify the notation, we define for each $C \in (0, \infty]$ the truncated function $f_C: E_1 \times \tilde{E_2} \times \R$ by
        \begin{equation}
                f_C(y, u, v):= f(y, u, -C \vee v \wedge C).
        \end{equation}
        Note that $f = f_\infty$.
        We assume that the family $(f_C)_{C \in (0, \infty]}$ satisfies the following Nagumo condition:
		For any $K \subset E_1 \times \tilde{E}_2$ compact, there exists some function $\varphi: [0, \infty) \to \R$ (depending on $K)$ satisfying the growth-condition
        \[
				\int_{r}^{\infty} \frac{y}{\varphi(y)} \d y = \infty \quad \text{ for all } \quad r > 0
		\]
        such that for all $C \in (0, \infty]$ 
		\[
				|f_C(y, u, v)| \le \varphi(|v|) \quad \text{ for all } (y, u, v) \in K \times \R
		.\]
        
		Note that this condition is e.g. satisfied for $f(y, u, v) = f_0(y, u) + f_1(y, u) v + f_2(y, u) v^2$, with $f_i: E_1 \times \tilde{E}_2 \to \R$ locally Lipschitz-continuous for $i=0,1,2$.

		We will also consider the auxiliary truncated equation 
		\begin{align}
				\label{global_existence:ode_trunc}
				u''  = f_C(y, u, u'), \quad C \in (0,\infty]
		.\end{align}

		Let $[e^{-}_m, e^{+}_m] \uparrow E_1$ be an approximating sequence for $E_1$, where $(e^{-}_m)_{m \in \N}$ is nonincreasing and $(e^{+}_m)_{m \in \N}$ is nondecreasing and $e^-_1 < e^+_1$.
		For functions $\alpha, \beta: E_1 \to \tilde{E}_2$ with $\alpha \le \beta$, we denote the regions enclosed by the graphs of $\alpha, \beta$ on $[e^{-}_m, e^{+}_m]$ and $\tilde{E_2}$ by
		\begin{align*}
				G_m &:= \{(y, u, v) \in [e^{-}_m, e^{+}_m] \times \tilde{E}_2 \times \R: \alpha(y) \le u \le \beta(y) \}, \\
				G &:= \bigcup_{m\in \N} G_m = \{(y, u, v) \in E_1 \times \tilde{E}_2 \times \R: \alpha(y) \le u \le \beta(y) \} 
		.\end{align*}
		The dependence on $\alpha$ and $\beta$ is omitted in the notation as it is clear from the context.

		We start by providing an a priori estimate on the derivative of the solution to \cref{global_existence:ode}.
		\begin{lemma}
				\label{global_existence:nagumo}
				Fix $m \in \N$, $\alpha, \beta\in C([e^{-}_m, e^{+}_m], \tilde{E}_2)$ with $\alpha \le \beta$.
				Then there exists some constant $R_m$ (independent of $C$) such that every solution $u$ of the $C$-truncated equation \eqref{global_existence:ode_trunc} on $[e^{-}_m, e^{+}_m]$ with $\alpha \le u \le \beta$ satisfies \[
						\|u'\|_{L^{\infty}([e^{-}_m, e^{+}_m])} \le R_m
				.\] 
		\end{lemma}
		\begin{proof}
				By the Nagumo condition, there exists some function $\varphi_m: [0, \infty) \to \R$ such that $|f_C(y, u, v)| \le \varphi_m(|v|)$ for all $(y, u, v) \in G_m$ and all $C \in (0, \infty]$.
				Denote \[ 
                        r_m = \max \left\{ \frac{\beta(e^{+}_m) - \alpha(e^{-}_m)}{e^{+}_m - e^{-}_m}, \frac{\beta(e^{-}_m) - \alpha(e^{+}_m)}{e^{+}_m - e^{-}_m}\right\}
                ,\] and choose $R_m > r_m$ large enough that
				\[
						\int_{r_m}^{R_m} \frac{s}{\varphi_m(s)} \d s > \max_{y\in [e^{-}_m, e^{+}_m]} \beta(y) - \min_{y\in [e^{-}_m, e^{+}_m]} \alpha(y)
				.\]
				This is possible since $\int_{r_m}^{\infty} \frac{s}{\varphi_m(s)} \d s = \infty $ and $\alpha, \beta $ are bounded over $[e^{-}_m,e^{+}_m]$ by continuity.
				By \cite[Prop.~4.1]{DeCoster2001}, it now follows that \[
						\|u'\|_{L^{\infty}([e^{-}_m,e^{+}_m])} \le R_m
				. \qedhere \] 
		\end{proof}

		In particular, this implies that any solution to the $C$-truncated equation \eqref{global_existence:ode_trunc} on $[e^{-}_m, e^{+}_m]$ is also a solution to the original equation \eqref{global_existence:ode} on $[e^{-}_m, e^{+}_m]$ provided that $C$ is large enough.
        
		\begin{corollary}
				\label{global_existence:trunc_is_sol}
				If $C\ge R_m$, every solution $u$ to the $C$-truncated equation \eqref{global_existence:ode_trunc} on $[e^{-}_m, e^{+}_m]$ with $\alpha \le u \le \beta$ is a solution to \cref{global_existence:ode}.
		\end{corollary}

		The construction of a solution to the original equation \eqref{global_existence:ode} on $E_1$ now proceeds in three steps: 
		First, we solve a boundary value problem on $[e^{-}_m, e^{+}_m]$. 
		Next, we show that these solutions converge at a reference point as $m\to \infty$.
		Lastly, we show that the initial value problem started at that limit is a global solution.
		
		We define sub- and supersolutions broadly as in \cite[Def.~2.1]{DeCoster2001}, but for simplicity use a slightly less general notion than the one in \cite[Def.~2.1]{DeCoster2001}. Note that this definition still includes the notions used in \cref{section:diffusion}: any subsolution in the sense of \cref{section:diffusion} is a subsolution in the sense of \cref{global_existence:definition:sub_supersolution}.
        For $y \in E_1$ and a sufficiently regular function $\gamma: E_1 \to \tilde{E}_2$, we denote by $D^{-} \gamma(y)$ and $D^{+} \gamma(y)$ the left and right derivative of $\gamma$ at $y$, respectively.

        \begin{definition}
                \label{global_existence:definition:sub_supersolution}
				A function $\alpha \in C(E_1, \tilde {E}_2)$ is called a subsolution to  $u'' = f(y, u, u')$ if for all $y\in E_1$ one of the following two conditions is satisfied
				\begin{itemize}
						\item $\alpha$ is $C^2$ in some neighbourhood of $y$, and $\alpha''(y) \ge f(y, \alpha(y), \alpha'(y))$,
						\item $\alpha$ admits finite left-and right derivatives at $y$ and $D^{-}\alpha(y) < D^{+} \alpha(y)$.
				\end{itemize}
				Analogously, $\beta\in C(E_1, \tilde{E}_2)$ is called a supersolution to the equation $u'' = f(y, u, u')$ if for all $y\in E_1$ one of the following two conditions is satisfied
				\begin{itemize}
						\item $\beta$ is $C^2$ in some neighbourhood of $y$, and $\beta''(y) \le f(y, \beta(y), \beta'(y))$,
						\item $\beta$ admits finite left-and right derivatives at $y$ and $D^{-} \beta(y) > D^{+} \beta(y)$.
				\end{itemize}
		\end{definition}

		As we want to pass to the truncated equation, we also need a boundedness assumption on the derivative of the sub- and supersolution.\footnote{See \cref{global_existence:relaxation} for conditions under which this assumption is not needed.}

        \begin{definition}
				\label{global_existence:assumption}
				We denote \[
						\mathcal{X} = \{\gamma \in C(E_1, \tilde{E}_2): \|\gamma'\|_{L^{\infty}(D(\gamma) \cap K)} < \infty \text{ for all } K \subset E_1 \text{ compact}\} 
				,\] where $D(\gamma) = \{y \in E_1: \gamma \text{ is differentiable at } y\} $.
		\end{definition}

		Note that this condition is satisfied when $\gamma \in C^{1}$, or when $\gamma \in C^{1}$ apart from isolated kinks $y_0$ at which $\lim_{y \uparrow y_0} \gamma'(y), \lim_{y \downarrow y_0} \gamma'(y)$ exist.\footnote{Note that \cref{hjb:existence:abstract,hjb:existence:abstract:general_correlation} hold more generally for $\beta \in \mathcal{X}$, not just $\beta \in C^1$.}
        
		Now, we solve the boundary value problem on $[e^{-}_m, e^{+}_m]$.
    
		\begin{lemma}
				\label{global_existence:compact_existence}
				Let $\alpha, \beta \in \mathcal{X}$ with $\alpha\le \beta$ be sub- and supersolutions to \cref{global_existence:ode}, respectively.
				For any $m\in \N$, there exists a solution $u_m$ to \cref{global_existence:ode} on $[e^{-}_m, e^{+}_m]$ that satisfies $\alpha \le u_m \le \beta$.
		\end{lemma}
		\begin{proof}
				By \cref{global_existence:trunc_is_sol}, it is sufficient to find a solution of the $C$-truncated equation \eqref{global_existence:ode_trunc} for some $C \ge R_m$.
				Fix \[
                        C := \max \{R_m, \|\alpha'\|_{L^{\infty}([e^{-}_m, e^{+}_m] \cap D(\alpha)}, \|\beta'\|_{L^{\infty}([e^{-}_m,e^{+}_m] \cap D(\beta))} \} < \infty
                .\]
				$\alpha$ and $\beta$ are then sub- and supersolutions to the $C$-truncated equation.

				Consider now the truncated problem $u''=f_C(y,u,u')$ endowed with the boundary condition $u(e^{\pm}_m) = \beta(e^{\pm}_m)$.
				Notice that $f_C$ is continuous and bounded over $G_m$.
				By \cite[Thm.~2.6]{DeCoster2001}, there exists a maximal solution $u_m$ (in the sense that any other solution $\tilde{u}_m$ to the BVP with $\alpha \le \tilde{u}_m \le \beta$ satisfies $\tilde{u}_m \le u_m$) over $[e^{-}_m, e^{+}_m]$ that satisfies $\alpha \le u_m \le \beta$ as well as the boundary condition $u(^{\pm}_m) = \beta(e^{\pm}_m)$.
		\end{proof}

		Next, we show that the solutions on bounded intervals converge at a reference point. 
        Note that here and in the subsequent result, we do not require $\alpha$ or $\beta$ to be a sub- or supersolutions, respectively.
		\begin{lemma}
				\label{global_existence:accumulation}
				Fix $y_0 \in [e^{-}_1, e^{+}_1]$ and $\alpha, \beta \in C(E_1, \tilde{E}_2)$ with $\alpha \le \beta$.
				Let $(u_m)_{m\in \N}$ be a sequence of solutions to \cref{global_existence:ode} over the interval $[e^{-}_m, e^{+}_m]$ with $\alpha \le u_m \le \beta$.
				Then the sequence $(u_m(y_0), u_m'(y_0)_{m\in \N})$ has an accumulation point.
		\end{lemma}
		\begin{proof}
				We will show that both components are bounded.
				The existence of an accumulation point then follows from Bolzano-Weierstraß.

				It is clear that $(u_m(y_0))_{m\in \N}$ is bounded as $\alpha(y_0) \le u_m(y_0) \le \beta(y_0)$ by assumption.

				As for the second component, notice that each $u_m$ is in particular also a solution to \cref{global_existence:ode} over $[e^{-}_1, e^{+}_1]$.
				By \cref{global_existence:nagumo}, there exists some $R_1>0$ (independent of $m$) such that \[
						|u_m'(y_0)| \le \|u_m'\|_{L^{\infty}([e^{-}_1, e^{+}_1])} \le R_1
				\] for all $m\in \N$.
		\end{proof}

		Finally, we conclude by showing that the solution of the initial value problem started at the limiting value we obtained is a global solution.
        Note that unlike in \cref{global_existence:compact_existence,global_existence:accumulation}, we require $\alpha, \beta $ to be $E_2$-valued here.
		\begin{lemma}
				\label{global_existence:convergence}
				Fix $\alpha, \beta \in C(E_1, E_2)$ with $\alpha \le \beta$ and $y_0 \in [e^{-}_1, e^{+}_1]$.
				Let $(u_m)_{m\in \N}$ be a sequence of solutions to \cref{global_existence:ode} over the interval $[e^{-}_m, e^{+}_m]$ with $\alpha \le u_m \le \beta$ s.t.\ $(u_m(y_0), u_m'(y_0)) \to (u(y_0), u'(y_0)) $ as $m \to \infty$ for some $(u(y_0), u'(y_0)) \in \R^2$.
                Then the solution $u$ to the initial value problem \eqref{global_existence:ode} started at $((u(y_0), u'(y_0))$ is a global solution to \cref{global_existence:ode} and satisfies $\alpha \le u \le \beta$. 
				Furthermore, $(u_m)_{m\in \N}$ converges to $u$ uniformly on compact sets.
		\end{lemma}
		\begin{proof}
				Let $I \subseteq E_1$ be the maximal domain of existence of $u$.
				Fix a closed interval $J_1 \subset I$ with $y_0 \in J_1^{\mathrm{o}}$, and fix $m\in \N$ such that $J_1 \subset [e^{-}_m, e^{+}_m]$.
				Let $R_m$ be the corresponding constant from \cref{global_existence:nagumo}.
                
				Let $J_2$ be a bounded open interval with $[\min_{J_1} (u \wedge \alpha), \max_{J_1} (u \vee \beta)] \subset J_2 \subset E_2$, and denote $C = R_m \vee \max_{J_1} |u'|$.
				Denote \[
						\tilde{G} = J_1^{\mathrm{o}} \times J_2 \times (-2C, 2C)
				.\] 
				Notice that $u$ is a solution to \cref{global_existence:ode} over the domain $\tilde{G}$.
				For any $k \ge m$, $u_{k}$ satisfies $\|u'\|_{L^{\infty}([e^{-}_m, e^{+}_m])} \le R_m \le C$ by \cref{global_existence:nagumo}, so $u_{k}$ is also a solution to \cref{global_existence:ode} over $\tilde{G}$.

				Since $\tilde{G}$ is relatively compact, $f$ is Lipschitz-continuous over $\tilde{G}$ with Lipschitz constant $L$.
				Fix $y_1 \in J_1$ and let $\varepsilon > 0$ be arbitrary.
				Fix $K$ such that \[
						\max \{|u_{k}(y_0) - u(y_0)|, |u_{k}'(y_0) - u'(y_0)|\} < \varepsilon e^{-L|y_1 - y_0|}
				\] for all $k \ge K$.
				By continuity in the initial data (see e.g.\ \cite[Thm.~1.2.1]{Coddington2012}), we get \begin{align}
						\label{global_existence:convergence:bound}
						\tag{$*$}
						|u(y_1) - u_{k}(y_1)| \le \left(\varepsilon e^{-L|y_1 - y_0|}\right) e^{L|y_1 - y_0|} = \varepsilon
				\end{align} for all $k \ge K$.
				In particular, this implies that $\alpha(y_1) - \varepsilon \le u(y_1) \le \beta(y_1) + \varepsilon$.
				Since $\varepsilon$, $y_1$, and $J_1$ were arbitrary, we have $\alpha \le u \le \beta$ over the entire interval of existence $I$.
                In particular, $u$ can thus never reach $\partial E_2$ in finite time.

                For any compact set $K \subset E_1$, denote by $R_K$ the minimal constant from \cref{global_existence:nagumo}.
                Fix a compact interval $K \subset E_1$ s.t.\ $K \cap I$ has positive length.
                Let $(K_n)_n$ be a sequence of compact intervals with $K_n \subset I$ and $K_n \uparrow K \cap I$.
                Since $\alpha \le u \le \beta$, \cref{global_existence:nagumo} yields \[
                        \|u'\|_{L^\infty(K_n)} \le R_{K_n} 
                \] for all $n \in \N$.
				Since $K_n \uparrow K \cap I$, $K \cap I$ has positive length, and $\alpha$ and $\beta $ are continuous, the sequences $r_n, \max_{K_n} \beta, \min_{K_n} \alpha$ from \cref{global_existence:nagumo} all converge as $n \to \infty$.
				Hence, we can pick a uniform constant $R$ that satisfies \[
						\int_{r_n}^{R} \frac{s}{\varphi(s)} \d s \ge \max_{K_n} \beta - \min_{K_n} \alpha
				\] for all $n \in \N$ (note that we can use a common function $\varphi$, namely $\varphi_K$, for all $n \in \N$), so that $\sup_{n \in \N} R_{K_n} \le R < \infty$.
				Thus, we have $\|u'\|_{L^{\infty}(K \cap I)} \le \sup_n R_{K_n} < \infty$, and so $u'$ can't blow up in the interior of $E_1$.
                
				In total, this shows that $u$ is a global solution.
				The uniform convergence on compact sets follows from \eqref{global_existence:convergence:bound}.
		\end{proof}

        \begin{remark}
				Due to the uniform convergence on compact sets, \cref{global_existence:convergence} is amenable to numerical implementation.
                If $u$ is the unique global solution of \cref{global_existence:ode} that lies between $\alpha$ and $\beta$, there is no need to pass to a subsequence in \cref{global_existence:accumulation}.
		\end{remark}

        \Cref{global_existence:convergence} extends straightforwardly to the case where $\alpha$ and $\beta$ are only semi-continuous.

        \begin{corollary}
				\label{global_existence:convergence:semicontinuous}
				Let $\alpha, \beta: E_1 \to E_2$ be lower and upper semi-continuous, respectively.
                Assume that $\alpha \le \beta$, and fix $y_0 \in [e^{-}_1, e^{+}_1]$.
				Let $(u_m)_{m\in \N}$ be a sequence of solutions to \cref{global_existence:ode} over the interval $[e^{-}_m, e^{+}_m]$ with $\alpha \le u_m \le \beta$ s.t.\ $(u_m(y_0), u_m'(y_0)) \to (u(y_0), u'(y_0)) $ as $m \to \infty$ for some $(u(y_0), u'(y_0)) \in \R^2$.
                Then the solution $u$ to the initial value problem \eqref{global_existence:ode} started at $((u(y_0), u'(y_0))$ is a global solution to \cref{global_existence:ode} and satisfies $\alpha \le u \le \beta$. 
				Furthermore, $(u_m)_{m\in \N}$ converges to $u$ uniformly on compact sets.
		\end{corollary}
        \begin{proof}
                By Baire's theorem for semi-continuous functions, there exist sequences of functions $(\alpha_n)_{n \in \N}, (\beta_n)_{n \in \N} \subset C(E_1, E_2)$ s.t.\ $\alpha_n \uparrow \alpha$ and $\beta_n \downarrow \beta$.
                Note that in particular we have $\alpha_n \le u_m \le \beta_n$ for all $n, m \in \N$.
                By \cref{global_existence:convergence}, $u$ is a global solution to \cref{global_existence:ode} with $\alpha_n \le u \le \beta_n$ for all $n \in \N$.
                Since $\alpha_n \uparrow \alpha$ and $\beta_n \downarrow \beta$ as $n \to \infty$, this means that $u$ is a global solution with $\alpha \le u \le \beta$.
        \end{proof}

		\begin{theorem}
				\label{global_existence:existence}
				Let $\alpha, \beta \in C(E_1, E_2) $ with $\alpha, \beta \in \mathcal{X}$ and $\alpha \le \beta$ be sub- and supersolutions to \cref{global_existence:ode}, respectively. 
				Then there exists a globally defined solution $u$ to \cref{global_existence:ode} that satisfies $\alpha \le u \le \beta$.
		\end{theorem}
		\begin{proof}
				Direct consequence of \cref{global_existence:compact_existence,global_existence:accumulation,global_existence:convergence}.
		\end{proof}

        \begin{remark}
        \label{global_existence:relaxation}
            The restriction $\alpha, \beta \in \mathcal{X}$ is only needed to pass to the truncated equation in \cref{global_existence:compact_existence}. 
            If the dependence of the right-hand side on $u'$ is bounded, i.e.\ $f$ is bounded over $K \times \R$ for all $K \subset E_1 \times \tilde{E}_2$ compact, the truncation is not necessary, and so \cref{global_existence:existence} holds for any $E_2$-valued sub- and supersolution.
            In particular, this is the case when $f$ does not depend on $u'$.
        \end{remark}

        While the maximum of a finite number of subsolutions is again a subsolution, the same is not true for the supremum of a general family of subsolutions.
        However, we can still prove the existence of a global solution to \cref{global_existence:ode} that lies between the supremum and a fixed supersolution in the case of a general family.
        In particular, this enables us to work with subsolutions that are $\tilde{E}_2$-valued, rather than only $E_2$-valued subsolutions as in \cref{global_existence:existence}.

        \begin{corollary}
                \label{global_existence:existence:family}
				Let $\beta \in C(E_1, E_2)$ with $\beta \in \mathcal{X}$ be a supersolution, and let $\{\alpha_i\}_{i \in I} \subset C(E_1, \tilde{E_2})$ be a family of  subsolutions with $\alpha_i \in \mathcal{X}$ and $\alpha_i \le \beta$ for all $i \in I$.
                Assume that $\sup_{i \in I} \alpha_i $ is $E_2$-valued. 
				Then there exists a globally defined $E_2$-valued solution $u$ to \cref{global_existence:ode} that satisfies $\alpha_i \le u \le \beta$ for all $i \in I$.
        \end{corollary}
        \begin{proof}
				For $i \in I$, let $(u_m^{i})_{m\in \N}$ be the sequence of local solutions from \cref{global_existence:compact_existence}. 
				Let $i, j \in I$ be arbitrary, and fix $m\in \N$.
				By \cite[Prop.~2.2]{DeCoster2001}, $\max \{\alpha_i, \alpha_j\} $ is a subsolution to \cref{global_existence:ode} over $[e^{-}_m, e^{+}_m]$.
				Moreover, we have $\max \{\alpha_i, \alpha_j\} \in \mathcal{X}$. 
				Hence, there exists a solution to the $C$-truncated equation  over $[e^{-}_m, e^{+}_m]$ that lies between $\max \{\alpha_i, \alpha_j\} $ and $\beta$ with boundary values $u(e^{\pm}_m) = \beta(e^{\pm}_m)$, and so we have $\alpha_j \le u_m^i$ and $\alpha_i \le u_m^j$ since $u_m^{i}$ and $u_m^{j}$ were defined to be the maximal solutions to the $C$-truncated equation with $\alpha_{k} \le u_m^{k} \le \beta$ for $k \in \{i, j\}$ that satisfy the boundary condition $u(e^{\pm}_m) = \beta(e^{\pm}_m)$.
                Note that while $C$ depends on the subsolution in \cref{global_existence:compact_existence}, we can always choose the maximum of the $C$ arising from $\alpha_i, \alpha_j, \max \{ \alpha_i, \alpha_j \}$.
                Hence, maximality now yields $u_m^j \le u_m^i$ and $u_m^i \le u_m^j$, i.e.\ $u_m^{i}$ and $u_m^{j}$ coincide.
                As $i, j$, and $m$ were arbitrary, all sequences $(u^{i}_m)_{m\in \N}$ coincide.
				Denote the sequence by $(u_m)_{m\in \N}$.
				In particular, we have $\sup_{i\in I} \alpha_i \le u_m \le \beta$ over $[e^{-}_m, e^{+}_m]$ for all $m\in \N$.

                Now, let $(u(y_0), u'(y_0)) $ be the accumulation point from \cref{global_existence:accumulation} with approximating subsequence $(u_{m_k})_{k\in \N}$.
				Let $u$ be the solution of the initial value problem \eqref{global_existence:ode} started at  $(u(y_0), u'(y_0))$.
                Note that $\sup_{i \in I} \alpha_i$ is lower semi-continuous.
				By \cref{global_existence:convergence:semicontinuous}, $u$ is a global solution to \cref{global_existence:ode} with $\sup_{i \in I} \alpha_i \le u \le \beta$.
        \end{proof}
        
		\section{Examples}
		\label{section:examples}

        In this section, we illustrate our results and their economic interpretation by considering four (classes of) examples of a stochastic factor with diffusion dynamics.

        Before considering specific models, we give a decomposition of the optimal strategy and the optimal consumption rate into a myopic and an intertemporal part that will be the focus of the economic interpretation.
        In the following, we denote by $u$ a global positive solution to the HJB equation \eqref{hjb:hjb} and assume that we are in a model where we can verify that the value function and optimal controls are given by \[
                V(x,y) = \frac{x^{1-R}}{1-R} u(y)^{-R}, \qquad \hat{\pi} = \frac{1}{R \sigma} \left(\lambda - R \rho b \frac{u'}{u}\right), \qquad \hat{\xi} = u
        .\]
        This will be the case for all models considered below.

        The optimal risky asset allocation then decomposes into the myopic allocation $\pi_m = \frac{\lambda}{R \sigma}$ and the intertemporal hedging demand $\pi_h = -\rho \frac{b}{\sigma} \frac{u'}{u} $.
        Similarly, the optimal consumption rate decomposes into the myopic rate $\eta$ and the intertemporal adjustment $\xi_h = u - \eta$.
        Importantly, the myopic consumption rate is not admissible when $\eta < 0$ as only non-negative consumption is possible, so the intertemporal consumption adjustment is necessary to even just obtain a valid strategy in models like the Vasicek model.
        
        Notice that the sign of the intertemporal hedging demand is governed by the correlation as well as the change in the optimal consumption rate.
        In particular, the intertemporal hedging demand is zero if the stochastic factor is independent of the stock price (note however, that even in this case the optimal consumption will depend on the stochastic factor).

        From the form of the optimal strategy, one easily sees the asymptotic behaviour of both the optimal consumption rate and the optimal portfolio allocation.
        The intertemporal hedging demand can materially change the risky asset allocation in the empirically relevant regime of the stochastic factor, while disappearing in the limit. 
        Similarly, the intertemporal adjustment for the consumption rate may be large in the empirically relevant regime, but be negligible relative to the myopic consumption rate in the limit.
        
        \begin{corollary}
                \label{intertemporal_hedging:asymptotic:myopicity}

                \begin{enumerate}[(\roman*),ref=\roman*]
                    \item \label{intertemporal_hedging:asymptotic:myopicity:portfolio}
                            If $\frac{b}{\sigma} \frac{u'}{u} \to 0$, then $\pi_h \to 0$ as $y \to \infty$.
                    \item \label{intertemporal_hedging:asymptotic:myopicity:consumption_relative}
                            If $\frac{u}{\eta} \to 1$, then $\frac{\xi_h}{\eta} \to 0$.
                    \item \label{intertemporal_hedging:asymptotic:myopicity:consumption_absolute}
                            If $\eta \Psi \eta - u \to 0$ as $y \to \infty$, then $\lim_{y \to \infty} \xi_h = \lim_{y \to \infty} (\eta \Psi \eta - \eta)$ 
                \end{enumerate}
        \end{corollary}

        Recall that $\eta \Psi \eta$ is an eventual upper bound for $u$ under the assumptions of \cref{hjb:asymptotics:mean_reverting:less_than_Psi_eta}.
        If the bound is asymptotically tight, i.e.\ $\eta \Psi \eta - u \to 0$ as $y \to \infty$,
        we obtain the absolute condition \eqref{intertemporal_hedging:asymptotic:myopicity:consumption_absolute}, which is stronger than the relative condition \eqref{intertemporal_hedging:asymptotic:myopicity:consumption_relative} if $\eta \to \infty$ as $y \to \infty$.
        We will see that in the Kim-Omberg, Heston, and Vasicek models the limiting value of $\xi_h$ can be interpreted as a mean-reversion adjustment.
        
        \subsection{Bounded Coefficients}

		Consider a model with bounded coefficients.
		Assume that the myopic consumption rate $\eta$ satisfies $0 < C_1 \le \eta \le C_2$ for some constants $C_1, C_2 \in \R$, and that the volatility of the stochastic factor $b$ is bounded away from $0$.
        Note that this implies that $E = \R$.

		By \cref{hjb:existence}, there exists a global solution $u$ to \cref{hjb:hjb} with $C_1 \le u \le C_2$.
		Furthermore, for all $(y, u, v) \in \R \times [C_1, C_2] \times \R$, we have \[
				\frac{2}{b^2} \left|\tilde{a} v + \eta u - u^2 - d\frac{v^2}{u}\right| \le C_3 + C_4 |v| + C_5 |v|^2
		\] for some $C_3, C_4, C_5 > 0$.
		For all $m \ge 1$, we have $r_m = \frac{C_2 - C_1}{2m} \le r_1 $.
		Thus, the constant $R_m \equiv R_{1}$ in \cref{global_existence:nagumo} can be chosen independently of $m \ge 1$, and so $u'$ is bounded over $\R$.

		Since the model is strongly myopically well-posed and all model coefficients as well as $\frac{u'}{u}$ are bounded, it follows by \cref{verification} and \cref{verification:sufficient} that $u$ is the optimal consumption rate for the control problem \eqref{setting:control_problem}.

		\subsection{Kim-Omberg Model}
		\label{subsection:mpr}

    	\subsubsection*{Verification}

		Consider a model in which the market price of risk is driven by an Ornstein-Uhlenbeck process:
		\begin{align*}
				dS_t &= S_t((r + \sigma Y_t) + \sigma dW_t), \\
				dY_t &= -\kappa (Y_t - \theta) dt + \nu d\tilde{W}_t
		.\end{align*}
		Here, $\sigma, \delta, \nu > 0$, $r, \kappa \ge 0$, $\theta\in \R$, and $\rho \in [-1,1]$ are constants.
		Assume that $R > 1$.
		We have \[
				\eta(y) = \frac{1}{R} \left( \delta - (1-R) \left( r + \frac{y^2}{2R} \right)  \right) \ge \frac{\delta}{R} > 0
		,\]
		so the model is strongly myopically well-posed.
		In this model, $\Psi \eta$ is a rational function that is bounded from above and converges to $1$ as $|y| \to \infty$ since the denominators are bounded away from $0$ and of higher degree than the numerators.
		One easily sees that the other conditions of \cref{hjb:asymptotics:assumption} are also satisfied.

		\Cref{hjb:existence:strongly_wellposed} implies that there exists a global positive solution $u$ to the HJB equation \eqref{hjb:hjb}, and that it satisfies $C_1 \eta \le u \le C_2 \eta$ for all $|y|$ large enough.
		By \cref{hjb:asymptotics}, $\frac{u}{\eta} \to 1$ as $|y| \to \infty$, and by \cref{hjb:asymptotics:derivative}, $\frac{u'}{u}$ grows at most linearly.
		By \cref{hjb:uniqueness}, $u$ is the unique solution with $\frac{u}{\eta} \to 1$ as $|y| \to \infty$.

        By \cref{verification,verification:sufficient}, we obtain that $u$ is indeed the optimal consumption rate of the optimal control problem \eqref{setting:control_problem}.

		Moreover, if $\kappa > \frac{1-R}{R} \rho \nu$, the model is mean-reverting under the minimal distortion measure and we have $\Psi \eta \le 1$ for all $|y|$ large enough, so that $u \le \eta \Psi \eta \le \eta$ eventually as $|y| \to \infty$ by \cref{hjb:asymptotics:mean_reverting,hjb:asymptotics:mean_reverting:less_than_Psi_eta}. 

        \subsubsection*{Numerical Approximation}

        We approximate the solution to the HJB equation using the numerical scheme from \cref{section:numerics}.
        We use grid size $10^{-4}$; the model parameters can be found in \cref{mpr:table:parameters}.

        \Cref{mpr:fig:convergence:convergence} shows the convergence of the solutions to the discretised problem over $[-m, m]$ to the global solution.
        The solutions are visually indistinguishable from $m = 2$ onwards.
        From \cref{mpr:fig:convergence_rate}, one observes that the maximal difference of two consecutive solutions over the interval $[-1, 1]$ decreases extremely rapidly, which matches the super-exponential decay one might expect from the complete-market analysis in \cref{numerics:reflection_error:ou}.

        \begin{figure}[htb]
            \centering
        
            \begin{subfigure}[t]{0.45\textwidth}
                \centering
                \includegraphics[width=0.95\linewidth]{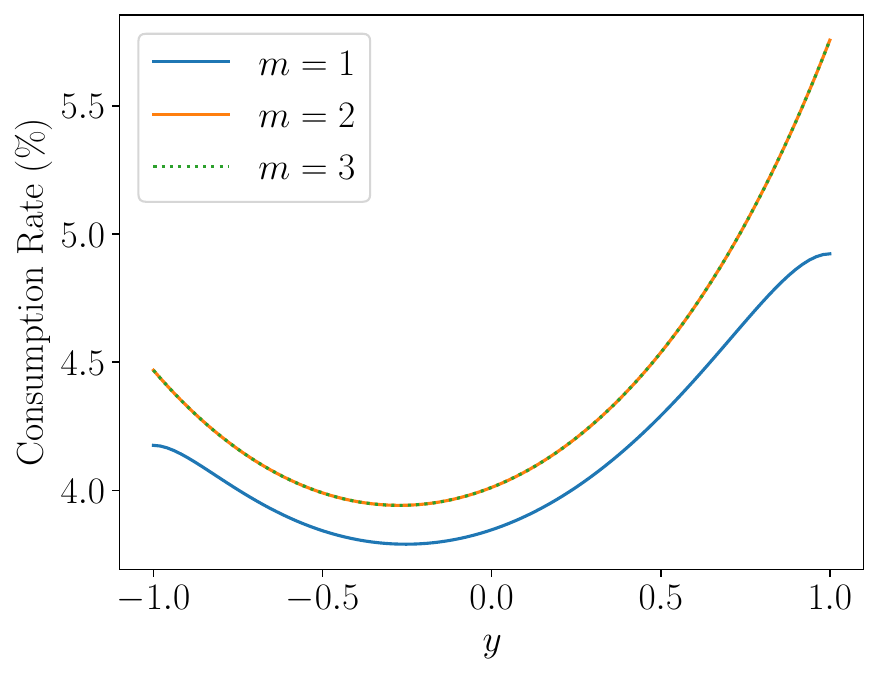}
                \caption{Convergence of solutions.}
                \label{mpr:fig:convergence:convergence}
            \end{subfigure}
            \hfill
            \begin{subfigure}[t]{0.45\textwidth}
                \centering
                \includegraphics[width=0.95\linewidth]{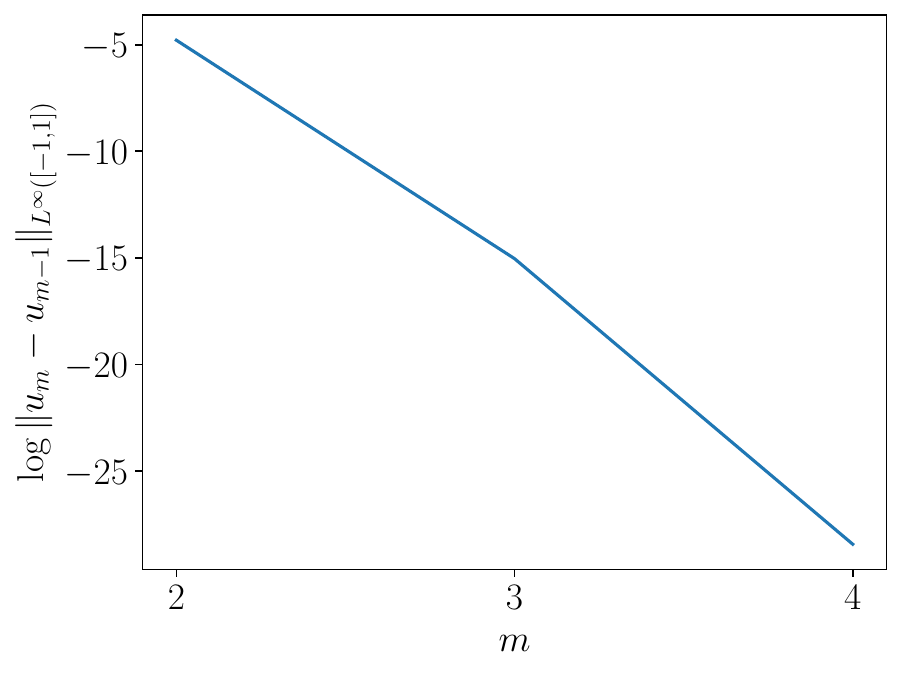}
                \caption{Difference of consecutive solutions.}
                \label{mpr:fig:convergence_rate}
            \end{subfigure}
        
            \caption{Numerical approximation (Kim-Omberg).}
        \end{figure}

		\begin{table}[htb]
				\centering
				\begin{tabular}{cccccccc}
						\toprule
						$R$   & $\delta$ & $r$    & $\sigma$ & $\kappa$ & $\theta$ & $\nu$ & $\rho$ \\
						\midrule
						$1.5$ & $0.02$   & $0.0168$ & $0.151$    & $0.271$    & $0.271$    & $0.227$ & $-0.935$ \\
						\bottomrule
				\end{tabular}
				\caption{Parameters of the Kim-Omberg model, adapted from \cite{MuhleKarbe2017,Wachter2002}.}
				\label{mpr:table:parameters}
		\end{table}
		
		\medskip{}
        
		\subsubsection*{Asymptotic Behaviour of Intertemporal Hedging and Consumption}

        In the Kim-Omberg model, the intertemporal hedge is given by \[
                \pi_h = - \rho \frac{\nu}{\sigma} \frac{u'}{u}
        .\]
        While we only proved above that $\frac{u'}{u}$ grows at most linearly (since that was sufficient for verification), the numerics suggest that we in fact have $\frac{u'}{u} \to 0$ as $|y| \to \infty$.
        This means that the intertemporal hedging demand asymptotically vanishes: $\pi_h \to 0$ as $|y| \to \infty$.
        
        \Cref{mpr_ou:fig:portfolio_limit} shows that the intertemporal hedge is maximal when the absolute market price of risk is about $4$, after which the asymptotically myopic regime kicks in.

        Above, we showed that $\frac{u}{\eta} \to 1$ as $|y| \to \infty$.
        Hence, we have $\frac{\xi_h}{\eta} \to 0$ as $|y| \to \infty$, i.e.\ the intertemporal consumption adjustment is asymptotically negligible relative to the myopic consumption rate.

        Moreover, the numerics indicate that the upper bound $u \le \eta \Psi \eta$ is asymptotically tight, i.e.\ $u - \eta \Psi \eta \to 0$ as $|y| \to \infty$.
        Since one can compute that $\eta \Psi \eta - \eta \to - 2 \tilde{\kappa} = - 2 (\kappa + \frac{R-1}{R} \rho \nu)$ as $|y| \to \infty$, we expect that $\xi_h = u - \eta \to -2 \tilde{\kappa}$ as $|y| \to \infty$.
        Indeed, this is exactly what we see in \cref{mpr_ou:fig:consumption_limit}.
        Thus, the intertemporal consumption adjustment $\xi_h$ becomes a pure mean-reversion adjustment in the limit.
        Note that the mean-reversion adjustment constant is the same for positive and negative market price of risk.
        Hence, the agent asymptotically consumes a constant amount $2 \tilde{\kappa}$ less under the optimal strategy than under the myopic strategy as they account for the expected mean-reversion of the market price of risk.

        \begin{figure}[htb]
				\centering

				\begin{subfigure}[t]{.45\textwidth}
						\includegraphics[width=0.95\textwidth]{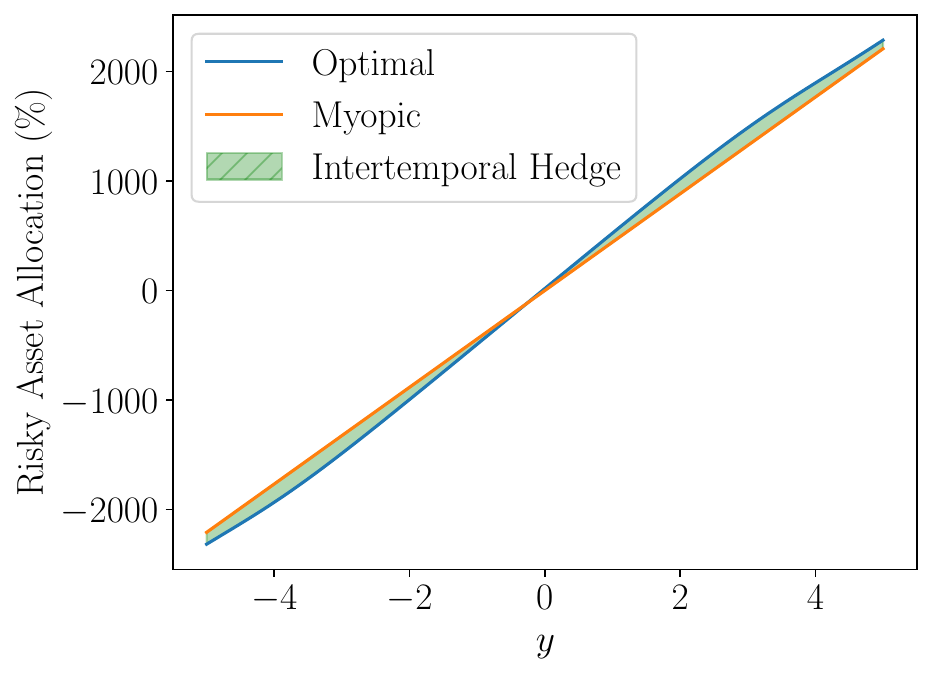}
						\caption{Risky asset allocation.}
						\label{mpr_ou:fig:portfolio_limit}		
				\end{subfigure}
                \hfill
                \begin{subfigure}[t]{.45\textwidth}
						\includegraphics[width=0.95\textwidth]{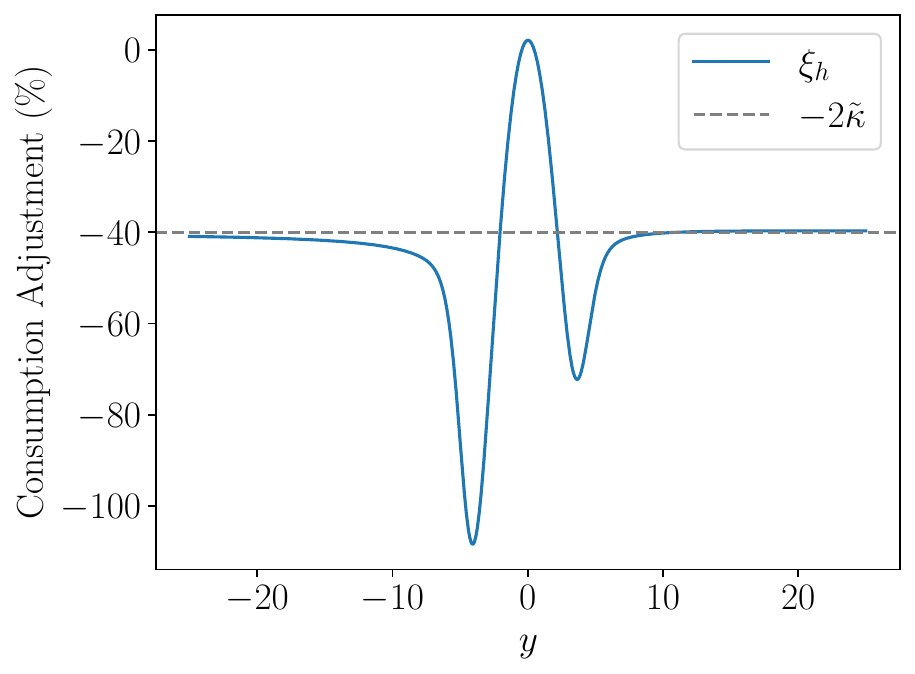}
						\caption{Consumption rate.}
						\label{mpr_ou:fig:consumption_limit}		
				\end{subfigure}
                \caption{Asymptotic regime (Kim-Omberg).}
		\end{figure}
		
		\medskip{}
        
		\subsubsection*{Economic Interpretation}

        The previous analysis showed that the intertemporal adjustments of the investment and consumption strategies are negligible if the market price of risk is large.
        This is not the case in the economically relevant parameter regime.
        
        \Cref{mpr_ou:fig:portfolio} shows the optimal and myopic portfolio as well as the intertemporal hedge over the 95\% coverage interval of the stationary distribution of the stochastic factor.
        The first observation is that the myopic portfolio $\pi_m(y)  = \frac{1}{R \sigma} y$ is increasing in $y$.
        For the parameters from \cref{mpr:table:parameters}, this is inherited by the optimal portfolio as the myopic component dominates the intertemporal hedge.
        However, this is not the case for extreme parameter combinations (for example $R=10$, $\nu = 1$ with the other parameters as in \cref{mpr:table:parameters}), in general the optimal portfolio might not be monotone in $y$.
        Due to the strong negative correlation between asset price and market price of risk, the optimal portfolio allocates a higher fraction of wealth into the risky asset than the myopic one.
        In particular, notice that the intertemporal hedge is close to $18\%$ when the market price of risk is zero. 
        Hence, the agent will hold a position in the risky asset purely for hedging purposes despite the market price of risk being zero.

        \Cref{mpr_ou:fig:consumption} shows the optimal and myopic consumption rates.
        The optimal consumption rate (and hence the value function) is larger compared to the myopic rate when the (absolute) market price of risk is small, but smaller when it is large. 
        This is due to the mean reversion of $Y$; accounting for the expected future mean-reversion leads to less variation in the optimal consumption rate compared to the myopic rate.
        
        \begin{figure}[htb]
				\centering

                \begin{subfigure}[t]{.45\textwidth}
						\includegraphics[width=0.95\textwidth]{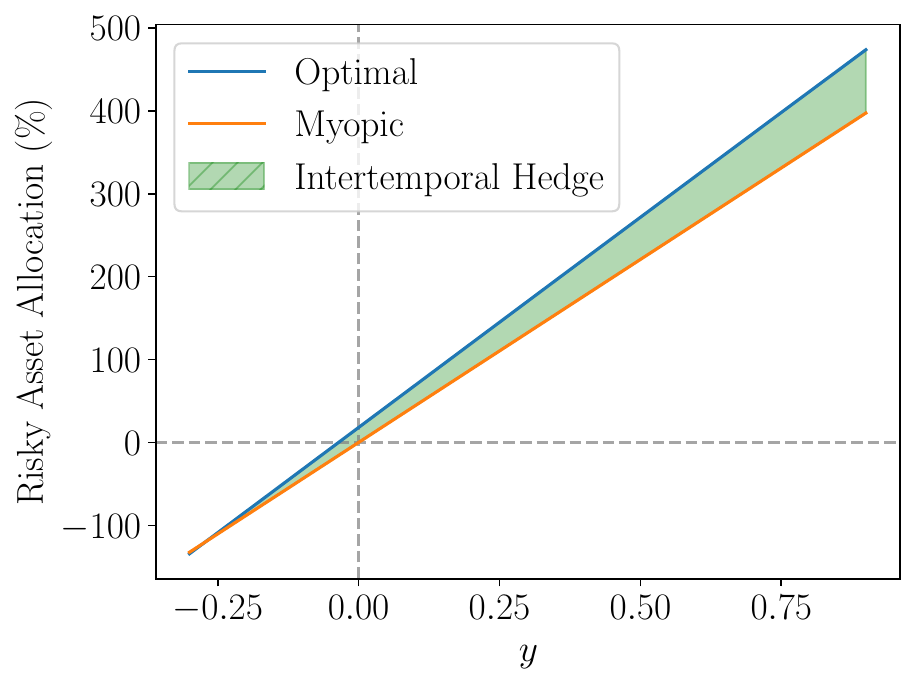}
						\caption{Risky asset allocation.}
						\label{mpr_ou:fig:portfolio}		
				\end{subfigure}
                \hfill
				\begin{subfigure}[t]{.45\textwidth}
						\includegraphics[width=0.95\textwidth]{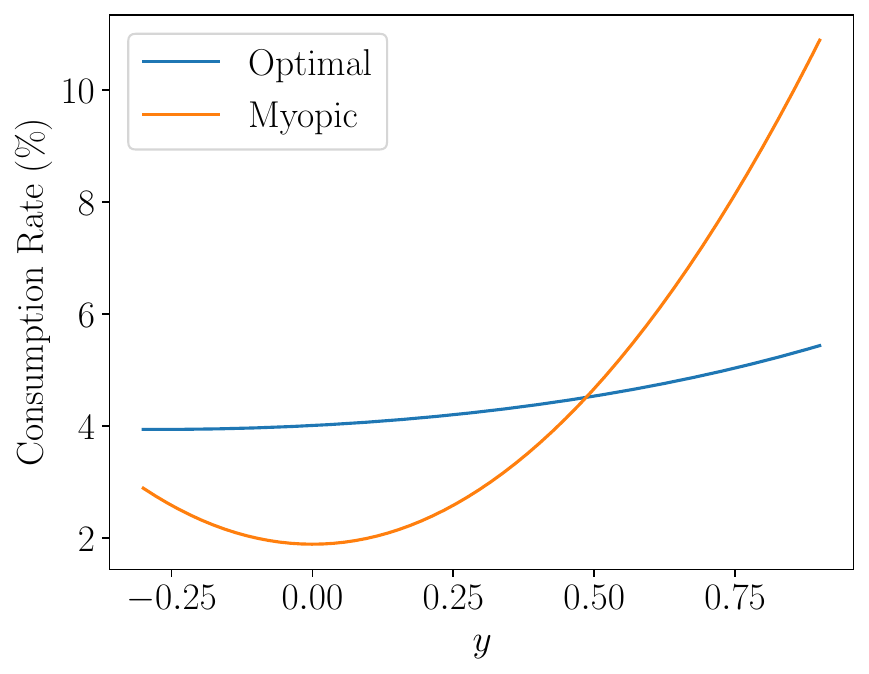}
						\caption{Consumption rate.}
						\label{mpr_ou:fig:consumption}		
				\end{subfigure}

                \caption{Economically relevant regime (Kim-Omberg).}
		\end{figure}

		\subsection{Heston Model}

        \subsubsection*{Verification}

		Consider the Heston model 
		\begin{align*}
				dS_t &= S_t ((r + \lambda Y_t) dt + \sqrt{Y_t} dW_t), \\
				dY_t &= -\kappa (Y_t - \theta) dt + \nu \sqrt{Y_t} d\tilde{W}_t
		,\end{align*} 
		i.e.\ $\sigma(y) = \sqrt{y}, \lambda(y) = \lambda \sqrt{y} $.
		Assume that the constants satisfy $R>1$, $\delta, r, \theta, \nu > 0$, $\kappa \ge 0$, and $\rho \in [-1, 1]$.
		Furthermore, assume that the Feller condition $\kappa \theta \ge \frac{\nu^2}{2}$ is satisfied, so $Y$ stays positive, i.e., E=$(0,\infty)$.
		In this model, we have \[
				\eta(y) = \frac{1}{R} \left(\delta - (1 - R) \left(r + \frac{\lambda^2 y}{2R}\right)\right) \ge \frac{\delta + (R-1) r}{R}
		\]

        Suppose first that $\delta>0$. Then $\eta(y) > \delta/R>0$ so that
		the model is strongly myopically well-posed. 
		$\Psi \eta$ is a rational function that is bounded from above and converges to $1$.
		One can easily check that \cref{hjb:asymptotics:assumption} is satisfied in this model.

		By \cref{hjb:existence:strongly_wellposed} there exists a global positive solution $u$ to the HJB equation \eqref{hjb:hjb} that satisfies $C_1 \eta \le u \le C_2 \eta$ for $y$ large enough.
		Hence, \cref{hjb:asymptotics} shows that $u$ satisfies $\frac{u}{\eta} \to 1$ as $y\to \infty$.

		By \cref{heston:limit_at_0} below, it follows that $\frac{u'}{u} = -\frac{R}{\varphi} \frac{v'}{v}$ is bounded around $0$, where $v = u^{-\frac{R}{\varphi}}$.
		Furthermore, $\frac{u'}{u} \to 0$ as $y\to \infty$ by \cref{hjb:asymptotics:derivative,hjb:asymptotics:derivative:vanishes} (or by \cref{hjb:asymptotics:derivative:mean_reverting} if $\kappa > \frac{1-R}{R} \rho \lambda \nu$, i.e.\ the stochastic factor is mean-reverting).
		By continuity, $\frac{u'}{u}$ is then bounded over $(0,\infty)$.
		It follows from \cref{hjb:uniqueness} there is no other solution $\tilde{u}$ to the HJB equation with $\tilde{u} \to u(0)$ as $y\to 0$ and $\frac{\tilde{u}}{\eta} \to 1$ as $y \to \infty$.

        Now, we will verify that the SDE \[
                dZ_t = \left(-\kappa(Z_t-\theta) + \frac{1-R}{R} \rho \lambda \nu Z_t - \nu^2 ((1-\rho^2) R + \rho^2) \frac{u'(Z_t)}{u(Z_t)} Z_t\right) dt + \nu \sqrt{Z_t} dW_t 
        \] has a unique global positive solution.
        Note that since $\frac{u'}{u}$ is bounded, this is an SDE of the form \[
                dZ_t = (\kappa \theta + f(Z_t) Z_t) dt + \nu \sqrt{Z_t} dW_t
        ,\] where $f$ is some bounded continuous function.
        By the Engelbert-Schmidt theorem (see e.g.\ \textcite[Thm.~5.15]{Karatzas1998}), the SDE has a unique weak solution up to an exit time.
        Let \[
                dZ^i_t = (\kappa \theta + m_i Z^i_t) dt + \nu \sqrt{Z^i_t} dW_t, \quad i=1,2
        \] be CIR processes with $m_1 \le f \le m_2$ for some $m_1, m_2 \in \R$.
        Since we assumed that the Feller condition $\kappa \theta \ge \frac{\nu^2}{2}$ is satisfied, both $Z^i$ are strictly positive.
        By the comparison theorem for SDEs (see e.g.\ \textcite[Prop.~2.18]{Karatzas1998}), we have $Z^1 \le Z \le Z^2$, so that $Z$ is a global positive solution.
        Hence, condition \eqref{verification:sde} in \cref{verification} is satisfied.

        Since the model is strongly myopically well-posed, condition \eqref{verification:infinite_consumption} is trivially satisfied.
        In total, \cref{verification} shows that $u$ is indeed the optimal consumption fraction for the control problem \eqref{setting:control_problem}.

        Moreover, if $\kappa > \frac{1-R}{R} \rho \lambda \nu$, the stochastic factor is mean-reverting under the minimal distortion measure and $\Psi \eta$ converges to $1$ from below.
        \Cref{hjb:asymptotics:mean_reverting,hjb:asymptotics:mean_reverting:less_than_Psi_eta} yield that in this case $u \le \eta \Psi \eta \le \eta$ eventually as $y \to \infty$.

		Now, we want to extend this to the case $\delta < 0$.\footnote{
            In the formulation of problem \eqref{setting:control_problem}, we discount utility, and so it is natural to pick $\delta \ge 0$.
            One could also formulate the problem differently and move the discount factor into the utility function to directly discount consumption.
            The two approaches are equivalent up to multiplying the discount rate by $(1-R)$.
            When $R \in (0, 1)$, the discount rate of utility and the equivalent discount rate of consumption have the same sign, and it is enough to consider $\delta \ge 0$.
            But when $R > 1$, the signs differ, and a negative discount rate of utility corresponds to a positive discount rate of consumption, so the case $\delta < 0$ should be considered as well.
            We refer to \cite[Appendix~D]{Herdegen2021} for a more thorough discussion on the role of $\delta$.
        }
		In this case, the model is no longer necessarily strongly myopically well-posed, or even myopically well-posed everywhere.
		Let $\delta_+$ be a discount rate such that a solution $u_+ > 0$ to \cref{hjb:hjb} for $\delta_+$ exists that is bounded and bounded away from $0$ at $0$ and satisfies $\frac{u_+}{\eta} \to 1$ as $y\to \infty$.
		In particular, any $\delta_+ > 0$ satisfies this.

		Let $\delta \in (\delta_+ - R \inf_{E} u_+, \delta_+)$.
        Note that if $\delta_+ > 0$, then $u_+ \ge \frac{\delta_+ + (R-1) r}{R} > \frac{\delta_+}{R}$ since $\eta_+ \ge \frac{\delta_+ + (R-1) r}{R}$, so the left boundary of the interval is always negative.
		Since the myopic consumption rate $\eta$ is monotone in $\delta$, $u_+$ is a supersolution to \cref{hjb:hjb} in the model with discount rate $\delta$.
		Setting $K = 1 - \frac{\delta_+ - \delta}{R \inf_E u_+}$, we have $K\in (0,1)$.
        Denoting by $\eta_+$ the myopic consumption rate in the model with discount rate $\delta_+$, we have \begin{multline*}
            \frac{1}{2} b^2 K u_+'' + \tilde{a} K u_+' + \eta K u_+ - K^2 u_+^2 - d K \frac{(u_+')^2}{u_+} =
            K (- \eta_+ u_+ + u_+^2 + \eta u_+ - Ku_+^2) =\\=
            K u_+ ((1-K) u_+ + \frac{1}{R} (\delta - \delta_+)) \ge
            0
        ,\end{multline*}
        so $Ku_+$ is a subsolution to \cref{hjb:hjb} with discount rate $\delta$.
		By \cref{global_existence:existence}, there exists a solution $0 < Ku_+ \le u \le u_+$ to \cref{hjb:hjb}.
		By the assumptions on $u_+$, $u$ is bounded and bounded away from $0$ around $0$ and satisfies $\frac{u}{\eta}\to 1$ by \cref{hjb:asymptotics}.
		As in the case $\delta > 0$, this together with \cref{hjb:asymptotics:derivative:vanishes,heston:limit_at_0,verification} implies that $u$ is indeed the optimal consumption fraction.
		
		Note that this process can be iterated by setting $\delta_+ = \delta$ to get existence for further values of $\delta$.
		For fixed $\delta$, a similar procedure can be used to extend the existence to some negative interest rates $r < 0$.

        \begin{lemma}
				\label{heston:limit_at_0}
				Consider the Heston model with parameters as above.
				Let $v$ be a positive solution to \cref{hjb:hjb:constant_correlation} over $(0, \infty)$ that is bounded and bounded away from $0$ around $0$.
				Then $v(0) = \lim_{y \downarrow 0} v(y) $ and $v'(0) = \lim_{y \downarrow 0} v'(y)$ exist and satisfy \[
						v'(0) = \frac{\frac{R}{\varphi} \left(\eta(0) v(0) - v(0)^{1 - \frac{\varphi}{R}}\right)}{\kappa \theta}
				.\] 
		\end{lemma}
		\begin{proof}
				Denote $h(y) = \frac{\frac{R}{\varphi}}{\frac{1}{2} \nu^2} \left(\eta(y) v(y) - v(y)^{1 - \frac{\varphi}{R}}\right)$, $a_0 = \frac{\kappa \theta}{\frac{1}{2} \nu^2}$, and $a_1 = \frac{-\kappa + \frac{1-R}{R} \rho \lambda \nu}{\frac{1}{2} \nu^2}$.
				Notice that $a_0 \ge 1$ as the model satisfies the Feller condition $\kappa \theta \ge \frac{1}{2} \nu^2$.
				$v'$ satisfies the linear first-order ODE \[
						y w' + (a_0 + a_1 y) w = h(y)
				.\] 	
				The general solution to this equation is given by \[
						w_C(y) = \frac{C + F(y)}{y^{a_0} e^{a_1 y}}, \quad F(y) = \int_{0}^{y} h(s) s^{a_0-1} e^{a_1 s} \d s, \quad C \in \R
				.\]  
				Note that $F$ is well-defined since $h$ is bounded around $0$.
				For all $C \neq 0$, $w_C \sim \frac{C}{y^{a_0}}$ as $y\to 0$.
				Since $v$ is bounded around $0$, $v'$ is integrable around $0$.
				Since $a_0 \ge 1$, this means that $v' = w_0$.
				
				Now, denote $\munderbar{h}(0) = \liminf_{y \downarrow 0} h(y), \bar{h}(0) = \limsup_{y \downarrow 0} h(y)$.
				By L'Hospital's rule, \begin{multline*}
						\frac{\munderbar{h}(0)}{a_0} = \liminf_{y \downarrow 0} \frac{h(y) y^{a_0-1} e^{a_1 y}}{a_0 y^{a_0-1} e^{a_1 y} + a_1 y^{a_0} e^{a_1 y}} \le \liminf_{y \downarrow 0} w_0(y) \le \\ \le \limsup_{y \downarrow 0} w_0(y) \le \limsup_{y \downarrow 0} \frac{h(y) y^{a_0-1} e^{a_1 y}}{a_0 y^{a_0-1} e^{a_1 y} + a_1 y^{a_0} e^{a_1 y}} = \frac{\bar{h}(0)}{a_0}
				,\end{multline*} 
				so $v' = w_0$ is bounded around $0$.

				Since $v'$ is bounded around $0$, $v$ is Lipschitz-continuous around $0$. 
				Hence, $v$ can be continuously extended to $0$. 
				Thus, $\munderbar{h}(0) = \bar{h}(0)$, and so $\lim_{y \downarrow 0} v'(y) = \frac{h(0)}{a_0}$.
		\end{proof}

        \subsubsection*{Numerical Approximation}

        We approximate the solution to the HJB equation using the scheme from \cref{section:numerics} with grid size $10^{-4}$ and model parameters from \cref{heston:table:parameters}.

        \Cref{heston:fig:convergence} shows the convergence of the approximations over $[\frac{1}{m}, \sqrt{m}]$ to the true solution. 
		\Cref{heston:fig:a_priori} shows the approximation of the consumption rate with $m = 10\,000$ together with the tight a priori bounds obtained by \textcite{Guasoni2020}.

		\begin{figure}[htb]
				\centering

				\begin{subfigure}[t]{.45\textwidth}
						\includegraphics[width=0.95\textwidth]{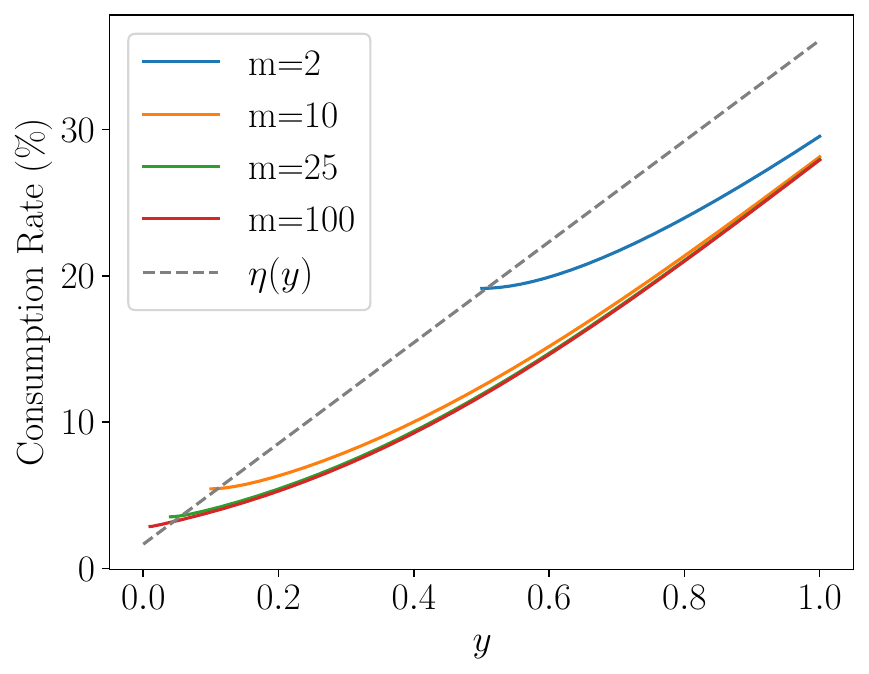}
						\caption{Convergence of solutions.}
						\label{heston:fig:convergence}		
				\end{subfigure}
                \hfill
				\begin{subfigure}[t]{.45\textwidth}
						\includegraphics[width=0.95\textwidth]{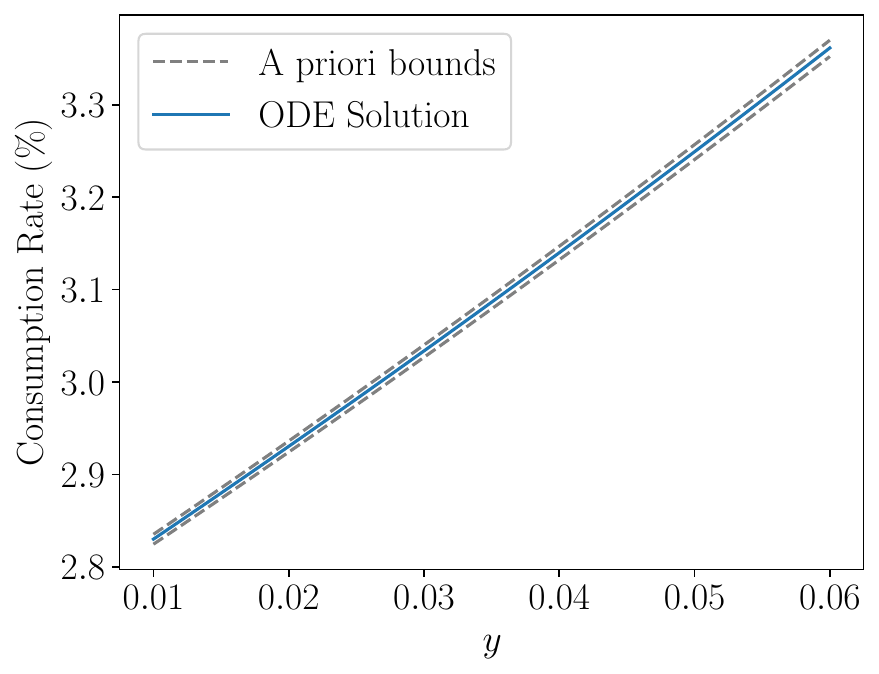}
						\caption{Comparison with a priori bounds.}
						\label{heston:fig:a_priori}		
				\end{subfigure}

                \caption{Numerical approximation (Heston).}
		\end{figure}

		\begin{table}[htb]
				\centering
				\begin{tabular}{cccccccc}
						\toprule
						$R$   & $\delta$ & $r$    & $\lambda$ & $\kappa$ & $\theta$ & $\nu$ & $\rho$ \\
						\midrule
						$2$ & $0.02$   & $0.013$ & $1.66$    & $0.088$    & $0.035$    & $0.031$ & $-0.84$ \\
						\bottomrule
				\end{tabular}
				\caption{Parameters of the Heston model, taken from \cite{Guasoni2020}.}
				\label{heston:table:parameters}
		\end{table}

        \subsubsection*{Asymptotic Behaviour of Intertemporal Hedging and Consumption}

        In the Heston model, the intertemporal hedge reads \[
                \pi_h = - \rho \nu \frac{u'}{u}
        .\]

        Above, we showed that $\frac{u'}{u} \to 0$ as $y \to \infty$. 
        Hence, the optimal portfolio is asymptotically myopic: $\pi_h \to 0$ as $y \to \infty$.
        From the numerics, we see that $u$ is increasing (for all parameter combinations), so $\frac{u'}{u}$ is positive.
        The positivity of $u'$ was proved in \cite[Thm.~5.3]{AbiJaber2026}.
        Hence, the sign of the intertemporal hedging demand is governed completely by the correlation; with negative correlation the intertemporal hedging demand is always positive.
        In \cref{heston:fig:portfolio_limit}, it is maximal when the variance is small and then decreases to the asymptotically myopic limit.

        For the consumption rate, we observe from \cref{heston:fig:consumption_limit} that $\xi_h \to -\tilde{\kappa} = - \left( \kappa + \frac{R-1}{R} \rho \lambda \nu \right)$ as $y \to \infty$, which is what we expect since $\eta \Psi \eta - \eta \to -\tilde{\kappa}$ as $y \to \infty$.
        Hence, the intertemporal adjustment to the consumption rate is asymptotically a pure mean-reversion adjustment. 
        
        \begin{figure}[htb]
				\centering

				\begin{subfigure}[t]{.45\textwidth}
						\includegraphics[width=0.95\textwidth]{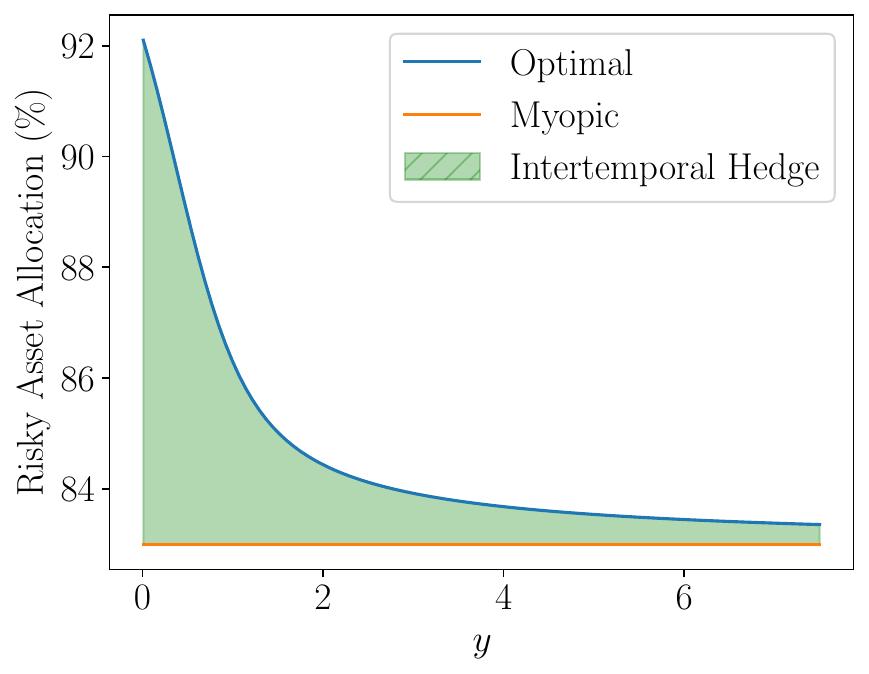}
						\caption{Risky asset allocation.}
						\label{heston:fig:portfolio_limit}		
				\end{subfigure}
                \hfill
                \begin{subfigure}[t]{.45\textwidth}
						\includegraphics[width=0.95\textwidth]{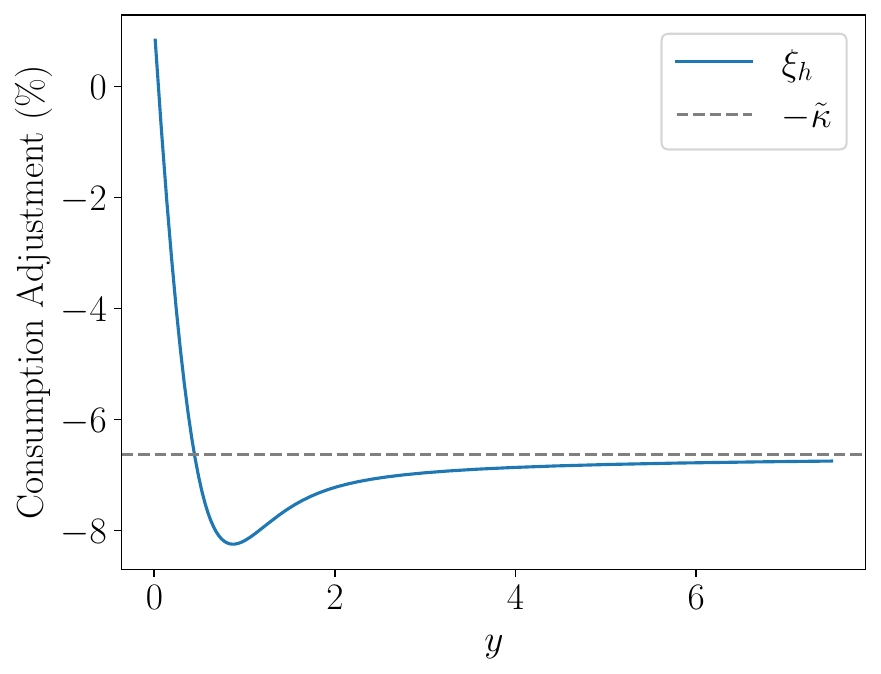}
						\caption{Consumption rate.}
						\label{heston:fig:consumption_limit}		
				\end{subfigure}

                \caption{Asymptotic regime (Heston).}
		\end{figure}

        \subsubsection*{Economic Interpretation}

        Now, we focus on the economically relevant parameter regime.

        \Cref{heston:fig:portfolio} shows that the optimal as well as myopic fraction of wealth to invest into the risky asset and the intertemporal hedge are all approximately constant over the 95\% coverage interval of the stationary distribution of the stochastic factor.
        Due to the negative correlation between stock price and variance, the intertemporal hedge is positive.
        \Cref{heston:fig:consumption} shows the optimal and myopic consumption rates.
        Both are approximately linear, but the optimal curve has a smaller slope.
        As the agent accounts for the expected mean reversion, the optimal consumption rate is higher than the myopic consumption rate when the variance is low, and lower when the variance is high.
        
        \begin{figure}[htb]
				\centering

                \begin{subfigure}[t]{.45\textwidth}
						\includegraphics[width=0.95\textwidth]{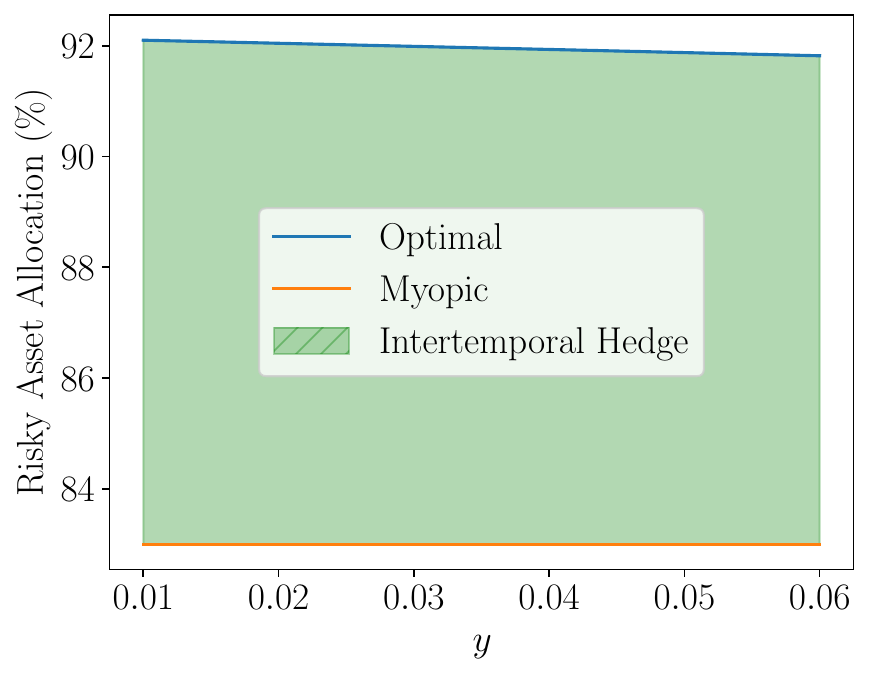}
						\caption{Risky asset allocation.}
						\label{heston:fig:portfolio}		
				\end{subfigure}
                \hfill
                \begin{subfigure}[t]{.45\textwidth}
						\includegraphics[width=0.95\textwidth]{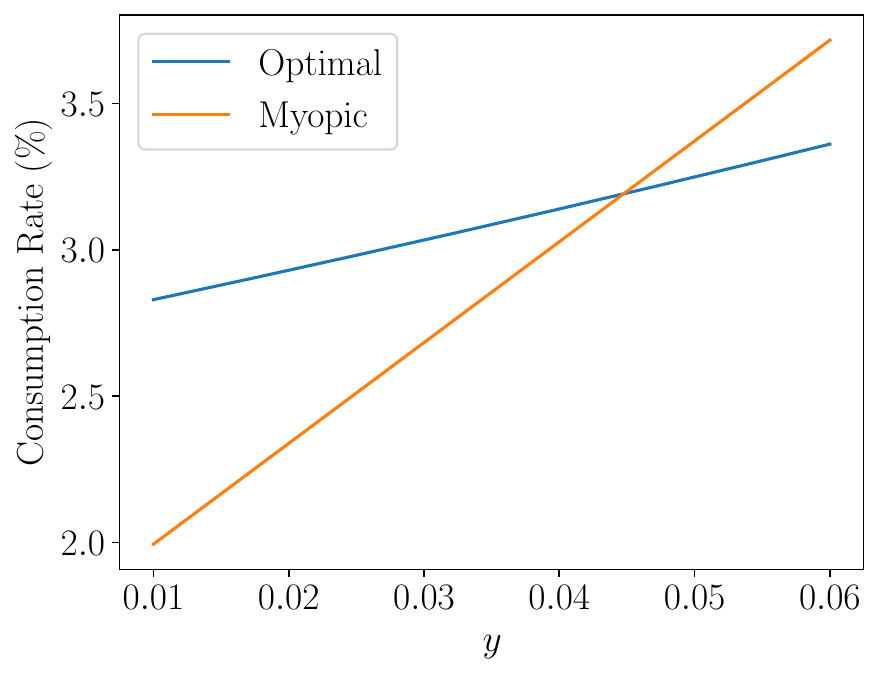}
						\caption{Consumption rate.}
						\label{heston:fig:consumption}		
				\end{subfigure}
                
                \caption{Economically relevant regime (Heston).}
		\end{figure}

		\subsection{Vasicek Model}
    
        \subsubsection*{Verification}

		Finally, we demonstrate that our results can also be applied to models where $\eta(y) < 0$ for some $y$, i.e., models that are not myopically well-posed everywhere.
		To this end, we consider the stochastic interest rate model
		\begin{align*}
				dS_t &= S_t ((Y_t + \lambda \sigma) dt + \sigma dW_t), \\
				dY_t &= -\kappa (Y_t - \theta) dt + \nu d\tilde{W}_t
		\end{align*} with risk aversion $R > 1$.
		Here, $\delta, \sigma, \nu > 0$, $\kappa > 0$, $\lambda, \theta \in \R$, and $\rho \in [-1, 1]$ are constants.
		Furthermore, assume that 
        \begin{align*}
				\frac{1}{2} \nu^2 q < C_1 :&= 
                \kappa \theta + \frac{1-R}{R} \rho \lambda \nu + \frac{1}{Rq} \left(\delta - (1-R) \frac{\lambda^2}{2R}\right) - \frac{1}{2}\nu^2 \left((1-\rho^2) R + \rho^2\right) q \\&= 
                \tilde{a}(0) + \frac{\eta(0)}{q} - \frac{1}{2}\nu^2 \left((1-\rho^2) R + \rho^2\right) q
		,\end{align*} where $q = \frac{R-1}{R \kappa}$.
        Note that this condition is slightly stronger than condition (13) in \textcite{Guasoni2019}, which is equivalent to $C_1 > 0$.

		In this model, \[
				\eta(y) = \frac{1}{R} \left(\delta - (1-R) \left(y + \frac{\lambda^2}{2R}\right)\right) < 0 \text{ for } y < \frac{\delta}{1-R} - \frac{\lambda^2}{2R} =: y^{*}
		.\]
        Note that $y^*<0$ and hence $\eta(0) > 0$.
		Since $\eta(y) < 0$ for some $y$, we cannot apply \cref{hjb:existence} to generate a subsolution to \cref{hjb:hjb}.
		Still, we can use it to construct a supersolution.
		Note that in light of \cref{hjb:existence:abstract} it is not surprising that it is easier to construct a supersolution than a subsolution.
		Set $g_2(y) = 1 - \frac{1-R}{R} \log(1 + \exp(y - y^{*}))$. 
		Then
		\[
				g_2(y) \ge 1 - \frac{1-R}{R} ((y - y^{*}) \vee 0) = 1 + \eta(y)^{+} \ge \eta(y)
		\]
		and $\tilde{C}_2 := \sup_{\R} \Psi g_2 < \infty	$.
		By \cref{hjb:existence}, $\beta(y) = \tilde{C}_2 g_2(y)$ is a supersolution to \cref{hjb:hjb}.
	
		We look for a subsolution of the form \[
				\alpha(y) = \begin{cases}
						\exp(qy+s), & y \le 0, \\
						e^{s} (1 + (q + \varepsilon) y), & y \ge 0
				\end{cases}
		\] for some $\varepsilon > 0$, $s \in \R$.

		Plugging $\alpha$ into \cref{hjb:hjb}, one sees that $\alpha$ is a subsolution on $(-\infty, 0)$ if \begin{align}
				\label{vasicek:subsolution:condition}
				\tilde{a}(y) + \frac{\eta(y)}{q} - \frac{1}{2} \nu^2 ((1-\rho^2)R + \rho^2)q - \frac{\alpha(y)}{q} \ge 0 \quad \text{ for all } y < 0
				\tag{$*$}
		.\end{align} 
		Notice that $\tilde{a} + \frac{\eta}{q}$ is constant.
		Hence, \eqref{vasicek:subsolution:condition} is equivalent to the condition $C_1 - \frac{\alpha(y)}{q} \ge 0$ for all $y < 0$.
		By monotonicity of $\alpha$, this is true for all $s \le \log(q C_1)$.
		Furthermore, $\alpha$ is a subsolution at $0$ for all $\varepsilon > 0$.
		Let  \[
				V(y, t, s) = -e^{s} (1+ty)^2 + (1+ty) \eta(y) + t \tilde{a}(y) - d \frac{t^2}{1+ty}
		.\]
		Plugging $\alpha$ into \cref{hjb:hjb}, one can check that $\alpha$ is a subsolution on $(0, \infty)$ if $V(y, q+\varepsilon, s) \ge 0$ for all $y > 0$.
		Notice that \[
                \frac{1+qy}{q} V(y, q, -\infty) = q^2 \kappa y^3 + q(\eta(0) + \kappa) y^2 + q\left(\tilde{a}(0) + 2 \frac{\eta(0)}{q}\right) y + \left(C_1 - \frac{1}{2} \nu^2 q\right)
        \] is a polynomial of degree $3$ in $y$, and that $\eta(0) \ge 0$.
        Since we have $\tilde{a}(0) + 2 \frac{\eta(0)}{q} \ge C_1 > 0$ and $C_1 - \frac{1}{2}v^2 q > 0$, all coefficients are positive.
		Hence, $V(y, q, -\infty) > 0$ for all $y > 0$ by Descartes' rule of signs.
		Thus, $V(y, q+\varepsilon, s) > 0$ for all $y > 0$ if $\varepsilon$ and $s$ are small enough.
		In total, $\alpha$ is a global subsolution to \cref{hjb:hjb}.
		Notice that $\alpha \le \beta$ as long as $s$ is small enough.

		By \cref{global_existence:existence}, there exists a candidate solution $u > 0$ with $\alpha \le u \le \beta$ to the HJB equation \eqref{hjb:hjb}.
		As $ \lim_{y \to \infty} \frac{\alpha(y)}{\eta(y)} > 0$ and $\lim_{y\to \infty} \frac{\beta(y)}{\eta(y)} = \tilde{C}_2$, $\frac{u}{\eta}$ is eventually bounded and bounded away from $0$.
		One easily sees that the model coefficients satisfy the assumptions of \cref{hjb:asymptotics,hjb:asymptotics:mean_reverting,hjb:asymptotics:mean_reverting:less_than_Psi_eta,hjb:asymptotics:derivative:mean_reverting}, so $\frac{u}{\eta} \to 1$, $u \le \eta \Psi \eta \le \eta $ eventually, and $\frac{u'}{u} \to 0$ as $y \to \infty$.

		By \cref{vasicek:behaviour_at_neg_inf} below, $\frac{u'}{u}$ is bounded as $y\to -\infty$, so $\frac{u'}{u}$ is bounded over $\R$ by continuity.
		By \cref{verification:sufficient}, this means that condition \eqref{verification:sde} of \cref{verification} is satisfied.
		Now, consider condition \eqref{verification:infinite_consumption}, and denote $C = \sup_{\R} |\frac{u'}{u}|$.
		The $\hat{\mathbb{P}}$-dynamics of $Y$ are given by \[
				dY_t = \left(-\kappa (Y_t - \theta) + \frac{1-R}{R} \rho \lambda \nu - \nu^2 ((1-\rho^2)R + \rho^2) \frac{u'}{u}\right) dt + \nu d \tilde{W}^{\hat{\mathbb{P}}}_t, \quad Y_0 = y
		.\]
		Consider the Ornstein-Uhlenbeck process \[
				d\hat{Y}_t = -\kappa \left(\hat{Y}_t - \left(\theta + \frac{1-R}{R \kappa} \rho \lambda \nu - \frac{\nu^2}{\kappa} ((1-\rho^2) R + \rho^2) C\right)\right) dt + \nu d \tilde{W}^{\hat{\mathbb{P}}}_t, \quad \hat{Y}_0 = y
		.\]
		By the comparison theorem for SDEs (cf.\ \cite[Prop.~5.2.18]{Karatzas1998}), we have $\hat{Y}_t \le Y_t$ for all $t \ge 0$ $\hat{\mathbb{P}}$-a.s.
		Since the subsolution $\alpha$ is increasing, we thus have \[
				\int_{0}^{\infty} u(Y_t) \d t \ge \int_{0}^{\infty} \alpha(Y_t) \d t \ge \int_{0}^{\infty} \alpha(\hat{Y}_t) \d t = \infty \quad \hat{\mathbb{P}}\text{-a.s.}
		\] by ergodicity of the Ornstein-Uhlenbeck process $\hat{Y}$ and $\alpha > 0$ (cf.\ \cite[Ex.~9.12', Thm.~9.6]{Hoepfner2014}).

		Thus, \cref{verification} shows that $u$ is indeed the optimal consumption rate for the control problem \eqref{setting:control_problem}.

		\begin{lemma}
				\label{vasicek:behaviour_at_neg_inf}
				Consider the Vasicek model with parameters as above.
				Let $u$ be a positive solution to \cref{hjb:hjb} on $(-\infty, 0]$ with $qy+s \le \log u \le \tilde{q}y+\tilde{s}$ for all $y \le 0$, where $q = \frac{R-1}{R \kappa}$ and $0 \le \tilde{q} \le q$.
				Then
				\[
						\tilde{q} \le \liminf_{y \to -\infty} \frac{u'}{u} \le \limsup_{y\to -\infty} \frac{u'}{u} \le q
				.\]
		\end{lemma}
		\begin{proof}
                Set $w = \frac{u'}{u}$.
				As $qy+s \le \log u \le \tilde{q} y + \tilde{s}$ and $w = (\log u)'$, we have \[
						\liminf_{y\to -\infty} w \le q, \quad \limsup_{y\to -\infty} w \ge \tilde{q}
				.\] 
				Since $w$ satisfies the ODE \[
					\frac{1}{2}b^2 w' = u - \eta - \tilde{a} w + \left(d - \frac{1}{2}b^2\right) w^2	
				,\] we have $w' < 0$ for $w\in (w_{-}, w_{+})$, where \[
						w_{\pm} = 
						\frac{\tilde{a} \pm \sqrt{\tilde{a}^2 -4 (d - \frac{1}{2}b^2) (u - \eta)} }{2 (d - \frac{1}{2}b^2)} 
						= \frac{u - \eta}{\tilde{a}} \frac{2}{1 \mp \sqrt{1 - 4 (d - \frac{1}{2}b^2) \frac{u-\eta}{\tilde{a}^2}} }
				.\] 
				Notice that $w_{-} \to q$ and $w_+ \to \infty$ as $y\to -\infty$.
				For all $\varepsilon > 0$, there exists some $y_0$ such that $w_{-}(y) < q + \varepsilon$ for all $y \le y_0$.
				Hence, if $w > q + \varepsilon$ for some $y \le y_0$, we have $w > q + \varepsilon$ for all $y \le y_0$.
				As this contradicts $\liminf_{y\to -\infty} w(y) \le q$, we have $w \le q + \varepsilon$ for all $y \le y_0$, and so $\limsup_{y\to -\infty}w \le q$.
				Analogously, it follows that $\liminf_{y \to -\infty} w(y) \ge \tilde{q}$.
		\end{proof}

		\begin{remark}
				Above, we used the supersolution $\beta(y) = \tilde{C}_2 g_2(y) \approx \tilde{C}_2 (1 + \eta(y)^{+})$, which corresponds to $\tilde{q} = 0$.
				Using the same functional form as for the subsolution, one can show that for any $0 < \tilde{q} < q$, there exists a supersolution that provides the corresponding upper bound on $\log u_{\tilde{q}}$.
				Analogously as for the case $\tilde{q} = 0$, the solution $u_{\tilde{q}}$ generated from the supersolution with parameter $0 < \tilde{q} < \frac{R-1}{R\kappa}$ is the optimal consumption rate of the control problem \eqref{setting:control_problem} and parameterises the value function.
				Hence, since the value function is unique, all the solutions $u_{\tilde{q}}$ coincide.
                Note that \cref{hjb:uniqueness} is not applicable here since the $u_{\tilde{q}}$ might a priori have different exponential decay rates as $y \to -\infty$.
				Letting $\tilde{q} \uparrow q$ in \cref{vasicek:behaviour_at_neg_inf} now yields that $\lim_{y \to -\infty} \frac{u'}{u} = q = \frac{R-1}{R\kappa} $.
		\end{remark}

        \subsubsection*{Numerical Approximation}

        We again use the scheme from \cref{section:numerics} with grid size $10^{-4}$.
        The model parameters are found in \cref{vasicek:table:parameters}.
        In \cref{vasicek:fig:apriori}, we see that our approximate solution lies between the tight a priori bounds obtained by \textcite{Guasoni2019}, which shows that the scheme from \cref{section:numerics} converges properly even in incomplete market models that are not strongly myopically well-posed. 

		\begin{figure}[htb]
				\centering

                \includegraphics[width=0.45\textwidth]{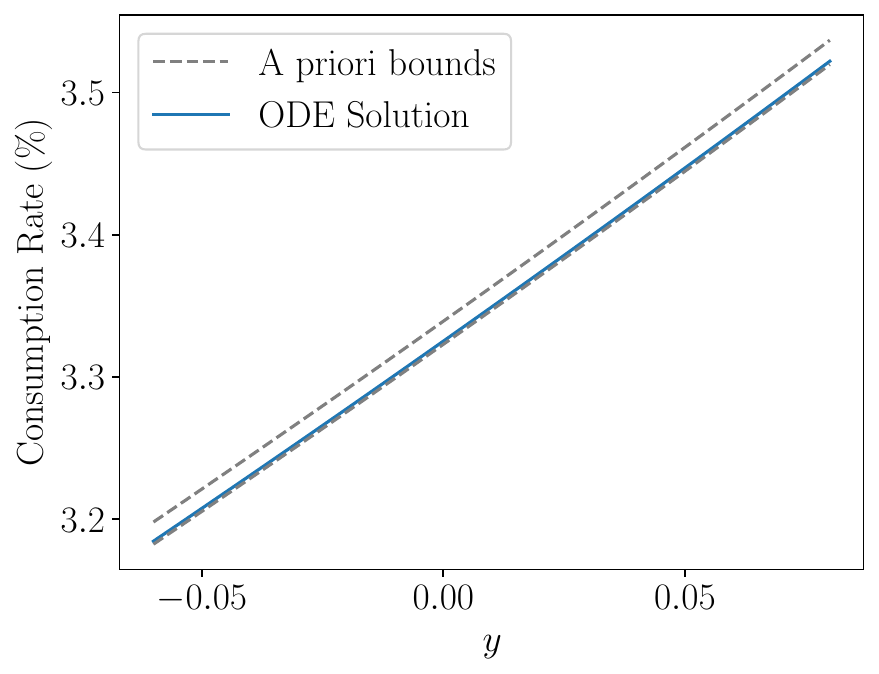}
                \caption{Comparison with a priori bounds (Vasicek).}
                \label{vasicek:fig:apriori}		
				
		\end{figure}

        \begin{table}[htb]
				\centering
				\begin{tabular}{cccccccc}
						\toprule
						$R$   & $\delta$ & $\lambda$ & $\sigma$ & $\kappa$ & $\theta$ & $\nu$ & $\rho$ \\
						\midrule
						$1.5$ & $0.02$ & $\frac{23}{60}$ & $0.18$    & $0.43$    & $0.013$    & $0.033$ & $-0.0012$ \\
						\bottomrule
				\end{tabular}
				\caption{Parameters of the Vasicek model, taken from \cite{Guasoni2019}.}
                
				\label{vasicek:table:parameters}
		\end{table}

        \subsubsection*{Asymptotic Behaviour of Intertemporal Hedging and Consumption}

        In the Vasicek model, the intertemporal hedge reads \[
                \pi_h = - \rho \frac{\nu}{\sigma} \frac{u'}{u}
        .\]

        Above, we showed that $\frac{u'}{u} \to 0$ as $y \to \infty$. 
        Hence, the optimal portfolio is asymptotically myopic as the interest rate goes to infinity: $\pi_h \to 0$ as $y \to \infty$. 
        As the interest rate gets very negative, we showed that $\frac{u'}{u} \to \frac{R-1}{R \kappa}$.
        Hence, the intertemporal hedge satisfies $\pi_h \to - \rho \frac{\nu}{\sigma} \frac{R-1}{R \kappa}$ as $y \to -\infty$.
        This is also what we observe numerically in \cref{vasicek:fig:portfolio_limit}.
        Note that the size of the intertemporal hedge in \cref{vasicek:fig:portfolio_limit} is very small since $\rho \approx 0$.

        The optimal consumption rate satisfies $\hat{\xi} \to 0$ as $y \to -\infty$ and $\xi_h \to -\kappa$ as $y \to \infty$, see \cref{vasicek:fig:consumption_limit}.
        It is important to note that the myopic consumption rate does not define an admissible strategy as $y \to -\infty$ since only non-negative consumption is possible.
        In particular, the intertemporal consumption adjustment $\xi_h$ satisfies $\xi_h \to \infty$ as $y \to -\infty$.
        In total, problem \eqref{setting:control_problem} is well-posed in the Vasicek model despite being myopically ill-posed when the interest rate is lower than the threshold $y^* = \frac{\delta}{1-R} - \frac{\lambda^2}{2R} \approx -0.088$.
      
        \begin{figure}[htb]
				\centering

				\begin{subfigure}[t]{.45\textwidth}
						\includegraphics[width=0.95\textwidth]{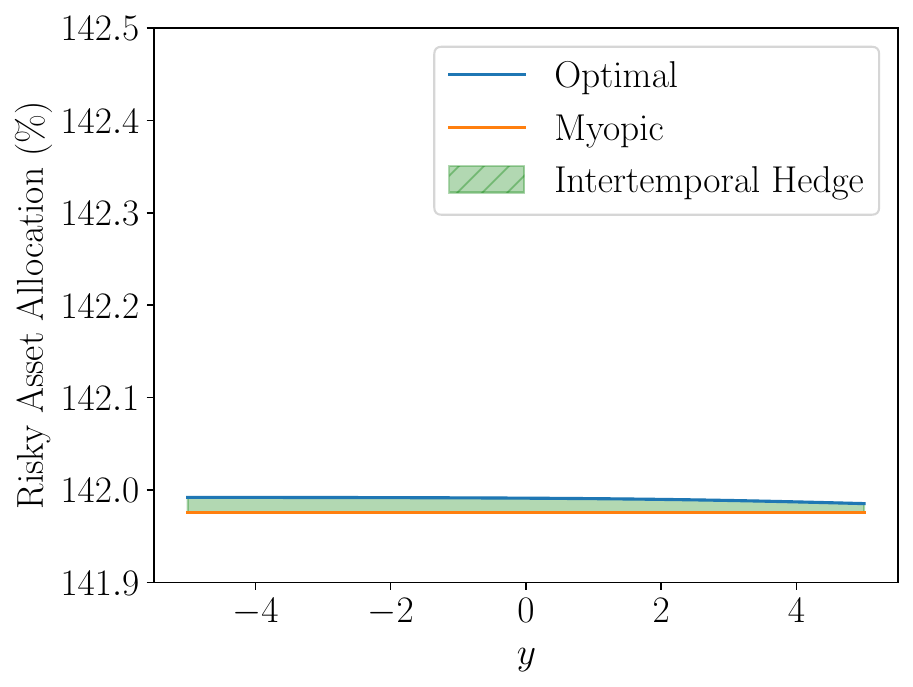}
						\caption{Risky asset allocation.}
						\label{vasicek:fig:portfolio_limit}		
				\end{subfigure}
                \hfill
                \begin{subfigure}[t]{.45\textwidth}
						\includegraphics[width=0.95\textwidth]{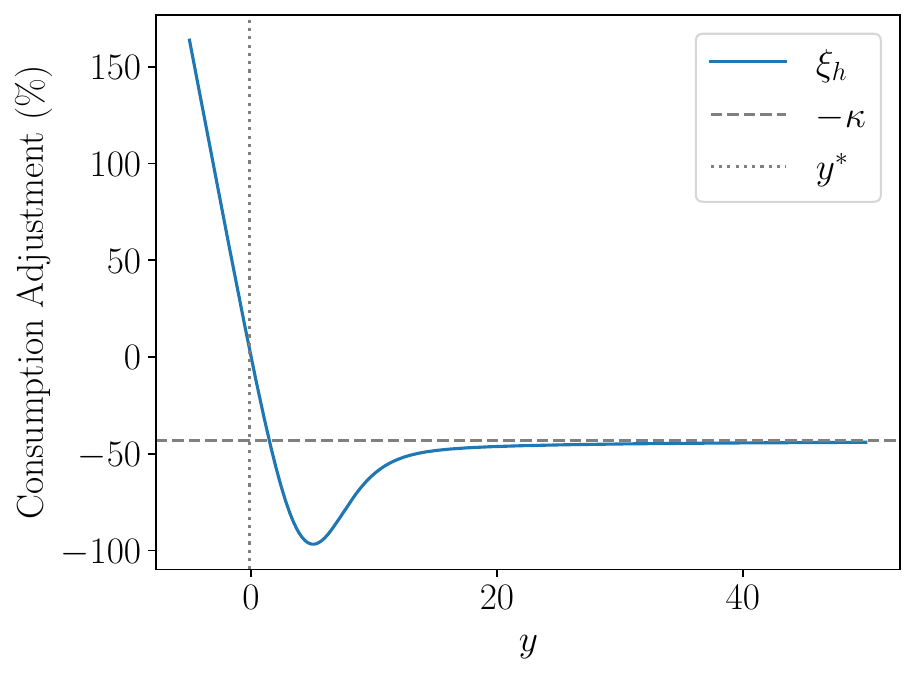}
						\caption{Consumption rate.}
						\label{vasicek:fig:consumption_limit}		
				\end{subfigure}

                \caption{Asymptotic regime (Vasicek).}
		\end{figure}
		
		\medskip{}
        
		\subsubsection*{Economic Interpretation}

		What makes the Vasicek model particularly interesting is that unlike in the models we consider above, the optimal consumption policy strongly deviates from the myopic policy.
        As interest rates become heavily negative, the optimal consumption rate converges to zero.
        Note that the myopic consumption rate is only well defined for $y \geq y^*$; for $y<y^*$ the frozen problem is ill-posed. 
        Indeed, $\eta$ converges to $-\infty$ as $y \to -\infty$, which drives the explosion of the intertemporal consumption adjustment $\xi_h = u - \eta$ (recall that $\eta(y) = \frac{1}{R} (\delta - (1-R) (y + \frac{\lambda^2}{2R}))$, which remains well-defined for $y < y^*$, despite the myopic consumption rate being ill-defined).
        
        \Cref{vasicek:fig:portfolio} shows that the optimal and myopic portfolios as well as the intertemporal hedge are all almost constant over the region of interest.
        The intertemporal hedge is positive, but negligible due to the correlation being almost zero.
        As in the other models, \cref{vasicek:fig:consumption} shows that the myopic consumption rates varies much more than the optimal one as the interest rate varies; in comparison to the myopic consumption rate, the optimal consumption rate is quite flat over the realistic region of the interest rate.
        
        \begin{figure}[htb]
				\centering

                \begin{subfigure}[t]{.45\textwidth}
						\includegraphics[width=0.95\textwidth]{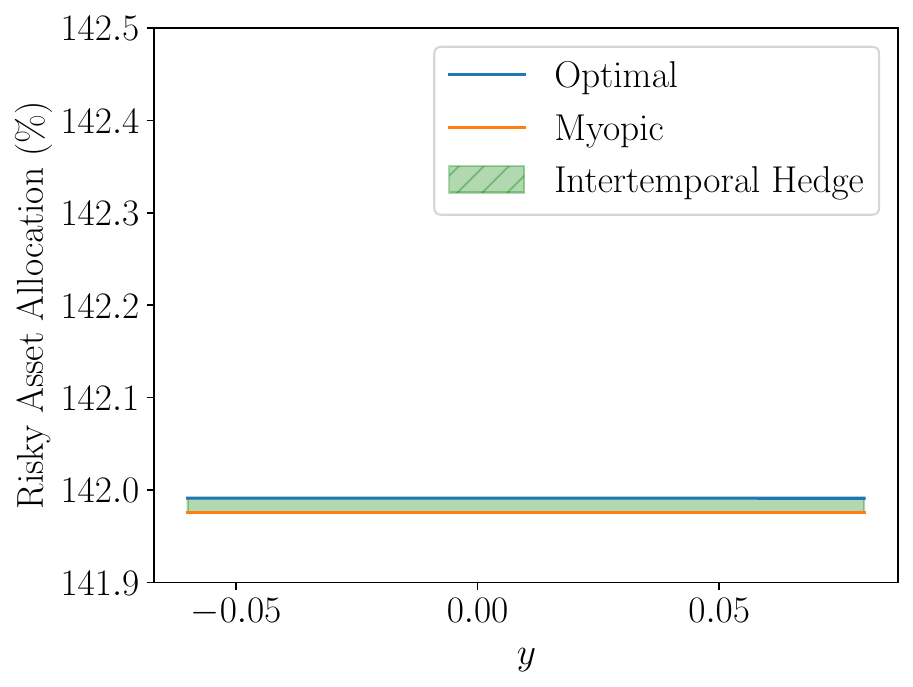}
						\caption{Risky asset allocation.}
						\label{vasicek:fig:portfolio}		
				\end{subfigure}
                \hfill
                \begin{subfigure}[t]{.45\textwidth}
						\includegraphics[width=0.95\textwidth]{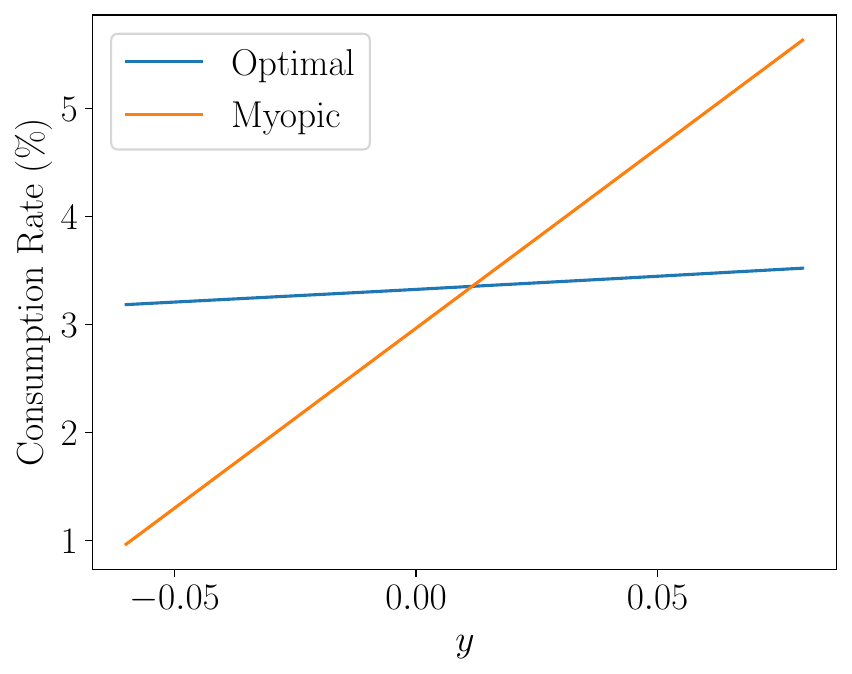}
						\caption{Consumption rate.}
						\label{vasicek:fig:consumption}		
				\end{subfigure}

                \caption{Economically relevant regime (Vasicek).}
		\end{figure}

\begin{appendix}

		\section{Solvability of the HJB Equation over a Bounded Domain}
		\label{section:bounded_domain}

		In this section we discuss the solvability of the Neumann and Dirichlet problems for the HJB equation \eqref{hjb:hjb:constant_correlation} over a bounded domain.
		The Neumann problem corresponds to the stochastic factor being a reflected diffusion process.
		The results obtained here for the Neumann problem are analogous to the case of a Markov chain with finite state space; a solution exists if and only if the real parts of all eigenvalues of the differential operator are positive.

		Let $\Omega \subset \R^{d}$ be a bounded domain that is $C^{2,\alpha}$, and $L = - \sum_{i, j} b_{ij} D_i D_j + \sum_i a_i D_i + \eta$ a uniformly elliptic differential operator with coefficients belonging to $C^{\alpha}( \bar{\Omega})$.
		Notice that we use the opposite sign convention for the elliptic operator than \cite{Du2006}, which is our main reference on properties of elliptic operators in this section.
		We aim to solve the equation \[
				Lu = u^{p}
		\] under homogeneous Neumann ($ \frac{\partial u}{\partial \nu} = 0$ on $\partial \Omega$) for $p < 1$ and Dirichlet ($u = 0$ on $\partial \Omega$) boundary conditions for $p \in [0, 1)$.

		We begin by considering the Neumann problem.

		\begin{theorem}
				\label{bounded_domain:criterion:neumann}
				The equation $Lu = u^{p}$, $p < 1$, under Neumann boundary conditions has a positive solution if and only if the principal eigenvalue of $L$ under Neumann boundary conditions is positive.
		\end{theorem}
		\begin{proof}
				First, note that the principal eigenvalue exists and is real-valued by \cite[Thm.~2.1]{Du2006}.
				Denote the principal eigenvalue by $\lambda$, and the corresponding positive eigenfunction by $v$.
				Note that $v \in C^{2,\alpha}( \bar{\Omega})$ by \cite[Thm.~A.4]{Du2006}.

				Assume that $u$ is a solution to $Lu = u^{p}$ under Neumann boundary conditions.
				By \cite[Thm.~2.4]{Du2006}, we have $\lambda > 0$ if and only if there exists a positive (on the interior) strict supersolution to $Lf = 0$ under Neumann boundary conditions.
                Since $u$ satisfies the Neumann boundary condition and $Lu = u^p > 0$, $u$ is a positive strict supersolution, and hence $\lambda > 0$.
                
				Conversely, assume that the principal eigenvalue of $L$ is positive.
				Assume for sake of contradiction that $\min_{\partial \Omega} v \le 0$. 
				Then $\frac{\partial v}{\partial \nu} \neq 0$ by Hopf's lemma (applied to $L + \eta^{-}$), which contradicts the Neumann boundary conditions. 
				Hence, $v > 0$ on $\partial \Omega$, and thus also over $ \bar{\Omega}$.
				Set $m = \min_{ \bar{\Omega}} v$, $M = \max_{ \bar{\Omega}} v$, and \[
						\alpha = \lambda^{-\frac{1}{1-p}} \frac{v}{M}, \quad \beta = \lambda^{-\frac{1}{1-p}} \frac{v}{m}
				.\] 
				Clearly, we have $\alpha \le \beta $, and $\alpha$ and $\beta $ both satisfy the Neumann boundary conditions.

				Furthermore, \[
						L \alpha = \lambda^{-\frac{1}{1-p}} \frac{1}{M} Lv = \lambda^{-\frac{1}{1-p}} \frac{1}{M} \lambda v = \lambda^{-\frac{p}{1-p}} \frac{v}{M} \le \lambda^{-\frac{p}{1-p}} \left( \frac{v}{M} \right)^{p} = \alpha^{p}
				\] since $\frac{v}{M} \le 1$, so $\alpha$ is a subsolution.
				Analogously, it follows that $ \beta $ is a supersolution.

				Notice that the map $x \mapsto x^{p}$ is Lipschitz over $[\lambda^{-\frac{1}{1-p}} \frac{m}{M}, \lambda^{-\frac{1}{1-p}} \frac{M}{m}]$.
				Hence, by \cite[Thm.~4.3]{Du2006}, there exists a solution $u$ to $Lu = u^{p}$ with $0 < \alpha \le u \le \beta$.
		\end{proof}

		Next, we consider the case of Dirichlet boundary conditions.
		Here, we have to assume that $d = 1$ since the lack of regularity of $x \mapsto x^{p}$ at $0$ prevents us from applying the sub- and supersolution theory for elliptic operators.
		Instead, we use the open domain theory from \cref{section:global_existence}.

		\begin{theorem}
				\label{bounded_domain:criterion:dirichlet}
				Assume that $d = 1$. 
				The equation $Lu = u^{p}$, $p \in [0, 1)$, under Dirichlet boundary conditions has a positive solution if and only if the principal eigenvalue of $L$ under Dirichlet boundary conditions is positive.
		\end{theorem}
		\begin{proof}
                The positivity of the principal eigenvalue being a necessary condition follows analogously as in \cref{bounded_domain:criterion:neumann}.

				Assume now that the principal eigenvalue $\lambda$ of $L$ under Dirichlet boundary conditions is positive, and let $v$ be a corresponding positive eigenfunction.

                For $p = 0$, the existence of a solution follows from \cite[Thm.~A.1,A.5]{Du2006}.
                Note that the solution is positive by the strong maximum principle.
                
				Now, assume that $p \in (0, 1)$.
                Set $M_1 = \max_{ \bar{\Omega}} v$.
				For a fixed $\varepsilon \in (0, 1)$, set \[
						\alpha = \varepsilon \lambda^{-\frac{1}{1-p}} \frac{v}{M_1}
				.\] 
				It follows analogously as in \cref{bounded_domain:criterion:neumann} that $\alpha$ is a subsolution. 
				
				Next, we construct a supersolution.
				Let $w$ be a solution to $Lw = 1$ with Dirichlet boundary conditions. 
				Note that $w$ exists by \cite[Thms.~A.1, A.5]{Du2006} since the principal eigenvalue of $L$ is positive.
				By \cite[Thm.~2.4]{Du2006}, the strong maximum principle holds for $L$, and so $w > 0$ in $\Omega$.

				Set $M_2 = \max_{ \bar{\Omega}} w$ and \[
						\beta = \frac{1}{\varepsilon} M_2^{\frac{p}{1-p}} w	
				.\] 
				Then $ \beta $ is a supersolution since \[
						L \beta = \frac{1}{\varepsilon} M_2^{\frac{p}{1-p}} Lw = \frac{1}{\varepsilon} M_2^{\frac{p}{1-p}} = \frac{1}{\varepsilon} M_2^{\frac{p^2}{1-p}} M_2^{p} \ge \left( \frac{1}{\varepsilon} \right)^{p} M_2^{\frac{p^2}{1-p}} w^{p} = \beta^{p}
				.\] 

				For $\varepsilon$ small enough, we have $\alpha \le \beta $:
				We have \[
						\frac{\alpha}{\beta } = \varepsilon^2 \lambda^{-\frac{1}{1-p}} M_1^{-1} M_2^{- \frac{p}{1-p}} \frac{v}{w}
				.\] 
				For $\varepsilon$ small enough, one obtains that $ \frac{\alpha}{\beta } \le 1$ as long as $\frac{v}{w}$ is bounded.
				Since $v > 0$ and $w > 0$ in $\Omega$, $\frac{v}{w}$ can only blow up at the boundary $\partial \Omega$.
				But by Hopf's lemma (applied to $L + \eta^{-}$), $\frac{\partial v}{\partial \nu} \neq 0$ and $\frac{\partial w}{\partial \nu} \neq 0$ on $\partial \Omega$, so $\frac{v}{w}$ stays bounded around $\partial \Omega$ by L'Hospital's rule.

				Now, a solution $u$ to $Lu = u^{p}$ with $\alpha \le u \le \beta $ on $\Omega$ exists by \cref{global_existence:existence}.
				Since $\alpha = 0$ and $\beta = 0$ on $\partial \Omega$, we can continuously extend $u$ to $u = 0$ on $\partial \Omega$.
		\end{proof}

		Finally, notice that the condition $\eta > 0$ implying existence of a solution in a finite regime setting (see \cref{finite_state_space:well-posedness:easy_criteria}) remains true on the bounded domain.

		\begin{lemma}
				\label{bounded_domain:strongly_wellposed}
				Assume that $\eta > 0$ in $ \Omega$. 
				Then the principal eigenvalue of $L$ is positive under both Dirichlet and Neumann boundary conditions.
		\end{lemma}
		\begin{proof}
				Denote by $1$ the constant function $1$.
				We have $L 1 = \eta > 0$ in $\Omega$, and $1$ satisfies the boundary condition (for Dirichlet boundary conditions, note that \cite[Thm.~2.4]{Du2006} only requires the boundary condition to be satisfied with an inequality). 
				Hence, the principal eigenvalue of $L$ is positive by \cite[Thm.~2.4]{Du2006}.
		\end{proof}

        \section{Auxiliary Results}
        \label{section:auxiliary}

		For completeness, we state Itô's formula for Markov-modulated diffusion processes.

		\begin{lemma}
				\label{auxiliary:ito_markov_modulated}
				Let $Y$ be a continuous-time Markov chain with state space $\{1, \ldots, N\} $ and Q-matrix $Q$, and let $X$ be a one-dimensional Itô-diffusion with dynamics  \[
						dX_t = \mu(X_t, Y_t) dt + \sigma(X_t, Y_t) dW_t 
				,\] 
				where $W$ is a Brownian motion independent of $Y$.
				For any function $V = V(t, x, y)$ that is continuously differentiable w.r.t.\ $t$ and twice continuously differentiable w.r.t.\ $x$, we have \[
						dV(t, X_t, Y_t) = \left( \frac{\partial V}{\partial t} + \mu \frac{\partial V}{\partial x} + \frac{1}{2} \sigma^2 \frac{\partial ^2 V}{\partial x^2} + QV \right) dt + \sigma \frac{\partial V}{\partial x} dW_t + dM_t
				,\] where $M$ is a local martingale with $M_0 = 0$.
                Moreover, if $\E\left[ \int_{0}^{t} V(s, X_s, j)^2 \d s \right] < \infty$ for all $t \ge 0$ and $j = 1,\dots,N$, then $M$ is a true martingale.
		\end{lemma}
		\begin{proof}
                We decompose \[
                        V(t, X_t, Y_t) = \sum_{j=1}^{N} V(t, X_t, j) \delta^{j}_t
                ,\] where $\delta^{j}_t = I_{\{Y_t = j\} }$.
                By Eq.~(5) in \cite{Bjoerk1980}, $M_t^j := \delta_t^j - \delta_0^j - \int_0^t \sum_{i=1}^{N} Q_{i,j} \delta^i_s \d s$ is a square-integrable martingale.
				The dynamics of $V$ and the (local) martingale property now follow from Eq.~(9) in \cite{Bjoerk1980} with $M_t = \sum_{j=1}^{N} \int_{0}^{t} V(s, X_s, j) dM^j_s$.
		\end{proof}
		
		Next, we collect some auxiliary results from real analysis that are used throughout the proofs of the main results.

        \begin{lemma}
                \label{auxiliary:non_eventually_monotone}
                Let $f \in C([y_0, \infty))$ and assume that $f$ is not eventually monotone.
                Then there exist sequences $(y^{1}_n)_{n\in \N}, (y^2_n)_{n\in \N}$ of local minima and maxima, respectively, with $y^{i}_n \to \infty$, $i=1,2$, s.t.\ $f(y^1_n) \to \liminf_{y \to \infty} f(y)$ and $f(y^2_n) \to \limsup_{y \to \infty} f(y)$.
        \end{lemma}
        \begin{proof}
                Since $f$ is continuous and not eventually monotone, the sets $A := \{y \geq y_0: y \text{ is a local minimum}\}$ and $B :=\{y: y \text{ is a local maximum}\}$ are unbounded.

                We just establish the claim for the $\liminf$. For each $x \geq y_0$ with $x \in A$, the definition of a local minimum implies
            \begin{equation*}
\inf_{y \geq x} f(y) = \inf_{y \geq x: y \in A} f(y).
            \end{equation*}
            Thus, by the definition of the $\liminf$ and elementary analysis,  there exists an increasing sequence $(\tilde{y}_n)_{n \in \N} \subset A$ converging to $\infty$ with $\lim_{n \to \infty} f(\tilde{y}_n) = \liminf_{y \to \infty} f(y)$.
        \end{proof}

		\begin{lemma}
				\label{auxiliary:monotone_convergence}
				Let $f \in C^2([y_0, \infty))$, and assume that $f$ converges as $y\to \infty$.
				Then there exist sequences $(y^{1}_n)_{n\in \N}, (y^2_n)_{n\in \N}$ with $y^{i}_n \to \infty$, $i=1,2$, s.t. \begin{alignat*}{3}
						f'(y^{1}_n) &\to \liminf_{y \to \infty} f'(y), \quad& f''(y^{1}_n) \to 0 \qquad& \text{ as } n \to \infty, \\
						f'(y^{2}_n) &\to \limsup_{y \to \infty} f'(y), \quad& f''(y^{2}_n) \to 0 \qquad& \text{ as } n \to \infty
				.\end{alignat*}
				In particular, if $f$ converges monotonically, there exists a sequence $(y_n)_{n\in \N}$ with $y_n \to \infty$ and $f'(y_n) \to 0, f''(y_n) \to 0$ as $n\to \infty$.
		\end{lemma}
		\begin{proof}
				First, consider the case that $f'$ eventually becomes monotone. Then it converges to $0$ since $f$ is convergent.
				Hence, $f'(y_n) \to 0$ for any sequence $(y_n)_{n\in \N}$ with $y_n \to \infty$.
				Assume without loss of generality that $f'$ is increasing.
				Then $\liminf_{y\to \infty} f''(y) = 0$,  which yields a suitable sequence.

                Otherwise, we can choose approximating sequences $y^{i}$ for $\limsup_{y\to \infty} f'(y)$ and $\liminf_{y\to \infty} f'(y)$ that consist only of local maxima and minima of $f'$, respectively, by \cref{auxiliary:non_eventually_monotone}.
				At these extrema, $f''(y^{i}_n) = 0$.

				The additional claim follows from the above since $\liminf_{y\to \infty} f'(y) = 0$ if $f$ converges and is increasing, and $\limsup_{y\to \infty} f'(y) = 0$ if $f$ converges and is decreasing.
		\end{proof}

        \begin{lemma}
                \label{auxiliary:trap}
                Let $w \in C^2([y_0, \infty))$.
                Assume that there exists a concave function $h \in C^1([y_0, \infty))$ s.t.\ 
                \begin{equation}
                    \label{auxiliary:trap:condition}
                    w''(y) > 0 \text{ for all } y \text{ with } w(y) > h(y) \text{ and } w'(y) \ge h'(y)
                .\end{equation}
                If $w(y_1) > h(y_1)$ and $w'(y_1) \ge h'(y_1)$ for some $y_1 \ge y_0$, then $w(y) > h(y)$ and $w'(y) \ge h'(y)$ for all $y \ge y_1$.
        \end{lemma}
        \begin{proof}
                First, we show an auxiliary statement that if $y_2 \geq y_0$ is such that $w(y_2) > h(y_2)$ and $w'(y_2) \ge h'(y_2) $, then there is $\varepsilon  > 0$ with $w > h$ and $w' \ge h'$ on $[y_2, y_2 + \varepsilon]$.
                By \eqref{auxiliary:trap:condition}, we have $w''(y_2) > 0$, and so $w'' > 0$ on $[y_2, y_2+\varepsilon]$ for some $\varepsilon > 0$ by continuity. 
                Hence, $w$ is convex over $[y_2, y_2+\varepsilon]$.
                Since $h$ is concave, $w - h$ is also convex over $[y_2, y_2+\varepsilon]$.
                Thus, we have $w'(y) - h'(y) \ge w'(y_2) - h'(y_2) \ge 0$ for all $y \in [y_2, y_2 + \varepsilon]$, and therefore also $w(y) - h(y) \ge w(y_2) - h(y_2) > 0$.
                Together, we have $w > h$ and $w' \ge h'$ over $[y_2, y_2 + \varepsilon]$.

                Now, set $y_2 = \inf \{ y \ge y_1: w(y) \le h(y) \text{ or } w'(y) < h'(y) \}$.
                Suppose for sake of contradiction that $y_2 < \infty$.
                By the auxiliary statement, we know that $y_2 > y_1$ and that $w-h$ is increasing and convex over $[y_1, y_2)$.
                But this means that by continuity $w(y_2) - h(y_2) \ge w(y_1) - h(y_1) > 0$ and $w'(y_2) - h'(y_2) \ge w'(y_1) - h'(y_1) \ge 0$.
                Hence, we can again apply the auxiliary statement at $y_2$ to obtain an interval $[y_2, y_2 + \varepsilon]$ over which $w > h$ and $w' \ge h'$, which contradicts the definition of $y_2$.
                Hence, we have $y_2 = \infty$, i.e.\ $w > h$ and $w' \ge h' $ holds over $[y_1, \infty)$.
        \end{proof}

		\begin{lemma}
            \label{auxiliary:crossing}
            Let $w \in C^2([y_0, \infty))$ with $w \to 1$.
            Assume that there exists an increasing concave function $h \in C^1([y_0, \infty))$ s.t.\ 
            \begin{equation}
                \label{auxiliary:crossing:condition}
                w''(y) > 0 \text{ for all } y \text{ with } w(y) > h(y) \text{ and } w'(y) \ge 0
            .\end{equation}
            Then we have either $w > h$ or $w \le h$ eventually.
            Moreover, if $w \le 1$ eventually, we have $w \le h$ eventually. 
        \end{lemma}
        \begin{proof}      
            First, we will show the auxiliary statement that no point with $w > h$ and $w' \ge h'$ can exist.
            Assume for sake of contradiction that such a point $y_1$ exists.
            Notice that \eqref{auxiliary:crossing:condition} implies \eqref{auxiliary:trap:condition} since $h$ is increasing, so that $w > h$ and $w' \ge h'$ over $[y_1, \infty)$ by \cref{auxiliary:trap}.
            By \eqref{auxiliary:crossing:condition}, this means that $w$ is strictly convex and increasing over $[y_1, \infty)$, which contradicts the convergence of $w$ to $1$.
            Hence, there exists no point $y_1$ with $w(y_1) > h(y_1)$ and $w'(y_1) \ge h'(y_1)$.

            Next, we will show that $w$ can't cross $h$ from below, which implies that either $w > h$ or $w \le h$ eventually.
			 Suppose for sake of contradiction that there exist $y_1 < y_3$ with $w(y_1) \le h(y_1)$ and $w(y_3) > h(y_3)$, and set $y_2 = \sup \{ y \in [y_1, y_3]: w(y) = h(y) \}$.
            Then $w > h$ over $(y_2, y_3]$, and by the mean-value theorem there exists some $y \in (y_2, y_3)$ with $w'(y) - h'(y) > 0$, and we arrive at a contradiction to the auxiliary statement.

			Finally, assume that $w \le 1$ eventually.
            Suppose for sake of contradiction that $w > h$ eventually.
            Since $w \to 1$ and $w \le 1$ eventually, the set $\{y: w'(y) \ge 0\}$ is unbounded. 
            Hence, there exists a point $y_1$ with $w > h$ over $[y_1, \infty)$ and $w'(y_1) \ge 0$.
            Since $w > h$ over $[y_1, \infty)$, \eqref{auxiliary:crossing:condition} yields that $w'' > 0$ whenever $w' \ge 0$.
            Hence, \cref{auxiliary:trap} with $\tilde{h} \equiv h(y_1)$ yields that $w' \ge 0$ over $[y_1, \infty)$. 
            Hence, $w$ is increasing and strictly convex over $[y_1, \infty)$ by \eqref{auxiliary:crossing:condition}, which contradicts $w \to 1$.
        \end{proof}

		\begin{lemma}
				\label{auxiliary:convex:derivative}
				Let $f,g \in C^{1}([y_0, \infty))$ be positive.
				Assume that $f$ is increasing and convex, and that $\frac{f}{g}\to 1$ and $\frac{g'}{g} \to 0$.
				Then $\frac{f'}{f} \to 0$ as $y\to \infty$.
		\end{lemma}
		\begin{proof}
				Since $f$ is convex and increasing, we have $0 \le f'(x) \le f(x+1) - f(x) $, so it is enough to show that $\frac{f(x+1)}{f(x)} \to 1$.
                By the mean value theorem, \[
						\log\left( \frac{g(x+1)}{g(x)} \right) = (\log(g(\cdot ))'(\xi_x) = \frac{g'(\xi_x)}{g(\xi_x)}
				\] for some $\xi_x \in [x, x+1]$.
				Since $\frac{g'}{g} \to 0$, this yields $\frac{g(x+1)}{g(x)} \to 1$.
                Together with $\frac{f}{g} \to 1$, we get \[
						\frac{f(x+1)}{f(x)} = \frac{f(x+1)}{g(x+1)} \frac{g(x+1)}{g(x)} \frac{g(x)}{f(x)} \longrightarrow 1
				. \qedhere \] 
		\end{proof}

\end{appendix}

		\printbibliography

\end{document}